\documentclass[a4paper,11pt]{article}

\usepackage{amsmath}
\usepackage{amsfonts}
\usepackage{amssymb}
\usepackage{amsthm}

\usepackage[english]{babel}
\usepackage{fontenc}
\usepackage{graphicx}
\usepackage{mathrsfs}
\usepackage{bbm}
\usepackage{float}
\usepackage{enumerate}
\usepackage{xcolor}
\usepackage{algorithm}
\usepackage{algorithmic}

\usepackage{siunitx}

\usepackage[section]{placeins}

\usepackage{booktabs}
\usepackage{authblk}

\usepackage{chapterbib}

\usepackage{fullpage}

\newcommand\resumesections{
  \setcounter{section}{0}
\def\thesection{S\arabic{section}}
}

\usepackage{natbib}

\theoremstyle{plain}
\newtheorem{thm}{Theorem}

\newtheorem{lem}[thm]{Lemma}
\newtheorem{prop}[thm]{Proposition}

\theoremstyle{definition}

\theoremstyle{remark}

\newcommand{\ind}{\mathbbm{1}}

\newcommand{\R}{\mathbb{R}}
\newcommand{\E}{\mathbb{E}}

\newcommand{\ceil}[1]{\left\lceil #1 \right\rceil}

\newcommand{\iid}{\stackrel{\mathrm{i.i.d.}}{\sim}}

\DeclareMathOperator*{\argmin}{arg\,min}
\DeclareMathOperator*{\sargmin}{sarg\,min}

\newcommand{\mb}{\mathbf}
\newcommand{\mbb}{\boldsymbol}

\title{Modelling High-Dimensional Categorical Data Using Nonconvex Fusion Penalties} 

\author[1]{Benjamin G.\ Stokell\thanks{Supported by the Cantab Capital Institute for the Mathematics of Information.}}
\author[1]{Rajen D.\ Shah\thanks{Supported by an EPSRC Programme Grant.}}
\author[2]{Ryan J.\ Tibshirani}
\affil[1]{University of Cambridge}
\affil[2]{Carnegie Mellon University}

\begin{document}

\maketitle

\begin{abstract}
\noindent

We propose a method for estimation in high-dimensional linear models with nominal categorical data. Our estimator, called SCOPE, fuses levels together by making their corresponding coefficients exactly equal. This is achieved using the minimax concave penalty on differences between the order statistics of the coefficients for a categorical variable, thereby clustering the coefficients. 
We provide an algorithm for exact and efficient computation of the global minimum of the resulting nonconvex objective in the case with a single variable with potentially many levels, and use this within a block coordinate descent procedure in the multivariate case. We show that an oracle least squares solution that exploits the unknown level fusions is a limit point of the coordinate descent with high probability, provided the true levels have a certain minimum separation; these conditions are known to be minimal in the univariate case. We demonstrate the favourable performance of SCOPE across a range of real and simulated datasets. An R package \texttt{CatReg} implementing SCOPE for linear models and also a version for logistic regression is available on CRAN.

\end{abstract}

\section{Introduction} \label{sec:intro}

Categorical data arise in a number of application areas. For example, electronic health data typically contain records of diagnoses received by patients coded within controlled vocabularies and also prescriptions, both of which give rise to categorical variables with large numbers of levels \citep{jensen2012mining}. Vehicle insurance claim data also contain a large number of categorical variables detailing properties of the vehicles and parties involved \citep{Hu2018}. When performing regression with such data as covariates, it is often helpful, both for improved predictive performance and interpretation of the fit, to fuse the levels of several categories together in the sense that the estimated coefficients corresponding to these levels have exactly the same value.

To fix ideas, consider the following ANOVA model relating response vector $Y=(Y_1,\ldots,Y_n)^T \in \R^n$ to categorical predictors $X_{ij} \in \{1, \ldots K_j\}$, $j=1,\ldots,p$: 
\begin{equation} \label{eq:ANOVA_mod}
Y_i = \mu^0 + \sum_{j=1}^p \sum_{k=1}^{K_j} \theta_{jk}^0\ind_{\{X_{ij}=k\}} + \varepsilon_i.
\end{equation}
Here the $\varepsilon_i$ are independent zero mean random errors, $\mu^0$ is a global intercept and $\theta_{jk}^0$
is the contribution
to the response of the $k$th level of the $j$th predictor; we will later place restrictions on the parameters to ensure they are identifiable.
We are interested in the setting where the coefficients corresponding to any given predictor are clustered, so defining
\begin{equation} \label{eq:s_j_def}
s_j := |\{\theta^0_{j1},\ldots,\theta^0_{jK_j}\}|,
\end{equation}
we have $s_j \ll K_j$, at least when $K_j$ is large.
Note that our setup can include high-dimensional settings where $p$ is large and many of the predictors do not contribute at all to the response: when $s_j=1$, the contribution of the $j$th predictor is effectively null as it may be absorbed by the intercept term.

\subsection{Background and motivation} \label{sec:related}
Early work on collapsing levels together in low-dimensional models of the form \eqref{eq:ANOVA_mod} focused on performing a variety of significance tests for whether certain sets of parameters were equal \citep{tukey1949comparing, scott1974cluster, calinski1985clustering}.
A more modern and algorithmic method based on these ideas is
Delete or merge regressors \citep{maj2015delete}, which involves agglomerative clustering based on $t$-statistics for differences between levels.

The CART algorithm \citep{breiman1984classification} for building decision trees effectively starts with all levels of the variables fused together and greedily selects which levels to split. One potential drawback of these greedy approaches is that in high-dimensional settings where the search space is very large, they may fail to find good groupings of the levels. 
The popular random forest procedure \citep{breiman2001random} uses randomisation to alleviate the issues with the greedy nature of the splits, but sacrifices interpretability of the fitted model.

An alternative to greedy approaches in high-dimensional settings is using penalty-based methods such as the Lasso \citep{tibshirani1996regression}. This can be applied to continuous or binary data and involves optimising an objective for which global minimisation is computationally tractable, thereby avoiding some of the pitfalls of greedy optimisation. In contrast to random forest, the fitted models are sparse and interpretable. 
Inspired by the success of the Lasso and related methods for high-dimensional regression, a variety of approaches have proposed estimating $\mbb\theta^0 = (\theta^0_{jk})_{j=1,\ldots,p,\: k=1,\ldots,K_j}$ and $\mu_0$ via optimising over $(\mu, \mbb\theta)$ a sum of a least squares criterion
\begin{equation} \label{eq:least_squares}
\ell(\mu, \mbb\theta) := \frac{1}{2n}\sum_{i=1}^n \bigg(Y_i - \mu - \sum_{j=1}^p \sum_{k=1}^{K_j} \theta_{jk}\ind_{\{X_{ij}=k\}}\bigg)^2
\end{equation}
and a penalty of the form
\begin{equation} \label{eq:all_pairs}
\sum_{j=1}^p \sum_{k=2}^{K_j} \sum_{l=1}^{k-1} w_{j,kl} |\theta_{jk} - \theta_{jl}| \,.
\end{equation}
This is the CAS-ANOVA penalty of \citet{bondell2009simultaneous}.
The weights $w_{j,kl}$ can be chosen to balance the effects of having certain levels of categories more prevalent than others in the data. The penalty is an `all-pairs' version of the fused Lasso 
and closely related to so-called convex clustering \citep{hocking2011clusterpath, chiquet2017fast}. 
We note that there are several other approaches besides using penalty functions. For instance, \citet{pauger2019bayesian} proposes a Bayesian modelling procedure using sparsity-inducing prior distributions to encourage fusion of levels. See also \citet{tutz2016regularized} and references therein for a review of other methods including those based on mixture models and kernels.

The fact that the optimisation problem resulting from \eqref{eq:all_pairs} is convex makes the procedure attractive. However, a drawback is that it may not give a desirable form of shrinkage. Indeed, consider the case where $p=1$, and dropping subscripts for simplicity, all $w_{kl}=1$. This would typically be the case if all levels were equally prevalent. Further suppose for simplicity that the number of levels $K$ is even.
Then if the coefficients are clustered into two groups where one contains only a single isolated coefficient, the number of non-zero summands in \eqref{eq:all_pairs} is only $K-1$. This almost doubles to $2(K-2)$ when one of the two groups is of size $2$. The extreme case where the two groups are of equal size yields 
 $(K/2)^2$ non-zero summands. This particular property of all-pairs penalties, which results in them favouring groups of unequal sizes, is illustrated schematically in Figure~\ref{fig:diag}.
\begin{figure}[h]
	\centering
	\includegraphics[width = 0.8\textwidth]{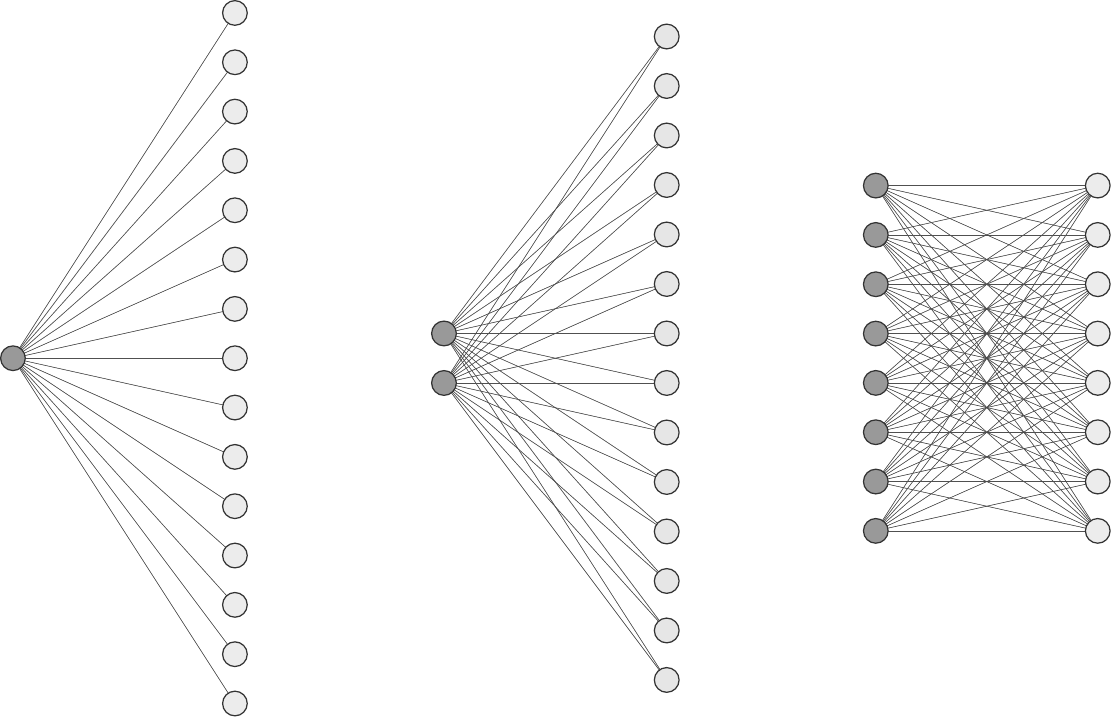}	
	\caption{Illustration of the number of non-zero summands in \eqref{eq:all_pairs} when $p=1$, $K=16$ and coefficients are clustered into two groups of equal size (right), and where one contains a single coefficient (left) and two coefficients (middle).} \label{fig:diag}
\end{figure}
We can see the impact of this in the following concrete example.
Suppose $K=20$ levels are clustered into four groups with
\begin{gather*}
	\theta^0_{1}=\cdots=\theta^0_{4}=-6, \qquad \theta^0_{5}= \cdots =\theta^0_{10}=-2.5 \\
	\theta^0_{11}= \cdots =\theta^0_{16}=2.5, \qquad \theta^0_{17}=\cdots=\theta^0_{20}=6.
\end{gather*}
If the coefficient estimates satisfy $\hat{\theta}_1 = \cdots=\hat{\theta}_4 < \hat{\theta}_5 = \cdots = \hat{\theta}_{10} \leq \hat{\theta}_k$ for all $k \geq 11$, so the first two groups have distinct coefficients, then moving any coefficient from the first group towards the second, and so increasing the number of estimated groups, actually \emph{decreases} the penalty contribution in \eqref{eq:all_pairs}. Specifically, if the $k$th coefficient for some $k \in \{1,\ldots,4\}$ moves to $\hat{\theta}_k + t$ for $t \in [0, \hat{\theta}_5 - \hat{\theta}_4]$ with all other coefficients kept fixed, the penalty contribution decreases by $13 t$.
In this case then, CAS-ANOVA will struggle to keep the groups intact, especially smaller ones. We see this in Figure~\ref{fig:earlydiagrams}, which shows the result of applying CAS-ANOVA to data generated according to \eqref{eq:ANOVA_mod} with $p=1$, $\mbb\theta^0$ as above, $n=20$ (so we have a single observation corresponding to each level), and $\varepsilon_i \iid \mathcal{N}(0, 1)$. There is no value of the tuning parameter $\lambda$  where the true groups are recovered.

\begin{figure}[h]
  \centering
  \includegraphics[width = \textwidth]{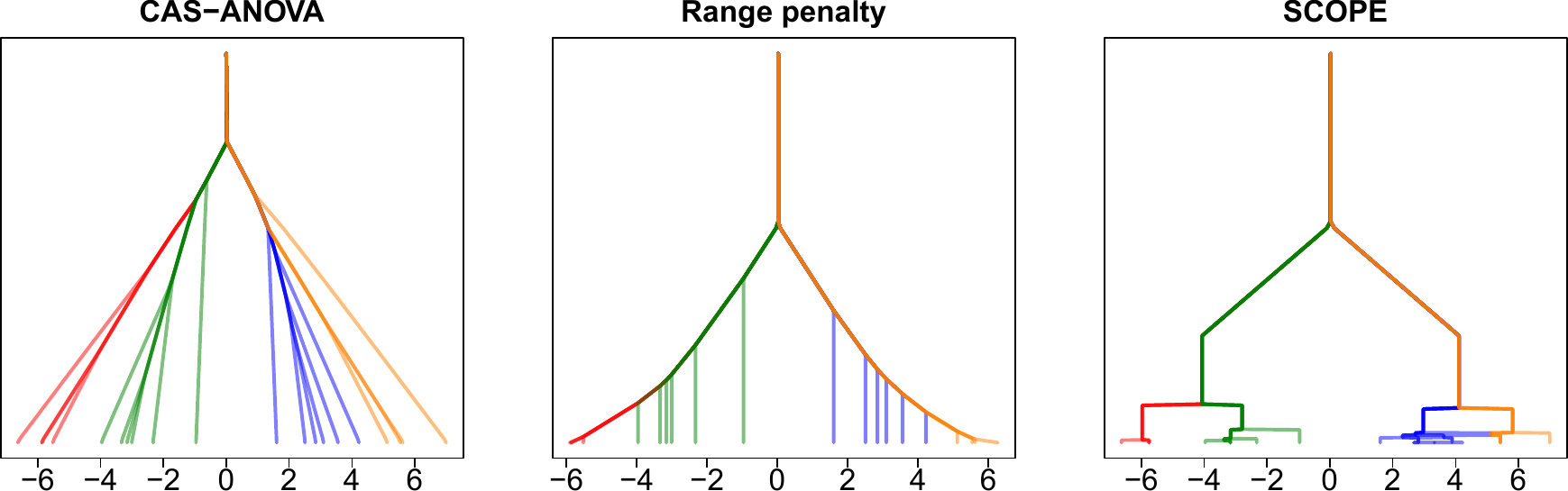}	
  \caption{Solution paths as the tuning parameter varies in a univariate example where there are four true groups. From left to right: CAS-ANOVA, the range penalty and SCOPE with $\gamma = 8$. The setup is as described in the main text of Section~\ref{sec:related}, with the different colours corresponding to the different true groups. The tuning parameter varies along the $y$ axis. In this example, only SCOPE identifies the 4 correct groups at any point along its solution path.} \label{fig:earlydiagrams}
  \end{figure}
As in the standard regression setting, the bias introduced by all-pairs $\ell_1$-type penalties may be reduced by choosing data-adaptive weights analogously to the adaptive Lasso \citep{zou2006adaptive}, or replacing the absolute value $|\theta_{jk} - \theta_{jl}|$ by $\rho(|\theta_{jk} - \theta_{jl}|)$ where $\rho$ is a concave and non-decreasing penalty function \citep{oelker2015selection, ma2017concave}. However, this does not address the basic issue of a preference for  groups of unequal sizes. 
Additionally, optimising an objective involving a penalty with $O\left(\sum_{j=1}^p K_j^2\right)$ summands can be computationally challenging, particularly in the case where $\rho$ is not convex, both in terms of runtime and memory.

To help motivate the new approach we are proposing in this paper, let us consider the setting where the predictors are ordinal rather than nominal, so there is an obvious ordering among the levels. In these settings, it is natural to consider a fused Lasso \citep{tibshirani2005sparsity}  penalty of the form
\begin{equation} \label{eq:perm_pen}
\sum_{j=1}^p \sum_{k=1}^{K_j - 1} |\theta_{j\pi_j(k + 1)} - \theta_{j\pi_j(k)}|,
\end{equation}
where $\pi_j$ is a permutation of $\{1,\ldots,K_j\}$ specifying the given order; this is done in \citet{gertheiss2010sparse} who advocate using it conjunction with the all-pairs-type CAS-ANOVA penalty for nominal categories.

If however we treat the nominal variable setting as analogous to having ordinal variables with unknown orderings $\pi_j$,  
one might initially think of choosing $\pi_j$ corresponding to the order of the estimates $\mbb\theta_j:=(\theta_{jk})_{k=1}^{K_j}$, such that $\theta_{j \pi_j(k)} = \theta_{j(k)}$, where $\theta_{j(k)}$ is the $k$th smallest entry in $\mbb\theta_j$. 
This however leads to what we refer to as the `range' penalty:
\begin{align}
 \sum_{k=1}^{K_j - 1} |\theta_{j(k + 1)} - \theta_{j(k)}| = \max_k \theta_{jk} - \min_k \theta_{jk}. \label{eq:rangepenalty}
\end{align}
Whilst this shrinks the largest and smallest of the estimated coefficients together, the remaining coefficients lying in the open interval between these are unpenalised and so no grouping of the estimates is encouraged, as we observe in Figure~\ref{fig:earlydiagrams}; see also \citet{oelker2015selection} for a discussion of this issue in the context of ordinal variables. 

\subsection{Our contributions and organisation of the paper} \label{sec:contents}
Given how all-pairs penalties have an intrinsic and undesirable preference for unequal group sizes, and how the fused Lasso applied to ordered coefficients \eqref{eq:rangepenalty} does not result in grouping of the coefficients, we propose the following solution.
Our approach is to use the penalty
\begin{align*}
	\sum_{j = 1}^p \sum_{k = 1}^{K_j - 1} \rho_j( \theta_{j(k+1)} - \theta_{j(k)}), 
\end{align*}
for concave (and nonconvex) non-decreasing penalty functions $\rho_j$,
which, for computational reasons we discuss in Section~\ref{sec:computation}, we base on the minimax concave penalty (MCP) \citep{zhang2010nearly}.
In Section~\ref{sec:scopeintro} we formally introduce our method, which we call SCOPE, standing for \textbf Sparse \textbf Concave \textbf Ordering \& \textbf Penalisation \textbf Estimator.

 Note that whereas in conventional high-dimensional regression, the use of nonconvex penalties has been primarily motivated by a need to reduce bias in the estimation of large coefficients \citep{fan2001variable}, here the purpose is very different: in our setting a nonconvex penalty is in fact even necessary for shrinkage to sparse solutions to occur (see Proposition~\ref{prop:convexpenalty}). Because of these fundamental differences, the rich algorithmic and statistical theory concerning high-dimensional regression with nonconvex penalties (see for example \citet{loh2012highdimensional,loh2015regularized, wang2014optimal, fan2018ilamm, zhao2018pathwise} and references therein) is not directly applicable to our setting.

In Section~\ref{sec:computation}, we therefore introduce a new dynamic programming approach that recovers the global minimum of the resulting objective function exactly in the univariate case, i.e.\ when $p = 1$.
We then build this into a blockwise coordinate descent approach to tackle the multivariate setting. 

In Section~\ref{sec:theory} we study the theoretical properties of SCOPE and give sufficient conditions for the estimator to coincide with the least squares solution with oracular knowledge of the level fusions in the univariate case. These conditions involve a minimal separation between unequal coefficients that is, up to constant factors, minimax optimal. Our result contrasts sharply with Theorem~2 of \citet{ma2017concave} for an all-pairs nonconvex penalty. The latter instead shows the existence of a local optimum that coincides with the oracle least squares solution. Whilst in conventional high-dimensional regression settings, it is known that under certain conditions, all local optima have favourable properties \citep{loh2015regularized}, we note that the separation requirements in \citet{ma2017concave} are substantially weaker than those indicated by the minimax lower bound, and so cannot be extended to a particular local optimum determined by the data; see the discussion following Theorem~\ref{thm:univarglobal}.

We use our univariate result to show that the oracle least squares solution is a fixed point of our blockwise coordinate descent algorithm in the multivariate case. In Section~\ref{sec:extensions} we outline some extensions of our methodology including a scheme for handling settings when there is a hierarchy among the categorical variables.  Section~\ref{sec:numerical} contains numerical experiments that demonstrate the favourable performance of our method compared to a range of competitors on both simulated and real data.
We conclude with a discussion in Section~\ref{sec:discuss}.
Further details of our algorithm can be found in the Appendix. The supplementary material contains additional information on the runtime of our algorithm, and an approximate version suitable for very large-scale settings, all the proofs, and additional information on the experiments in Section~\ref{sec:numerical}.

\section{SCOPE methodology} \label{sec:scopeintro}
Recall that our goal is to estimate parameters $(\mu^0, \mbb\theta^0)$ in model \eqref{eq:ANOVA_mod}.
Let us first consolidate some notation.
For any $\mbb\theta \in \R^{K_1} \times \cdots \times \R^{K_p}$, we define $\mbb\theta_j :=  (\theta_{jk})_{k=1}^{K_j} \in \R^{K_j}$.
We will study the univariate setting where $p=1$ separately, and so it will be helpful to introduce some simplified notation for this case, dropping any extraneous subscripts. We thus write $K \equiv K_1$, $X_i \equiv X_{i1}$ and $\rho \equiv \rho_1$.
Additionally, we let $\bar{Y}_k$ denote the average of the $Y_i$ with $X_i=k$:
\begin{equation} \label{eq:subaverage}
\bar{Y}_k = \frac{1}{n_k} \sum_{i=1}^n Y_i \ind_{\{X_{i} = k\}},
\end{equation}
where $n_k = \sum_{i=1}^n \ind_{\{X_{i} = k\}}$.

In order to avoid an arbitrary choice of corner point constraint, we instead impose the following to ensure that $\mbb\theta^0$ is identifiable: for all $j=1, \ldots, p$ we have
\begin{align}
g_j(\mbb\theta^0_j) &= 0, \text{ where } g_j(\mbb\theta_j) = \sum_{k=1}^{K_j} n_{jk}\theta_{jk} \text{ and } n_{jk} = \sum_{i=1}^n \ind_{\lbrace X_{ij}=k\rbrace }. \label{eq:identifiability}
\end{align}
Let $\Theta_j = \{\mbb\theta_j \in \R^{K_j} : g_j(\mbb\theta_j)=0\}$, and let $\Theta = \Theta_1 \times \cdots \times \Theta_p$.
We will construct estimators by minimising over $\mu \in \R$ and $\mbb\theta \in \Theta$ an objective function of the form
\begin{align*}
\tilde{Q}(\mu, \mbb\theta) = \ell(\mu, \mbb\theta) + \sum_{j=1}^p  \sum_{k=1}^{K_j-1} \rho_j(\theta_{j(k+1)} - \theta_{j(k)}),\label{eq:objective1}
\end{align*}
where $\ell$ is the least squares loss function \eqref{eq:least_squares} and $\theta_{j(1)} \leq \cdots \leq \theta_{j(K_j)}$ are the order statistics of $\mbb\theta_j$. We allow for different penalty functions $\rho_j$ for each predictor in order to help balance the effects of varying numbers of levels $K_j$.
The identifiability constraint that $\mbb\theta \in \Theta$ ensures that the estimated intercept $\hat{\mu} := \argmin_{\mu} \tilde{Q}(\mu, \mbb\theta)$ satisfies $\hat{\mu} = \sum_{i=1}^n Y_i/n$.

We note that whilst the form of the identifiability constraint would not have a bearing on the fitted values of unregularised least squares regression, this is not necessarily the case when regularisation is imposed. For example, consider the simple univariate setting with $p=1$ and the corner point constraint $\theta_{1}=0$. Then the fitted value for an observation with level $1$ would simply be the average $\bar{Y}_1$, coinciding with that of unpenalised least squares. However the fitted values with observations with other level $k \geq 2$ would be subject to regularisation and in general be different to $\bar{Y}_k$. This inequitable treatment of the levels is clearly undesirable as they may have been labelled in an arbitrary way. Our identifiability constraint treats the levels more symmetrically, but also takes into account the prevalence of levels, so the fitted values corresponding to more prevalent levels effectively undergo less regularisation.
 
As the estimated intercept $\hat{\mu}$ does not depend on the tuning parameters, we define
\begin{equation}
Q(\mbb\theta) = \frac{1}{2n}\sum_{i=1}^n \bigg(Y_i - \hat{\mu} - \sum_{j=1}^p \sum_{k=1}^{K_j} \theta_{jk}\ind_{\{X_{ij}=k\}}\bigg)^2 + \sum_{j=1}^p  \sum_{k=1}^{K_j-1} \rho_j(\theta_{j(k+1)} - \theta_{j(k)}).\label{eq:objective1}
\end{equation}
We will take the regularisers $\rho_j: [0, \infty) \to [ 0, \infty)$ in \eqref{eq:objective1} to be concave (and nonconvex); as discussed in the introduction and formalised in Proposition~\ref{prop:convexpenalty} below, a nonconvex penalty is necessary for fusion to occur.
\begin{prop}\label{prop:convexpenalty}
Consider the univariate case with $p=1$. Suppose the subaverages $(\bar{Y}_k)_{k=1}^K$ \eqref{eq:subaverage} are all distinct, and that $\rho_1 \equiv \rho$ is convex.
Then any minimiser $\hat{\mbb\theta}$ of $Q$ has $\hat{\theta}_{k} \neq \hat{\theta}_{l}$ for all $k \neq l$  such that $\hat{\theta}_{(1)}<\bar{Y}_{k} - \hat{\mu}<\hat{\theta}_{(K)}$ or $\hat{\theta}_{(1)}<\bar{Y}_{l} - \hat{\mu}<\hat{\theta}_{(K)}$.
\end{prop}
We base the penalties $\rho_j:[0, \infty) \to [0, \infty)$ on the MCP (Minimax Concave Penalty) \citep{zhang2010nearly}:
\begin{equation*}
\rho(x) = \rho_{\gamma, \lambda}(x) = \int_0^x \lambda \left( 1 - \frac{t}{\gamma \lambda }\right)_+ dt,
\end{equation*}
where $(u)_+ = u \ind_{\{u \geq 0\}}$. This is a piecewise quadratic function with gradient $\lambda$ at $0$ and flat beyond $\gamma \lambda$.
For computational reasons which we discuss in Section~\ref{sec:computation}, the simple piecewise quadratic form of this is particularly helpful.
In the multivariate case we take $\rho_j = \rho_{\gamma, \lambda_j}$ with $\lambda_j = \lambda \sqrt{K_j}$. This choice of scaling is motivated by requiring that when $\mbb\theta^0=0$ we also have $\hat{\mbb\theta}=0$ with high probability; see Lemma~\ref{thm:nullconsistency} in the Supplementary material.
We discuss the choice of the tuning parameters $\lambda$ and $\gamma$ in Section~\ref{sec:tuning}, but first turn to the problem of optimising \eqref{eq:objective1}. 

\section{Computation} \label{sec:computation}
In this section we include details of how SCOPE is computed. Section~\ref{sec:univariate} motivates and describes the dynamic programming algorithm we use to compute global minimiser of the SCOPE objective, which is highly non-convex.
Section~\ref{sec:bcd} contains details of how this is used to solve the multivariate objective by embedding it within a blockwise coordinate descent routine.
 Discussion of practical considerations is contained in Section~\ref{sec:tuning}.

\subsection{Univariate model} \label{sec:univariate}
\subsubsection{Preliminaries} \label{sec:preliminaries}
We now consider the univariate case ($p = 1$) and explain how the solutions are computed.
In this case, we may rewrite the least squares loss contribution to the objective function in the following way.
\begin{align}
		\frac{1}{2n}\sum_{i=1}^{n}\bigg(Y_i - \hat{\mu} -  \sum_{k=1}^K \theta_k  \ind_{\{ X_i = k\}} \bigg)^2 &= \frac{1}{2n} \sum_{k=1}^K \sum_{i = 1}^n\ind_{\{ X_i = k \} } (Y_i - \hat{\mu} - \theta_k)^2 \notag \\
	&= \frac{1}{2} \sum_{k=1}^K w_k ( \bar{Y}_k - \hat{\mu} - \theta_k )^2 + \frac{1}{2n} \sum_{i=1}^n \sum_{k=1}^K \ind_{\{ X_j = k \}} ( Y_i - \bar{Y}_k)^2  \label{eq:alt_objective}
\end{align}
where $w_k$ = $n_k / n$. Thus the optimisation problem \eqref{eq:objective1} can be written equivalently as
\begin{align} \label{eq:univar}
	\hat{\mbb\theta} \in \argmin_{\mbb\theta \in \Theta} \frac{1}{2}\sum_{k=1}^K w_k\left( \bar{Y}_k - \hat{\mu} - \theta_k \right)^2 +  \sum_{k = 1}^{K-1} \rho \left( \theta_{(k+1)} - \theta_{(k)} \right), 
\end{align}
suppressing the dependence of the MCP $\rho$ on tuning parameters $\gamma$ and $\lambda$.
In fact, it is straightforward to see that the constraint that the solution lies in $\Theta$ will be automatically satisfied, so we may replace $\Theta$ with $\R^K$.
Two challenging aspects of the optimisation problem above are the presence of the nonconvex $\rho$ and the order statistics. The latter however are easily dealt with using the result below, which holds more generally whenever $\rho$ is a concave function.
\begin{prop} \label{thm:orderprop}
Consider the univariate optimisation \eqref{eq:univar} with $\rho$ any concave function such that a minimiser $\hat{\mbb\theta}$ exists. If for $k, l$ we have $\bar{Y}_{k} > \bar{Y}_{l}$, then $\hat{\theta}_{k} \geq \hat{\theta}_{l}$.	
\end{prop}
This observation substantially simplifies the optimisation: after re-indexing such that $\bar{Y}_1 \leq \bar{Y}_2 \leq \cdots \leq \bar{Y}_K$, we may re-express \eqref{eq:univar} as,
\begin{align} \label{eq:univarobjordered}
	\hat{\mbb\theta} \in \argmin_{\mbb\theta : \theta_1 \leq \cdots \leq \theta_K} \left\{ \frac{1}{2}\sum_{k=1}^K w_k\left( \bar{Y}_k- \hat{\mu} - \theta_k \right)^2 +  \sum_{k = 1}^{K-1} \rho \left( \theta_{k+1} - \theta_{k} \right)\right\}.
\end{align}

We use the following intermediate functions to structure the algorithm:
\begin{align}
f_1(\theta_1) &= \frac{1}{2}w_1( \bar{Y}_1 - \hat{\mu} - \theta_1 )^2, \notag\\
f_k(\theta_k) &= \min_{\theta_{k-1} : \theta_{k-1} \leq \theta_k} \{f_{k-1}(\theta_{k-1}) + \rho(\theta_{k} - \theta_{k-1})\} + \frac{1}{2}w_k(\bar{Y}_k - \hat{\mu} - \theta_k)^2,  \label{eq:f_k_def}\\ 
b_k(\theta_{k}) &= \sargmin_{\theta_{k-1} : \theta_{k-1} \leq \theta_{k}} \{f_{k-1}(\theta_{k-1}) + \rho(\theta_{k} - \theta_{k-1})\}, \notag
\end{align}
for $k=2,\ldots,K$; here $\sargmin$ refers to the smallest minimiser in the case that it is not unique. Invariably however this will be unique, as the following result indicates.
\begin{prop} \label{prop:uniquewhp}
	The set of $(\bar{Y}_k)_{k=1}^K$ that yields distinct solutions to \eqref{eq:univar} has Lebesgue measure zero as a subset of $\R^K$.
\end{prop}
We will thus tacitly assume uniqueness in some of the discussion that follows, though this is not required for our algorithm to return a global minimiser.
Observe now that $\hat{\theta}_K$ is the minimiser of the univariate objective function $f_K$: indeed for $k \geq 2$,
\begin{align}
f_k(\theta_k) = \min_{(\theta_1,\ldots,\theta_{k-1}) : \theta_1 \leq \cdots \leq \theta_{k-1}\leq \theta_k } \bigg\{ \frac{1}{2} \sum_{l=1}^{k} w_l(\bar{Y}_l -\hat{\mu} - \theta_l)^2 + \sum_{l=1}^{k-1}\rho(\theta_{l+1} - \theta_l) \bigg\}. \label{eq:fkfull}
\end{align}
Furthermore, we have $\hat{\theta}_{K-1} = b_{K}(\hat{\theta}_K)$, and more generally $\hat{\theta}_k = b_{k+1}(\hat{\theta}_{k+1})$ for $k=K-1,\ldots,1$. Thus provided $f_K$ can be minimised efficiently (which we shall see is indeed the case), given this and the functions $b_2,\ldots,b_K$ we can iteratively compute $\hat{\theta}_K, \hat{\theta}_{K-1},\ldots,\hat{\theta}_1$. In order to make use of these properties, we must be able to compute $f_K$ and the $b_k$ efficiently; we explain how to do this in the following subsection.

\subsubsection{Computation of $f_K$ and $b_2,\ldots,b_{K}$} \label{sec:comp_details}
The simple piecewise quadratic form of the MCP-based penalty is crucial to our approach for computing the $f_K$ and the $b_k$. Some important consequences of this piecewise quadratic property are summarised in the following lemma.

\begin{lem} \label{lem:optim_prop}
For each $k$,
\begin{enumerate}[(i)]
\item $f_k$ is continuous, coercive and piecewise quadratic with finitely many pieces;
\item $b_k$ is piecewise linear with finitely many pieces;
\item for each $\theta_{k+1} \in \R$, if a minimiser $\tilde{\theta}_k = \tilde{\theta}_k(\theta_{k+1})$ of $\theta_k \mapsto f_k(\theta_k) + \rho(\theta_{k+1} - \theta_k)$ over $(-\infty, \theta_{k+1}]$ satisfies 
$\tilde{\theta}_k < \theta_{k+1}$,
then $f_k$ must be differentiable at $\tilde{\theta}_k$.
\end{enumerate}
\end{lem}
Properties (i) and (ii) above permit exact representation of $f_k$ and $b_k$ with finitely many quantities.
The key task then is to form the collection of intervals and corresponding coefficients of quadratic functions for
\begin{align}
g_k(\theta_{k+1}) := \min_{\theta_k : \theta_k \leq \theta_{k+1}} \{f_{k}(\theta_k) + \rho(\theta_{k+1} - \theta_k)\} \label{eq:gfunction}
\end{align}
given a similar piecewise quadratic representation of $f_k$; and also the same for the linear functions composing $b_k$. A piecewise quadratic representation of $f_{k+1}$ would then be straightforward to compute, and we can iterate this process. To take advantage of property (iii) above, in computing $g_k(\theta_{k+1})$ we can separately search for minimisers at stationary points in $(-\infty, \theta_{k+1})$ and compare the corresponding function values with $f_k(\theta_{k+1})$; the fact that we need only consider potential minimisers at points of differentiability will simplify things as we shall see below.

Suppose $I_{k,1},\ldots,I_{k,m(k)}$ are intervals that partition $\R$ (closed on the left) and $q_{k,1},\ldots,q_{k,m(k)}$ are corresponding quadratic functions such that $f_k(\theta_k) = q_{k,r}(\theta_k)$ for $\theta_k \in I_{k,r}$. Let us write
\[
\tilde{q}_{k,r}(\theta_k) = \begin{cases} q_{k,r}(\theta_k) &\text{ if } \theta_k \in I_{k,r} \\
\infty &\text{ otherwise.}
\end{cases}
\]
We may then express $f_k$ as $f_k(\theta_k) = \min_r \tilde{q}_{k,r}(\theta_k)$. We can also express the penalty $\rho = \rho_{\gamma, \lambda}$ in a similar fashion. Let
\begin{align*}
\tilde{\rho}_1(x) &:= -\gamma\lambda^2\{1-x/(\gamma\lambda)\}^2/2 + \gamma\lambda^2/2 \;\text{ if }\; 0 \leq x < \gamma\lambda \;\text{ and }\; \infty \;\text{ otherwise},\\
\tilde{\rho}_2(x) &:= \gamma\lambda^2/2  \;\text{ if }\; x \geq \gamma\lambda \;\text{ and }\; \infty \;\text{ otherwise}.
\end{align*}
Then $\rho(x) = \min_t \tilde{\rho}_t(x)$ for $x \geq 0$. Let $D_k$ be the set of points at which $f_k$ is differentiable. We then have, using Lemma~\ref{lem:optim_prop} (iii) that
\begin{align}
g_k(\theta_{k+1}) &= \min_{\theta_k: \theta_k \leq \theta_{k+1}} \{ \min_r \tilde{q}_{k,r}(\theta_k) + \min_t \tilde{\rho}_t(\theta_{k+1} - \theta_k)\} \notag\\
&= \min [\, \tilde{\min_{\theta_k \in D_k : \theta_k < \theta_{k+1}}} \min_{r,t} \{ \tilde{q}_{k,r}(\theta_k) + \tilde{\rho}_t(\theta_{k+1} - \theta_k)\}, f_k(\theta_{k+1})] \notag\\
&= \min [ \, \min_{r,t} \tilde{\min_{\theta_k \in D_k : \theta_k < \theta_{k+1}}} \{ \tilde{q}_{k,r}(\theta_k) + \tilde{\rho}_t(\theta_{k+1} - \theta_k)\}, f_k(\theta_{k+1})], \label{eq:g_k_def}
\end{align}
where $\tilde{\min}$ denotes the minimum if it exists and $\infty$ otherwise. The fact that in the inner minimisation we are permitted to consider only points in $D_k$ simplifies the form of
\begin{equation} \label{eq:u_def}
u_{k, r, t}(\theta_{k+1}) := \tilde{\min_{\theta_k \in D_k : \theta_k < \theta_{k+1}}} \{ \tilde{q}_{k,r}(\theta_k) + \tilde{\rho}_t(\theta_{k+1} - \theta_k)\}.
\end{equation}
We show in Section~\ref{sec:analyticminima} of the Appendix that this is finite only on an interval and there takes the value of a quadratic function; coefficients for this function and the interval endpoints have closed form expressions that are elementary functions of the coefficients and intervals corresponding to $\tilde{q}_{k,r}$.
With this, we have an explicit representation of $g_k$ as the minimum of a collection of functions that are quadratic on intervals and $\infty$ everywhere else. Let us refer to these intervals (closed on the left) and corresponding quadratic functions as $J_{k,1},\ldots, J_{k,n(k)}$ and $p_{k,1},\ldots,p_{k,n(k)}$ respectively.

In order to produce a representation of $f_{k+1}$ for use in future iterations, we must express $g_k$ as a collection of quadratics defined on \emph{disjoint} intervals.
To this end, define for each $x \in \R$ the \emph{active set at} $x$, $A(x) = \{r : x \in J_{k,r}\}$. Note that the endpoints of the intervals $J_{k,r}$ are the points where the active set changes and it is thus straightforward to determine $A(x)$ at each $x$. Let $r(x)$ be the index such that $g_k(x)=p_{k,r(x)}(x)$.
For large negative values of $x$, $A(x)$ will contain a single index and for such $x$ this must be $r(x)$.
Consider also for each $r \in A(x) \setminus \{r(x)\}$, the horizontal coordinate $x'$ of the first intersection beyond $x$ (if it exists) between $p_{k,r}$ and $p_{k,r(x)}$. We refer to the collection of all such tuples $(x', r)$ as the \emph{intersection set at} $x$ and denote it by $N(x)$. Given $r(x)$, $N(x)$ can be computed easily. The intersection set $N(x)$ then in turn helps to determine the smallest $x'>x$ where $r(x') \neq r(x)$ changes, that is the next knot of $g_k$ beyond $x$, as we now explain. Suppose at a point $x_{\text{old}}$, we have computed $r_{\text{old}} = r(x_{\text{old}})$. We set $x_{\text{cur}} = x_{\text{old}}$ and perform the following.
\begin{enumerate}
\item Given $r(x_{\text{cur}})$, compute $N(x_{\text{cur}})$ and set $(x_{\text{int}}, r_{\text{int}}) = \argmin_{(x, r) \in N(x_{\text{cur}})} x$.
\item If there are no changes in the active set between $x_{\text{cur}}$ and $x_{\text{int}}$, we have found the next knot point at $x_{\text{int}}$ and $r_{\text{int}} = r(x_{\text{int}})$.
\item If instead the active set changes, move $x_{\text{cur}}$ to the leftmost change point. We have that $r(x)=r_{\text{old}}$ for $x \in [x_{\text{old}}, x_{\text{cur}})$. To determine if $r(x)$ changes at $x_{\text{cur}}$, we check if
\begin{enumerate}[(i)]
\item $r_{\text{old}}$ leaves the active set at $x_{\text{cur}}$, so $r_{\text{old}} \notin A(x_{\text{cur}})$, or
\item some $r_{\text{new}}$ enters the active set at $x_{\text{cur}}$ and `beats' $r_{\text{old}}$, so $r_{\text{new}} \in A(x_{\text{cur}}) \setminus A(x_{\text{old}})$ and $p_{k, r_{\text{new}}}(x_{\text{cur}} + \epsilon) < p_{k, r_{\text{old}}}(x_{\text{cur}} + \epsilon)$ for $\epsilon>0$ sufficiently small.
\end{enumerate}
If either hold $x_{\text{cur}}$ is a knot and $r(x_{\text{cur}})$ may be computed via $r(x_{\text{cur}}) = \argmin_{r \in A(x_{\text{cur}})} p_{k, r}(x_{\text{cur}})$. If neither hold, we conclude that $r(x_{\text{cur}}) = r_{\text{old}}$ and go to step 1 once more.
\end{enumerate}
Hence we can proceed from one knot of $g_k$ to the next by comparing the values and intersections of a small collection of quadratic functions, and thereby form a piecewise quadratic representation of $g_k$ in a finite number of steps. Figure~\ref{fig:algillustration1} illustrates the steps outlined above. The pieces of $b_k$ may be computed in a similar fashion.

We note there are several modifications that can speed up the algorithm: for example, for each $r$, $u_{k,r,2}$ \eqref{eq:u_def} is a constant function where it is finite (see $p_{k, 2}$ in the figure), and these can be dealt with more efficiently. For further details including pseudocode see Section~\ref{sec:algdetail} of the Appendix.
\begin{figure}[h!]
\centering
\includegraphics[width=\textwidth]{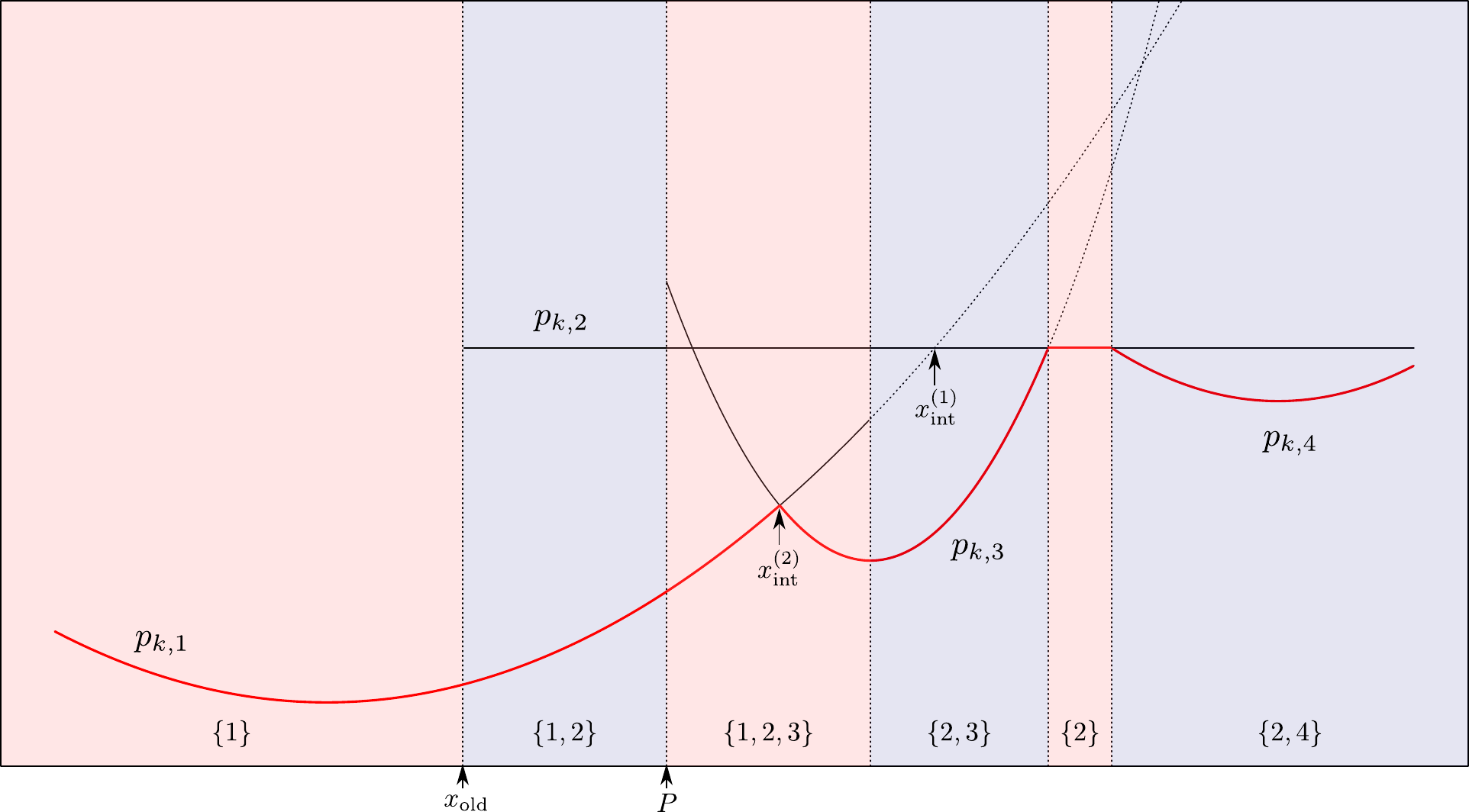}
\caption{Illustration of the optimisation problem and our algorithm, to be interpreted with reference to steps 1, 2, 3 in the main text. Shading indicates regions where the active set, displayed at the bottom of the plot, is invariant, and vertical dotted lines signify changes. Dotted curves correspond to parts of quadratic functions $p_{k,l}$ lying outside their associated intervals $J_{k,l}$.
At $x_{\text{old}}$, we have $r(x_{\text{old}})=1$, $A(x_{\text{old}})=\{1,2\}$ and $N(x_{\text{old}}) = \{(x_{\text{int}}^{(1)}, 2)\}$. Since the active set changes between $x_{\text{old}}$ and $x_{\text{int}}^{(1)}$, we move $x_{\text{cur}}$ to the first change point $P$ and  see neither (i) nor (ii) occur. We therefore return to step 1 and compute $N(x_{\text{cur}})$  which additionally contains $(x_{\text{int}}^{(2)}, 2)$. As the active set is unchanged between $x_{\text{cur}}$ and $x_{\text{int}}^{(2)}$, we have determined the next knot point $x_{\text{int}}^{(2)}$ and minimising quadratic $p_{k,3}$.
\label{fig:algillustration1}}	
\end{figure}

In summary, our algorithm produces a piecewise quadratic representation of $f_K$, which we can minimise efficiently to obtain $\hat{\theta}_K$. We also have piecewise linear representations of functions $b_2,\ldots,b_K$ through which we may iteratively obtain $\hat{\theta}_k = b_{k+1}(\hat{\theta}_{k+1})$ for $k=K-1,\ldots,1$.

It seems challenging to obtain meaningful bounds on the number of computations that must be performed at each stage of this process in terms of parameters of the data. However, to give an indication of the scalability of this algorithm, we ran a simple example with 3 true levels and found that with 50 categories the runtime was under $10^{-3}$ seconds; with 2000 categories it was still well under half a second. More details on computation time can be found in Sections~\ref{sec:comptimeexp} and \ref{sec:furtherexpdet} of the Supplementary material. In Section~\ref{sec:discretealg} of the Supplementary material, we describe an approximate version of the algorithm that can be used for fast computation in very large-scale settings.

\subsection{Multivariate model} \label{sec:bcd}
Using our dynamic programming algorithm for the univariate problem, we can attempt to minimise the objective \eqref{eq:objective1} for the multivariate problem using block coordinate descent.
This has been shown empirically to be a successful strategy for minimising objectives for high-dimensional regression with nonconvex penalties such as the MCP \citep{breheny2011coordinate, mazumder2011sparsenet, breheny2015group}, and we take this approach here. Considering the multivariate case, we iteratively minimise the objective $Q$ over $\mbb\theta_j := (\theta_{jk})_{k=1}^{K_j} \in \Theta_j$ keeping all other parameters fixed.
Then for a given $(\gamma, \lambda)$ and initial estimate $\hat{\mbb\theta}^{(0)} \in \Theta$, we repeat the following until a suitable convergence criterion is met:
\begin{enumerate}
\item Initialise $m = 1$,
and set for $i=1,\ldots,n$
\[
R_i = Y_i - \hat{\mu} - \sum_{l = 1}^p \sum_{k=1}^{K_l} \hat{\theta}_{lk}^{(m-1)}\ind_{\{X_{i l}=k\}}.
\]
\item For $j = 1, \ldots, p$, compute
\begin{align}
R^{(j)}_i &= R_i + \sum_{k = 1}^{K_j} \hat{\theta}_{jk}^{(m-1)} \ind_{\{X_{i j}=k\}} \qquad \text{for each } i,  \label{eq:firstpartialresiduals}\\
\hat{\mbb\theta}^{(m)}_j &= \argmin_{\mbb\theta_j \in \Theta_j} \bigg\{\frac{1}{2n}\sum_{i=1}^n \bigg(R^{(j)}_i - \sum_{k=1}^{K_j} \theta_{jk}\ind_{\{X_{ij}=k\}}\bigg)^2 +  \bigg(\sum_{k=1}^{K_j-1} \rho_j(\theta_{j(k+1)} - \theta_{j(k)})\bigg) \bigg\} \label{eq:iterupdate} \\
R_i &= R^{(j)}_i - \sum_{k = 1}^{K_j} \hat{\theta}_{jk}^{(m)} \ind_{\{X_{i j}=k\}} \qquad \text{for each } i.  \notag
\end{align}
\item Increment $m \to m+1$.
\end{enumerate}
We define a blockwise optimum of $Q$ to be any $\hat{\mbb\theta} \in \Theta$, such that for each $j = 1, \ldots, p$,
\begin{equation} \label{eq:block_opt}
\hat{\mbb\theta}_j \in \argmin_{\mbb\theta_j \in \Theta_j} Q(\hat{\mbb\theta}_1, \ldots, \hat{\mbb\theta}_{j-1}, \mbb\theta_j, \hat{\mbb\theta}_{j+1}, \ldots, \hat{\mbb\theta}_p).
\end{equation}
This is equivalent to $\hat{\mbb\theta}$ being a fixed point of the block coordinate descent algorithm above. Provided $\gamma > 0$, $Q$ is continuous in $\mbb\theta$. As a consequence of \citet{tseng2001convergence}, Theorem 4.1 (c), provided the minimisers $\hat{\mbb\theta}^{(m)}_j$ in \eqref{eq:iterupdate} are unique for all $j$ and $m$ (which will invariably be the case when the responses are realisations of continuous random variables; see Proposition~\ref{prop:uniquewhp}), then all limit points of the sequence $(\hat{\mbb\theta}^{(m)})_{m=0}^\infty$ are blockwise optima.

\subsection{Practicalities} \label{sec:tuning}
In practice the block coordinate descent procedure described above must be performed over a grid of $(\gamma, \lambda)$ values to facilitate tuning parameter selection by cross-validation. In line with analogous recommendations for other penalised regression optimisation procedures \citep{breheny2011coordinate,friedman2010regularization}, we propose, for each fixed $\gamma$, to iteratively obtain solutions for an exponentially decreasing sequence of $\lambda$ values, warm starting each application of block coordinate descent at the solution for the previous $\lambda$. It is our experience that this scheme speeds up convergence and helps to guide the resulting estimates to statistically favourable local optima, as has been shown theoretically for certain nonconvex settings \citep{wang2014optimal}.

The grid of $\gamma$ values can be chosen to be fairly coarse as the solutions appear to be less sensitive to this tuning parameter; in fact fixing $\gamma \in \{8, 32\}$ yields competitive performance across a range of settings (see Section~\ref{sec:numerical}). The choice $\gamma \downarrow 0$, which mimics the $\ell_0$ penalty, has good statistical properties (see Theorem~\ref{thm:univarglobal} and following discussion). However the global optimum typically has a smaller basin of attraction and can be prohibitively hard to locate, particularly in low signal to noise ratio settings where larger $\gamma$ tends to dominate.

\section{Theory} \label{sec:theory}
In this section, we study the theoretical properties of SCOPE. Recall our model
\begin{equation} \label{eq:ANOVA_mod2}
Y_i = \mu^0 + \sum_{j=1}^p \sum_{k=1}^{K_j} \theta_{jk}^0\ind_{\{X_{ij}=k\}} + \varepsilon_i
\end{equation}
for $i=1,\ldots,n$, where $\mbb\theta^0 \in \Theta$. We will assume the errors $(\varepsilon_i)_{i=1}^n$ have mean zero, are independent and sub-Gaussian with parameter $\sigma$.
Let
\[
\Theta_0 = \left\lbrace \mbb\theta \in \Theta  : \theta_{jk} = \theta_{jl} \text{ whenever } \theta^0_{jk} = \theta^0_{jl} \text{ for all } j \right\rbrace
\]
and define the \emph{oracle least squares estimate} 
\begin{align}
		\hat{\mbb\theta}^0 := \argmin_{\mbb\theta \in \Theta_{0}} \frac{1}{2n} \sum_{i=1}^n\left( Y_i - \hat{\mu} - \sum_{j=1}^p \sum_{k=1}^{K_j} \theta_{jk} \ind_{\lbrace X_{ij} = k \rbrace } \right)^2. \label{eq:oraclelse}
\end{align}
This is the least squares estimate of $\mbb\theta^0$ with oracular knowledge of which categorical levels are fused in $\mbb\theta^0$.

Note that in the case where the errors have equal variance $v^2$, the expected mean squared prediction error of $\hat{\mbb\theta}^0$ satisfies
\[
\E \left\{ \frac{1}{n} \sum_{i=1}^n \left(\hat{\mu} - \mu^0 +  \sum_{j=1}^p \sum_{k=1}^{K_j} (\hat{\theta}^0_{jk} - \theta^0_{jk})\ind_{\lbrace X_{ij} = k \rbrace } \right)^2 \right\} \leq \frac{v^2}{n} \left(1 + \sum_{j=1}^p (s_j-1)\right),
\]
with equality when $\hat{\mbb\theta}^0$ is unique.

Our results below establish conditions under which $\hat{\mbb\theta}^0$ is a blockwise optimum \eqref{eq:block_opt} of the SCOPE objective function $Q$ \eqref{eq:objective1}, or in the univariate case when this in fact coincides with SCOPE. The minimum differences between the signals defined for each $j$ by
\begin{align}
	\Delta(\mbb\theta^0_j) := \min_{k, l} \left\lbrace | \theta^0_{jk} - \theta^0_{jl} | : \theta^0_{jk} \neq \theta^0_{jl} \right\rbrace, \label{eq:minsigsep}
\end{align}
will play a key role. If all components of $\mbb\theta^0_j$ are equal we take $\Delta(\mbb\theta^0_j)$ to be $\infty$.
We also introduce $n_{j,\text{min}} = \min_k n_{jk}$,
\[
n^0_{j, \text{min}} = \min_k \sum_{l \,:\, \theta^0_{jl} = \theta^0_{jk}} n_{jl} \qquad \text{and} \qquad n^0_{j, \text{max}} = \max_k \sum_{l \,:\, \theta^0_{jl} = \theta^0_{jk}} n_{jl};
\]
these latter two quantities are the minimum and maximum number of observations corresponding to a set of fused levels in the $j$th predictor respectively.

\subsection{Univariate model} \label{sec:univartheory}
We first consider the univariate case, where as usual we will drop the subscript $j$ for simplicity. The following result establishes conditions for recovery of the oracle least squares estimate \eqref{eq:oraclelse}.

\begin{thm} \label{thm:univarglobal}
Consider the model \eqref{eq:ANOVA_mod2} in the univariate case with $p=1$. 
Suppose there exists $\eta \in (0, 1]$ such that $\eta /s \leq n^0_{j, \text{min}}/n \leq n^0_{j, \max}/n \leq 1 / \eta s$.
Let $\gamma_* = \min \{ \gamma, \eta s \}$ and $\gamma^* = \max \{ \gamma, \eta s \}$.
Suppose further that
	\begin{align}
		\Delta( \mbb\theta^0) \geq 3 \left(1 + \sqrt{2}/\eta \right) \sqrt{\gamma \gamma^*} \lambda.\label{eq:univarsigstrength}
	\end{align}
	Then with probability at least 
	\begin{align}
		 1 - 2 \exp \left( - \frac{ n_{\text{min}} \eta  s \gamma_* \lambda^2}{8\sigma^2 } + \log(K) \right), \label{eq:univartail}
	\end{align}
	the oracle least squares estimate $\hat{\mbb\theta}^0$ \eqref{eq:oraclelse} is the global optimum of \eqref{eq:objective1}, so $\hat{\mbb\theta} = \hat{\mbb\theta}^0$.
\end{thm}
For a choice of the tuning parameters $(\gamma, \lambda)$ with $\gamma \leq \eta s$ and $\lambda$ such that equality holds in \eqref{eq:univarsigstrength}, we have, writing $\Delta \equiv \Delta(\mbb\theta^0)$, that $\hat{\mbb\theta} = \hat{\mbb\theta}^0$ with probability at least
\[
1-2\exp \left(-c \eta^2 n_{\text{min}} \Delta^2 / \sigma^2  + \log(K) \right),
\]
where $c$ is an absolute constant.
The quantity $\eta$ reflects how equal the number of observations in the true fused levels are: in settings where the prevalences of the underlying true levels are roughly equal, we would expect this to be closer to $1$.

Consider now an asymptotic regime where $K$, $s$ and $1/\Delta$ are allowed to diverge with $n$, $n_{\text{min}} \asymp n/K$, so all levels have roughly the same prevalence, and $\eta$ is bounded away from zero, so all true underlying levels also have roughly the same prevalence. Then in order for $\hat{\mbb\theta} = \hat{\mbb\theta}^0$ with high probability, we require $\Delta \gtrsim \sigma \sqrt{K \log(K) / n}$. This requirement cannot be weakened for any estimator; this fact comes as a consequence of minimax lower bounds on mis-clustering errors in Gaussian mixture models \citep[Theorem~3.3]{lu2016statistical}.

We remark that our result here concerning properties of the global minimiser of our objective is very different from existing results on local minimisers of objectives involving all-pairs-type penalties. For example, in the setting above where $K = n$, Theorem~2 of \citet{ma2017concave} gives that provided $s = o(n^{1/3} (\log n)^{-1/3})$ and $\Delta \gg \sigma s^{3/2} n^{-1/2} \sqrt{\log(n)}$, there exists a sequence of local minimisers converging to the oracle least-squares estimate with high probability. This is significantly weaker than the condition $\Delta \gtrsim \sigma \sqrt{\log(n)}$ required for any estimator to recover oracle least-squares in this setting, illustrating the substantial difference between results on local and global optima here.

\subsection{Multivariate model} \label{sec:multivartheory}
When the number of variables is $p >1$, models can become high-dimensional, with ordinary least squares estimation failing to provide a unique solution. We will however assume that the solution for $\mbb\theta \in \Theta_0$ to
\begin{align*}
\sum_{j=1}^p\sum_{k=1}^{K_j} \theta^0_{jk} \ind_{\{X_{ij}=k\}}   = \sum_{j=1}^p\sum_{k=1}^{K_j} \theta_{jk} \ind_{\{X_{ij}=k\}} 
\end{align*}
is unique, which occurs if and only if the oracle least squares estimate \eqref{eq:oraclelse} is unique. In this case, we note that $\hat{\mbb\theta}^0 = AY$ for a fixed matrix $A$. A necessary condition for this is that $\sum_j ( s_j - 1 ) < n$.

Our result below provides a bound on the probability that the oracle least squares estimate is a blockwise optimum of the 
SCOPE objective \eqref{eq:objective1} with $\rho_j = \rho_{\gamma_j, \lambda_j}$.
This is much more meaningful than an equivalent bound for $\hat{\mbb\theta}^0$ to be a local optimum
as the number of local optima will be enormous.  
In general though there may be several blockwise optima, and it seems challenging to obtain a result giving conditions under which our blockwise coordinate descent procedure is guaranteed to converge to $\hat{\mbb\theta}^0$. Our empirical results (Section~\ref{sec:numerical}) however show that the fixed points computed in practice tend to give good performance.
\begin{thm} \label{thm:multivaroracle}
Consider the model \eqref{eq:ANOVA_mod2} and assume $\hat{\mbb\theta}^0 = AY$.
Suppose that there exists $\eta \in (0, 1]$ such that $\eta /s_j \leq n^0_{j,\text{min}}/n  \leq n^0_{j,\text{max}}/n \leq 1 / \eta s_j$ for all $j = 1, \ldots, p$. Let ${\gamma_*}_j = \min \{ \gamma_j, \eta s_j \}$ and $\gamma^*_j = \max \{ \gamma_j, \eta s_j \}$. Further suppose that 
\begin{align}
		\Delta(\mbb\theta^0_j) \geq 3 \left( \frac{4}{3} + \frac{\sqrt{2}}{\eta} \right) \sqrt{\gamma_j \gamma^*_j} \lambda_j . \label{eq:multivarsigsep}
\end{align}
Then letting $c_\text{min} := (\max_l (AA^T)_{ll})^{-1}$, with probability at least 
\begin{align}
	1 - 4 \sum_{j=1}^p \exp \left( - \frac{ (n_{j, \text{min}} \wedge c_\text{min})  \eta {\gamma_*}_j s_j \lambda_j^2}{8\sigma^2 } + \log(K_j) \right), 
	\label{eq:multivarprob} 
\end{align} 
the oracle least squares estimate $\hat{\mbb\theta}^0$
is a blockwise optimum of \eqref{eq:objective1}.
\end{thm}
Now suppose $\gamma_j \leq \eta s_j$ and $\lambda_j$ are such that equality holds in \eqref{eq:multivarsigsep} for all $j$. Then writing $K_{\text{max}} = \max_j K_j$, $n_{\text{min}} = \min_j n_{j,\text{min}}$ and $\Delta_{\text{min}} = \min_j \Delta(\mbb\theta^0_j)$, we have that $\hat{\mbb\theta}^0$ is a blockwise optimum of \eqref{eq:objective1} with probability at least
\[
1-4\exp \left(-c \eta^2 (n_{\text{min}} \wedge c_{\text{min}}) \Delta_{\text{min}}^2 / \sigma^2  + \log(K_{\text{max}} p) \right),
\]
where $c$ is an absolute constant. Consider now an analogous asymptotic regime to that described in the previous section for the univariate case. Specifically assume $n_{\text{min}} \asymp n / K_{\text{max}}$ and $c_{\text{min}} \gtrsim n_{\text{min}}$ for simplicity. We then see that in order for $\hat{\mbb\theta}^0$ to be a blockwise optimum with high probability, it is sufficient that $\Delta_{\text{min}} \gtrsim \sigma\sqrt{K_{\text{max}} \log(K_{\text{max}} p) /n} $.

\section{Extensions} \label{sec:extensions}
In this section, we describe some extensions of our SCOPE methodology.

\paragraph{Continuous covariates.} If some of the covariates are continuous rather than categorical, we can apply any penalty function of choice to these, and perform a regression by optimising the sum of a least squares objective, our SCOPE penalty and these additional penalty functions, using (block) coordinate descent.

For example, consider the model \eqref{eq:ANOVA_mod} with the addition of $d$ continuous covariates. Let $Z \in \R^{n \times d}$ be the centred design matrix for these covariates with $i$th row $Z_i \in \R^d$. One can fit a model with SCOPE penalising the categorical covariates, and the Lasso with tuning parameter $\alpha > 0$ penalising the continuous covariates, resulting in the following objective over $\mbb\beta \in \R^d$ and $\mbb\theta \in \Theta$:
\begin{align*}
	\frac{1}{2n}\sum_{i = 1}^{n}\left( Y_i - \hat{\mu} - Z_i^T\mbb\beta - \sum_{j = 1}^p \sum_{k = 1}^{K_j} \theta_{jk} \ind_{\{ X_{ij} = k \}} \right)^2 + \alpha \| \mbb\beta \|_1 + \sum_{j=1}^p \sum_{k = 1}^{K_j - 1} \rho_j( \theta_{j (k+1)} - \theta_{j(k)} ).
\end{align*}
This sort of integration of continuous covariates is less straightforward when attempting to use tree-based methods to handle categorical covariates, for example.

\paragraph{Generalised linear models.} Sometimes a generalised linear model may be appropriate. Although a quadratic loss function is critical for our exact optimisation algorithm described in Section~\ref{sec:univariate}, we can iterate local quadratic approximations to the loss term in the objective and minimise this. This results in a proximal Newton algorithm and is analogous to the standard approach for solving $\ell_1$-penalised generalised linear models \citep[Section~3]{friedman2010regularization}.
An implementation of this scheme in the case of logistic regression for binary responses is available in the accompanying R package \texttt{CatReg}. We remark that when computing logistic regression models with a SCOPE penalty it is advisable to use a larger value of $\gamma$ than with a continuous response to aid convergence of the proximal Newton step; we recommend a default setting of $\gamma = 100$. In Section~\ref{sec:adultds1} we use the approach described above to perform a logistic regression using SCOPE on US census data.

\paragraph{Hierarchical categories.}
Often certain predictors may have levels that are effectively subdivisions of the levels of other predictors. Examples include category of item in e-commerce or geographical data with predictors for continent, countries and district. For simplicity, we will illustrate how such settings may be dealt with by considering a case with two predictors, but this may easily be generalised to more complex hierarchical structures. Suppose there is a partition $G_1 \cup \cdots \cup G_{K_1}$ of $\{1,\ldots,K_2\}$ such that for all $k=1,\ldots,K_1$,
\[
X_{i2} \in G_k \Longrightarrow X_{i1} = k,
\]
so the levels of the second predictor in $G_k$ represent subdivisions of $k$th level of the first predictor. Let $K_{2k} := |G_k|$ and let $\mbb\theta_{2k}$ refer to the subvector $(\theta_{2l})_{l \in G_k}$ for each $k=1,\ldots,K_1$, so components of $\mbb\theta_{2k}$ are the coefficients corresponding to the levels in $G_k$. Also let $\theta_{2k(r)}$ denote the $r$th order statistic within $\mbb\theta_{2k}$.
It is natural to encourage fusion among levels within $G_k$ more strongly than for levels in different elements of the partition. To do this we can modify our objective function so the penalty takes the form
\[
\sum_{k=1}^{K_1-1} \rho_1(\theta_{1(k+1)} - \theta_{1(k)}) + \sum_{k=1}^{K_1} \sum_{l=1}^{K_{2k}-1} \rho_{2k}(\theta_{2k(l+1)} - \theta_{2k(l)}).
\]
We furthermore enforce the identifiability constraints that
\[
\sum_{l=1}^{K_1} n_{1l} \theta_{1l} = 0 \;\;\;\text{ and } \;\;\;\sum_{l \in G_k} n_{2l}\theta_{2l} = 0 \text{ for all }k=1,\ldots,K.
\]
As well as yielding the desired shrinkage properties, an additional advantage of this approach is that the least squares criterion is separable in $\mbb\theta_{21}, \ldots, \mbb\theta_{2K_1}$ so the blockwise update of $\mbb\theta_2$ can be performed in parallel. This can lead to a substantial reduction in computation time if $K_2$ is large.

\section{Numerical experiments} \label{sec:numerical}
In this section we explore the empirical properties of SCOPE. We first present  results on the performance on simulated data, and then in Sections~\ref{sec:adultds1}~to~\ref{sec:covidhub} present analyses and experiments on US census data, insurance data and COVID-19 modelling data.

We denote SCOPE with a specific choice of $\gamma$ as SCOPE-$\gamma$, and write SCOPE-CV to denote SCOPE with a cross-validated choice of $\gamma$. SCOPE solutions are computed using our R \citep{rlang} package \texttt{CatReg} \citep{CatReg}, using 5-fold cross-validation to select $\lambda$ for all examples except those in Section~\ref{sec:covidhub}. We compare SCOPE to linear or logistic regression where appropriate and a range of existing methods, including CAS-ANOVA \citep{bondell2009simultaneous} \eqref{eq:all_pairs}, and an adaptive version where the weights $w_{j,kl}$ are multiplied by a factor proportional to the $|\hat\theta_{jk}^{\mathrm{init}} - \hat\theta_{jl}^{\mathrm{init}}|^{-1}$, where $\hat{\mbb\theta}^\mathrm{init}$ is an initial CAS-ANOVA estimate. 
For these methods the tuning parameter $\lambda$ was also selected by 5-fold cross-validation.
As well as this, we include Delete or merge regressors (DMR) \citep{maj2015delete} and Bayesian effect fusion (BEF) \citep{pauger2019bayesian} in some experiments. 
With the former, models were fitted using \texttt{DMRnet} \citep{DMRnet} and selected by 5-fold cross-validation where possible; otherwise an information criterion was used. With BEF, coefficients were modelled with a Gaussian mixture model with posterior mean estimated using 1000 samples using \texttt{effectFusion} \citep{effectFusion}.
We also include comparison to the tree-based approaches CART \citep{breiman1984classification} and random forests (RF) \citep{breiman2001random}. Lastly, in some experiments, models were also fitted using the Lasso \citep{tibshirani1996regression}. CART was implemented using \texttt{rpart} \citep{rpart} with pruning according to the one standard error rule. Random forests and Lasso were implemented using the default settings in \texttt{randomForest} \citep{randomForest} and \texttt{glmnet} \citep{friedman2010regularization} packages respectively.
 For full details of the specific versions of these methods and software used in the numerical experiments, see Section~\ref{sec:methoddetails} of the Supplementary material. 

\subsection{Simulations} \label{sec:simulations}
We simulated data according to the model \eqref{eq:ANOVA_mod} with the covariates $X_{ij}$ generated randomly in the following way. We first drew $(W_{ij})_{j=1}^p$ from a multivariate $\mathcal{N}_p(0, \Sigma)$ distribution where the covariance matrix $\Sigma$ had ones on the diagonal. The off-diagonal elements of $\Sigma$ were chosen such that $U_{ij} := \Phi^{-1}(W_{ij})$ had $\text{corr}(U_{ij}, U_{ik}) = \rho$ for $j \neq k$. The marginally uniform $U_{ij}$ were then quantised this to give $X_{ij} = \ceil{24 U_{ij}}$, so the number of levels $K_j=24$.

The errors $\varepsilon_i$ were independently distributed as $\mathcal{N}(0, \sigma^2)$. The performance of SCOPE and competitor methods was measured using mean squared prediction error on $10^5$ new (noiseless) observations generated in the same way as the training data, and final results are averages over 500 draws of training and test data. We considered various settings of $(n, p, \rho, \mbb\theta^0, \sigma^2)$ below with low-dimensional and high-dimensional scenarios considered in Sections~\ref{sec:lowdimsims} and \ref{sec:highdimsims} respectively. 
The coefficient vectors for each experiment are specified up to an additive constant, which is required to satisfy the identifiability condition \eqref{eq:identifiability}.

We measured predictive performance by the mean squared prediction error (MSPE) given by
\begin{equation} \label{eq:MSPE}
	\text{MSPE} := \E_x \{g(x) - \hat{g}(x)\}^2,
\end{equation}
where $g$ is the true regression function, $\hat{g}$ an estimate, and the expectation is taken over the covariate vector $x$.

\subsubsection{Low-dimensional experiments} \label{sec:lowdimsims}
Results are presented for three settings with $n=500$, $p = 10$ given below. 
\begin{enumerate}
  \item $\mbb\theta^0_j = (\overbrace{-3, \ldots, -3}^\text{10 times}, \overbrace{0,\ldots, 0}^\text{4 times}, \overbrace{3, \ldots, 3}^\text{10 times})$ for $j = 1, 2, 3$, and $\mbb\theta^0_j = 0$ otherwise; $\rho = 0$.
  \item $\mbb\theta^0_j = (\overbrace{-3, \ldots, -3}^\text{8 times}, \overbrace{0,\ldots, 0}^\text{8 times}, \overbrace{3, \ldots, 3}^\text{8 times})$ for $j = 1, 2, 3$, and $\mbb\theta^0_j = 0$ otherwise; $\rho = 0$.
  \item As Setting 1, but with $\rho = 0.8$.
\end{enumerate}
\begin{table}[h!]
\centering
\footnotesize
\setlength{\tabcolsep}{5pt}
 		\begin{tabular}{r S[table-format=2.3]@{\hskip 0.22in} S[table-format=2.3]@{\hskip 0.22in} S[table-format=2.3]@{\hskip 0.22in} S[table-format=2.3] S[table-format=2.3] S[table-format=2.3]@{\hskip 0.22in} S[table-format=2.3]@{\hskip 0.22in} S[table-format=2.3]@{\hskip 0.22in} S[table-format=2.3]}
	& \multicolumn{4}{c}{Setting 1} &  & \multicolumn{4}{c}{Setting 2}  \\ 
 $\sigma^2$: & \multicolumn{1}{c}{1} &\multicolumn{1}{c}{6.25} & \multicolumn{1}{c}{25} & \multicolumn{1}{c}{100} & & \multicolumn{1}{c}{1} &\multicolumn{1}{c}{6.25} & \multicolumn{1}{c}{25} & \multicolumn{1}{c}{100}  \\ 
	 \text{SNR}: & \multicolumn{1}{c}{4.7} & \multicolumn{1}{c}{1.9} & \multicolumn{1}{c}{0.95} & \multicolumn{1}{c}{0.47} && \multicolumn{1}{c}{4.2} & \multicolumn{1}{c}{1.7} & \multicolumn{1}{c}{0.85} & \multicolumn{1}{c}{0.42} \\\hline \hline 
	SCOPE-8  & $\mathbf{0.014}$\rlap{\tiny{(0.0)}}  & 0.450\rlap{\tiny{(0.5)}} & 4.571\rlap{\tiny{(1.0)}} & 12.936\rlap{\tiny{(2.8)}} & & $\mathbf{0.015}$\rlap{\tiny{(0.0)}} & $\mathbf{0.285}$\rlap{\tiny{(0.3)}} & 6.775\rlap{\tiny{(0.9)}} & 12.697\rlap{\tiny{(2.3)}} \\ \hline
	SCOPE-32  & 0.018\rlap{\tiny{(0.0)}} & 0.878\rlap{\tiny{(0.6)}} & 4.151\rlap{\tiny{(0.9)}} & $\mathbf{12.356}$\rlap{\tiny{(2.1)}} & & 0.019\rlap{\tiny{(0.0)}} & 0.655\rlap{\tiny{(0.4)}} & 5.026\rlap{\tiny{(1.0)}} & $\mathbf{12.037}$\rlap{\tiny{(2.0)}}   \\ \hline
	SCOPE-CV &  0.015\rlap{\tiny{(0.0)}} & $\mathbf{0.407}$\rlap{\tiny{(0.4)}} & $\mathbf{4.120}$\rlap{\tiny{(0.9)}} & 12.513\rlap{\tiny{(2.5)}} & & 0.016\rlap{\tiny{(0.0)}} & 0.292\rlap{\tiny{(0.3)}} & $\mathbf{5.005}$\rlap{\tiny{(1.1)}} & 12.444\rlap{\tiny{(2.5)}}\\ \hline
	Linear regression  & 0.851\rlap{\tiny{(0.1)}} & 5.317\rlap{\tiny{(0.7)}} & 21.503\rlap{\tiny{(2.7)}} & 86.745\rlap{\tiny{(10.7)}} & & 0.869\rlap{\tiny{(0.1)}} & 5.406\rlap{\tiny{(0.7)}} & 21.216\rlap{\tiny{(2.5)}}  & 85.439\rlap{\tiny{(10.9)}}   \\ \hline
		Oracle least squares  & 0.014\rlap{\tiny{(0.0)}} & 0.091\rlap{\tiny{(0.1)}} & 0.333\rlap{\tiny{(0.2)}} & 1.405\rlap{\tiny{(0.8)}} & & 0.014\rlap{\tiny{(0.0)}} & 0.088\rlap{\tiny{(0.0)}} & 0.336\rlap{\tiny{(0.2)}} & 1.532\rlap{\tiny{(0.8)}} \\ \hline
	CAS-ANOVA  & 0.617\rlap{\tiny{(0.3)}} & 1.602\rlap{\tiny{(0.3)}} & 5.448\rlap{\tiny{(1.0)}} & 14.814\rlap{\tiny{(2.2)}} & & 1.483\rlap{\tiny{(0.4)}} & 1.626\rlap{\tiny{(0.3)}} & 5.466\rlap{\tiny{(1.0)}} & 13.421\rlap{\tiny{(2.2)}}   \\ \hline
	Adaptive CAS-ANOVA  & 0.135\rlap{\tiny{(0.1)}} & 0.880\rlap{\tiny{(0.4)}} & 5.076\rlap{\tiny{(1.2)}} & 22.896\rlap{\tiny{(4.7)}} & & 0.134\rlap{\tiny{(0.1)}} & 0.912\rlap{\tiny{(0.3)}} & 5.535\rlap{\tiny{(1.2)}} & 22.213\rlap{\tiny{(4.9)}}  \\ \hline
		DMR & $\mathbf{0.014}$\rlap{\tiny{(0.0)}} & 0.448\rlap{\tiny{(0.4)}} & 4.884\rlap{\tiny{(1.4)}} & 18.394\rlap{\tiny{(3.6)}} & & 0.016\rlap{\tiny{(0.0)}} & 0.409\rlap{\tiny{(0.4)}} & 6.430\rlap{\tiny{(1.4)}} & 17.457\rlap{\tiny{(2.1)}}\\ \hline
	BEF & 0.020\rlap{\tiny{(0.0)}} & 2.209\rlap{\tiny{(1.1)}} & 6.297\rlap{\tiny{(1.8)}} & 21.927\rlap{\tiny{(2.3)}} & & 0.019\rlap{\tiny{(0.0)}} & 1.055\rlap{\tiny{(0.9)}} & 8.183\rlap{\tiny{(2.0)}} & 18.236\rlap{\tiny{(1.5)}} \\ \hline
		CART & 3.844\rlap{\tiny{(0.4)}} & 5.099\rlap{\tiny{(0.9)}} & 13.219\rlap{\tiny{(2.1)}} & 22.431\rlap{\tiny{(1.2)}} & & 5.530\rlap{\tiny{(0.6)}} & 7.457\rlap{\tiny{(0.9)}} & 13.280\rlap{\tiny{(1.8)}} & 18.198\rlap{\tiny{(0.7)}}  \\ \hline
	RF & 9.621\rlap{\tiny{(0.5)}} & 10.944\rlap{\tiny{(0.5)}} & 13.217\rlap{\tiny{(0.7)}} & 16.344\rlap{\tiny{(0.9)}} & & 8.947\rlap{\tiny{(0.3)}} & 9.747\rlap{\tiny{(0.4)}} & 11.249\rlap{\tiny{(0.6)}} & 13.646\rlap{\tiny{(0.8)}}
 	\end{tabular}

	\begin{tabular}{r  S[table-format=2.3]@{\hskip 0.22in} S[table-format=2.3]@{\hskip 0.22in} S[table-format=2.3]@{\hskip 0.22in} S[table-format=2.3] }
	&  \multicolumn{4}{c}{Setting 3} \\ 
	  $\sigma^2$: & \multicolumn{1}{c}{1} &\multicolumn{1}{c}{6.25} & \multicolumn{1}{c}{25} & \multicolumn{1}{c}{100}  \\ 
	 \text{SNR}: & \multicolumn{1}{c}{7.3} & \multicolumn{1}{c}{2.9} & \multicolumn{1}{c}{1.5} & \multicolumn{1}{c}{0.73} \\ \hline \hline 
		SCOPE-8  & $\mathbf{0.015}$\rlap{\tiny{(0.0)}} & 0.967\rlap{\tiny{(0.7)}} & 5.060\rlap{\tiny{(1.3)}} & 14.555\rlap{\tiny{(2.9)}}  \\ \hline
	SCOPE-32  & 0.018\rlap{\tiny{(0.0)}} & 0.713\rlap{\tiny{(0.4)}} & 3.580\rlap{\tiny{(0.8)}} & $\mathbf{9.721}$\rlap{\tiny{(1.9)}}  \\ \hline
	SCOPE-CV  & 0.022\rlap{\tiny{(0.1)}} & 0.582\rlap{\tiny{(0.3)}} & $\mathbf{3.368}$\rlap{\tiny{(0.9)}} & 10.168\rlap{\tiny{(2.6)}} \\ \hline
	Linear regression   & 0.879\rlap{\tiny{(0.1)}} & 5.485\rlap{\tiny{(0.7)}} & 21.987\rlap{\tiny{(2.7)}} & 87.820\rlap{\tiny{(11.9)}}   \\ \hline
			Oracle least squares  &0.014\rlap{\tiny{(0.0)}} & 0.092\rlap{\tiny{(0.0)}} & 0.362\rlap{\tiny{(0.2)}} & 1.488\rlap{\tiny{(1.0)}} \\ \hline
	CAS-ANOVA   & 0.710\rlap{\tiny{(0.2)}} & 1.601\rlap{\tiny{(0.3)}} & 4.732\rlap{\tiny{(0.9)}} & 12.708\rlap{\tiny{(2.1)}}  \\ \hline
	Adaptive CAS-ANOVA  & 0.189\rlap{\tiny{(0.2)}} & 0.701\rlap{\tiny{(0.3)}} & 3.705\rlap{\tiny{(1.0)}} & 16.186\rlap{\tiny{(3.6)}}  \\ \hline
		DMR &  $\mathbf{0.015}$\rlap{\tiny{(0.0)}} & $\mathbf{0.553}$\rlap{\tiny{(0.5)}} & 5.730\rlap{\tiny{(1.9)}} & 18.594\rlap{\tiny{(4.5)}} \\ \hline
	BEF   & 0.019\rlap{\tiny{(0.0)}} & 1.716\rlap{\tiny{(0.9)}} & 8.143\rlap{\tiny{(2.6)}} & 26.923\rlap{\tiny{(7.0)}} \\ \hline
	CART &  4.336\rlap{\tiny{(0.6)}} & 5.685\rlap{\tiny{(1.0)}} & 9.910\rlap{\tiny{(1.7)}} & 18.543\rlap{\tiny{(2.2)}} \\ \hline
	RF   & 4.039\rlap{\tiny{(0.3)}} & 5.673\rlap{\tiny{(0.5)}} & 9.157\rlap{\tiny{(0.9)}} & 13.766\rlap{\tiny{(1.7)}}
 	\end{tabular}
 	\caption{Mean squared prediction errors (and standard deviations thereof) of various methods on the settings described. }

\label{tab:simmspetable}
\end{table}

Each of these experiments were performed with noise variance $\sigma^2 = $ 1, 6.25, 25 and 100. Note that the variance of the signal varies across each setting, and signal-to-noise ratio (SNR) for each experiment is displayed in Table~\ref{tab:simmspetable}.
Methods included for comparison were SCOPE-8, SCOPE-32, SCOPE-CV, linear regression, vanilla and adaptive CAS-ANOVA, DMR, Bayesian effect fusion, CART and random forests. Also included are the results from the oracle least squares estimator \eqref{eq:oraclelse}.

Results are shown in Table~\ref{tab:simmspetable} and further details are given in Section~\ref{sec:lowdimsims2} of the Supplementary material.
Across all experiments, SCOPE with a cross-validated choice of $\gamma$ exhibits prediction performance at least as good as the optimal approaches, and in all but the lowest noise settings performs better than the other methods that were included. In these exceptions, we see that fixing $\gamma$ to be a small value (corresponding to high-concavity) provides leading performance. 

In these low noise settings, we see that the methods based on first estimating the clusterings of the levels and then estimating the coefficients without introducing further shrinkage, such as DMR or Bayesian effect Fusion, perform well. However they tend to struggle when the noise is larger.
In contrast the tree-based methods perform poorly in low noise settings but exhibit competitive performance in high noise settings.

\FloatBarrier

\subsubsection{High-dimensional experiments} \label{sec:highdimsims}
We considered 8 settings as detailed below, each with $n = 500$, $p = 100$ and simulated 500 times. 

\begin{enumerate}
\item $\mbb\theta^0_j = (\overbrace{-2, \ldots, -2}^\text{8 times}, \overbrace{0,\ldots, 0}^\text{8 times}, \overbrace{2, \ldots, 2}^\text{8 times})$ for $j = 1, 2, 3$, $\mbb\theta^0_j = (\overbrace{-2, \ldots, -2}^\text{10 times}, \overbrace{0,\ldots, 0}^\text{4 times}, \overbrace{2, \ldots, 2}^\text{10 times})$ for $j = 4, 5, 6$, and $\mbb\theta^0_j = 0$ otherwise; $\rho = 0$ and $\sigma^2 = 50$.
  \item As Setting 1, but with $\rho = 0.5$.
  \item $\mbb\theta^0_j = (\overbrace{-2, \ldots, -2}^\text{8 times}, \overbrace{0,\ldots, 0}^\text{8 times}, \overbrace{2, \ldots, 2}^\text{8 times})$ for $j = 1, 2, 3$, $\mbb\theta^0_j = (\overbrace{-2, \ldots, -2}^\text{16 times}, \overbrace{3,\ldots, 3}^\text{8 times})$ for $j = 4, 5, 6$, and $\mbb\theta^0_j = 0$ otherwise; $\rho = 0.5$ and $\sigma^2 = 100$.
\item $\mbb\theta^0_j = (\overbrace{-2, \ldots, -2}^\text{5 times}, \overbrace{-1,\ldots, -1}^\text{5 times}, \overbrace{0, \ldots, 0}^\text{4 times}, \overbrace{1, \ldots, 1}^\text{5 times}, \overbrace{2, \ldots, 2}^\text{5 times})$ for $j = 1, \ldots, 5$, and $\mbb\theta^0_j = 0$ otherwise; $\rho = 0$ and $\sigma^2 = 25$.
\item $\mbb\theta^0_j = (\overbrace{-2, \ldots, -2}^\text{16 times}, \overbrace{3, \ldots, 3}^\text{8 times})$ for $j = 1, \ldots , 25$, and $\mbb\theta^0_j = 0$ otherwise; $\rho = 0$ and $\sigma^2 = 1$.
  \item As Setting 5, but with $\rho = 0.5$.
  \item $\mbb\theta^0_j = (\overbrace{-2, \ldots, -2}^\text{4 times}, \overbrace{0,\ldots, 0}^\text{12 times}, \overbrace{2, \ldots, 2}^\text{8 times})$ for $j = 1, \ldots, 10$, and $\mbb\theta^0_j = 0$ otherwise; $\rho = 0$ and $\sigma^2 = 25$.
  \item $\mbb\theta^0_j = (\overbrace{-3, \ldots, -3}^\text{6 times}, \overbrace{-1,\ldots, -1}^\text{6 times}, \overbrace{1,\ldots, 1}^\text{6 times}, \overbrace{3, \ldots, 3}^\text{6 times})$ for $j = 1, \ldots, 5$, and $\mbb\theta^0_j = 0$ otherwise; $\rho = 0$ and $\sigma^2 = 25$.
\end{enumerate}
Models were fitted using SCOPE-8, SCOPE-32, SCOPE-CV, DMR, CART, Random forests and the Lasso.  Table~\ref{tab:hdsimmspetable} gives the mean squared prediction errors across each of the settings.

As well as prediction performance, it is interesting to see how the methods perform in terms of variable selection performance. With categorical covariates, there are two potential ways of evaluating this. The first is to consider the number of false positives and false negatives across the $p=100$ categorical variables, defining a variable $j$ to have been selected if $\hat{\mbb\theta}_j \neq 0$. These results are shown in Table~\ref{tab:hdsimfpfn}. This definition of a false positive can be considered quite conservative; typically one can find that often the false signal variables have only two levels, each with quite small coefficients. This means that the false positive rate can increase substantially with only a small increase in the dimension of the estimated linear model.

The second is to see within the signal variables (i.e., the $j$ for which $\mbb\theta^0_j \neq 0$), how closely the estimated clustering resembles the true structure. To quantify this, we use the \emph{adjusted Rand index} \citep{hubert1985comparing}. This is the proportion of all pairs of observations that are either (i) in different true clusters and different estimated clusters, or (ii) in the same true cluster and estimated cluster; this is then corrected to ensure that its value is zero when exactly one of the clusterings is `all-in-one'. In Table~\ref{tab:hdsimrand} we report the average adjusted Rand index over the true signal variables in each setting.

\begin{table}[h!]
\centering
\footnotesize
 		\begin{tabular}{r S[table-format=2.3]@{\hskip 0.26in} S[table-format=2.3]@{\hskip 0.26in} S[table-format=2.3]@{\hskip 0.26in} S[table-format=2.3]@{\hskip 0.28in} S[table-format=2.3]@{\hskip 0.28in} S[table-format=2.3]@{\hskip 0.26in} S[table-format=2.3]@{\hskip 0.26in} S[table-format=2.3]}
	 Setting: & \multicolumn{1}{c}{1} & \multicolumn{1}{c}{2} & \multicolumn{1}{c}{3} & \multicolumn{1}{c}{4} & \multicolumn{1}{c}{5} &  \multicolumn{1}{c}{6} & \multicolumn{1}{c}{7} & \multicolumn{1}{c}{8}   \\ 
	\text{SNR}: & \multicolumn{1}{c}{0.6} & \multicolumn{1}{c}{1.0} & \multicolumn{1}{c}{1.0} & \multicolumn{1}{c}{0.64} & \multicolumn{1}{c}{12} & \multicolumn{1}{c}{36} & \multicolumn{1}{c}{0.87} & \multicolumn{1}{c}{1.0} \\\hline \hline 
			SCOPE-8  & 14.319\rlap{\tiny{(2.0)}} & 15.445\rlap{\tiny{(2.9)}} & 30.597\rlap{\tiny{(5.6)}} & 7.254\rlap{\tiny{(1.2)}} & 96.538\rlap{\tiny{(25.0)}} & 7.960\rlap{\tiny{(23.2)}} & 15.867\rlap{\tiny{(1.4)}} & 11.028\rlap{\tiny{(1.6)}} \\ \hline
	SCOPE-32  & $\mathbf{14.009}$\rlap{\tiny{(1.6)}} & $\mathbf{10.780}$\rlap{\tiny{(1.6)}} & $\mathbf{21.841}$\rlap{\tiny{(3.4)}} & 7.256\rlap{\tiny{(0.9)}} & 65.344\rlap{\tiny{(13.4)}} & 0.107\rlap{\tiny{(0.0)}} & 14.867\rlap{\tiny{(1.2)}} & 11.218\rlap{\tiny{(1.4)}}   \\ \hline
	SCOPE-CV &  14.026\rlap{\tiny{(1.7)}} & 10.843\rlap{\tiny{(1.8)}} & 22.004\rlap{\tiny{(3.9)}} & $\mathbf{7.191}$\rlap{\tiny{(1.0)}} & $\mathbf{54.030}$\rlap{\tiny{(19.2)}} & $\mathbf{0.084}$\rlap{\tiny{(0.0)}} & $\mathbf{14.865}$\rlap{\tiny{(1.3)}} & $\mathbf{10.941}$\rlap{\tiny{(1.5)}} \\ \hline
		Oracle LSE & 5.044\rlap{\tiny{(0.6)}} & 5.130\rlap{\tiny{(0.6)}} & 2.664\rlap{\tiny{(1.0)}} & 1.09\rlap{\tiny{(0.3)}} & 0.054\rlap{\tiny{(0.0)}} & 0.055\rlap{\tiny{(0.0)}} & 1.087\rlap{\tiny{(0.3)}} & 0.799\rlap{\tiny{(0.3)}} \\ \hline
		DMR & 18.199\rlap{\tiny{(1.4)}} & 22.627\rlap{\tiny{(4.4)}} & 42.979\rlap{\tiny{(9.2)}} & 9.645\rlap{\tiny{(1.2)}} & 139.095\rlap{\tiny{(4.3)}} & 213.691\rlap{\tiny{(35.7)}} & 19.298\rlap{\tiny{(0.8)}} & 11.737\rlap{\tiny{(2.4)}} \\ \hline
	CART & 18.146\rlap{\tiny{(0.5)}} & 31.235\rlap{\tiny{(3.6)}} & 58.73\rlap{\tiny{(6.6)}} & 10.466\rlap{\tiny{(0.3)}} & 139.35\rlap{\tiny{(2.1)}} & 614.739\rlap{\tiny{(42.8)}} & 19.021\rlap{\tiny{(0.4)}} & 23.775\rlap{\tiny{(1.5)}}  \\ \hline
	RF & 16.181\rlap{\tiny{(0.6)}} & 16.345\rlap{\tiny{(1.4)}} & 31.561\rlap{\tiny{(2.6)}} & 9.053\rlap{\tiny{(0.4)}} & 128.618\rlap{\tiny{(2.2)}} & 264.374\rlap{\tiny{(14.4)}} & 17.224\rlap{\tiny{(0.4)}} & 19.783\rlap{\tiny{(0.7)}} \\ \hline
	Lasso & 18.136\rlap{\tiny{(0.5)}} & 24.839\rlap{\tiny{(1.3)}} & 48.162\rlap{\tiny{(2.5)}} & 10.473\rlap{\tiny{(0.4)}} & 135.375\rlap{\tiny{(5.0)}} & 154.656\rlap{\tiny{(7.8)}} & 18.886\rlap{\tiny{(0.6)}} & 23.813\rlap{\tiny{(1.6)}}
 	\end{tabular}
 	 	\caption{Mean squared prediction errors (and standard deviations thereof) of each of the methods in the $8$ high-dimensional settings considered.} \label{tab:hdsimmspetable}
\end{table}

\begin{table}[h!]
\centering
\footnotesize
 		\begin{tabular}{r c c c c c c c c}
	Setting: & 1 & 2 & 3 & 4 & 5 & 6 & 7 & 8 \\ \hline \hline 
SCOPE-8 & 0.02/0.35 & 0.04/0.23 & 0.04/0.25 & 0.02/0.15 & 0.02/0.23 & 0.02/0.01 & 0.02/0.35 & 0.01/0.00 \\ \hline
SCOPE-32 & 0.14/0.15 & 0.30/0.02 & 0.30/0.02 & 0.15/0.04 & 0.52/0.00 & 0.00/0.00 & 0.21/0.08 & 0.21/0.00 \\ \hline  
SCOPE-CV & 0.12/0.20 & 0.30/0.02 & 0.29/0.03 & 0.12/0.07 & 0.59/0.00 & 0.00/0.00 & 0.21/0.11 & 0.09/0.00 \\ \hline
DMR & 0.00/0.86 & 0.00/0.44 & 0.00/0.47 & 0.00/0.62 & 0.00/0.91 & 0.03/0.60 & 0.00/0.88 & 0.00/0.02 \\ \hline
Lasso & 0.01/0.88 & 0.00/1.00 & 0.00/1.00 & 0.01/0.83 & 0.00/0.98 & 0.00/1.00 & 0.00/0.91 & 0.00/0.90  
 	\end{tabular}
 	 	\caption{(False positive rate)/(False negative rate) of linear modelling methods considered in the high-dimensional settings.} \label{tab:hdsimfpfn}
\end{table}

\begin{table}[h!]
\centering
\footnotesize
 		\begin{tabular}{r S[table-format=1.2] S[table-format=1.2] S[table-format=1.2] S[table-format=1.2] S[table-format=1.2] S[table-format=1.2] S[table-format=1.2] S[table-format=1.2] }
	Setting: & \multicolumn{1}{c}{1} & \multicolumn{1}{c}{2} & \multicolumn{1}{c}{3} & \multicolumn{1}{c}{4} & \multicolumn{1}{c}{5} &  \multicolumn{1}{c}{6} & \multicolumn{1}{c}{7} & \multicolumn{1}{c}{8}   \\\hline \hline 
			SCOPE-8  & 0.23 & 0.36 & 0.38 & 0.15 & 0.39 & 0.96 & 0.13 & 0.29 \\ \hline
			SCOPE-32  & 0.29 & 0.46 & 0.48 & 0.19 & 0.56 & 1.00 & 0.17 & 0.34 \\ \hline 
			SCOPE-CV  & 0.27 & 0.45 & 0.46 & 0.18 & 0.56 & 1.00 & 0.17 & 0.31 \\ \hline 
			DMR  & 0.04 & 0.20 & 0.23 & 0.06 & 0.04 & 0.19 & 0.03 & 0.28 \\ \hline
			Lasso  & 0.00 & 0.00 & 0.00 & 0.00 & 0.00 & 0.00 & 0.00 & 0.00
 	\end{tabular}
 	 	\caption{Average adjusted Rand index among true signal variables for the high-dimensional settings.} \label{tab:hdsimrand}
\end{table}

Further details can be found in Section~\ref{sec:highdimsims2} of the Supplementary material. In particular we include a table with the distribution of cross-validated choices of $\gamma$ (from a grid $\{4,8,16,32,64\}$) for each experimental setting. 
Note that a choice of $\gamma = 4$ is close to the setting of $\gamma = 3$ recommended in \cite{zhang2010nearly}, though the problem of categorical covariates is very different in nature than the vanilla variable selection problem considered there. Our results there suggest that for SCOPE, a larger value of $\gamma$ is preferable across a range of settings, which is also visible in the comparison between $\gamma = 8$ and $\gamma = 32$ in Table~\ref{tab:hdsimmspetable}.

Across all the settings in this study, SCOPE performs better than any of the other methods included. 
This is regardless of which of the three $\gamma$ regimes is chosen, although cross-validating $\gamma$ gives the strongest performance overall.
Comparing the results for $\gamma = 8$ and $\gamma = 32$ suggests that a larger (low-concavity) choice of $\gamma$ is preferable for higher-dimensional settings.
In setting 6, we see from Tables~\ref{tab:hdsimfpfn}~and~\ref{tab:hdsimrand} that SCOPE obtains the true underlying groupings of the coefficients and obtains the oracle least-squares estimate in every case, giving these striking results. This is also achieved for some of the experiments in setting 5. In contrast, DMR, which initially applies a group Lasso \citep{yuan2006model} to screen the categorical variables and give a low-dimensional model, necessarily misses some signal variables in this first stage and hence struggles here.

\FloatBarrier
\subsection{Adult dataset analysis} \label{sec:adultds1}
The \textit{Adult dataset}, available from the UCI Machine Learning Repository \citep{dua2019}, contains a sample of $45\,222$ observations based on information from the 1994 US census. The binary response variable is $0$ if the individual earns at most \$50\,000 a year, and 1 otherwise.
There are $2$ continuous and 8 categorical variables; some such as `native country' have large numbers of levels, bringing the total dimension to 93. 
An advantage of using SCOPE here over black-box predictive tools such as Random forests is the interpretability of the fitted model.

In Table~\ref{fig:adultfullmodel}, we show the 25-dimensional fitted model.
Within the Education category, we see that six distinct levels have been identified.
These agree almost exactly with the stratification one would expect, with all school dropouts before 12th grade being grouped together at the lowest level.
\begin{table}[h!]
\centering
\footnotesize
	\begin{tabular}{r S[table-format=1.3] l}
	Variable & \multicolumn{1}{c}{Coefficient} & Levels \\ \hline \hline
	Intercept & -3.048 & -- \\ \hline
	Age & 0.027 & -- \\ \hline
	Hours per week & 0.029 & -- \\ \hline
	Work class & 0.378 & Federal government, Self-employed (incorporated) \\ 
	& 0.058 & Private \\
	& -0.143 & Local government \\
	& -0.434 & Self-employed (not incorporated), State government, Without pay \\ \hline
	Education level  & 1.691 & Doctorate, Professional school \\ 
	& 1.023 & Master's \\
	& 0.646 & Bachelor's \\
	 & -0.132 & Associate's (academic), Associate's (vocational), Some college (non-graduate) \\ 
	& -0.546 & 12th, High school grad \\
	& -1.539 & Preschool, 1st-4th, 5th-6th, 7th-8th, 9th, 10th, 11th \\\hline
	Marital status & 0.059	& Divorced, Married (armed forces spouse), Married (civilian spouse), Married \\ && (absent spouse), Separated, Widowed \\ 
	& -0.476 & Never married \\ \hline
	Occupation & 0.560 & Executive/Managerial \\
	& 0.311 & Professional/Specialty, Protective service, Tech support \\
	& -0.003 & Armed forces, Sales \\
	& -0.168 & Admin/Clerical, Craft/Repair \\
	& -0.443 & Machine operative/inspector, Transport \\ 
	& -1.107 & Farming/Fishing, Handler/Cleaner, Other service, Private house servant \\\hline
	Relationship* & 1.498 & Wife \\ 
	 & 0.332 & Husband \\
	& -1.220 & Not in family \\
	& -1.482 & Unmarried, Other relative \\
	 & -2.144 & Own child \\ \hline
	Race & 0.013 & White \\
	& 0.008	& Asian/Pacific islander, Other \\ 
	& -0.182 & Native-American/Inuit, Black \\  \hline
	Sex 	& 0.139 & Male \\ 
	& -0.619 & Female \\ \hline
	Native country & 0.018 & KH, CA, CU, ENG, FR, DE, GR, HT, HN, HK, HU, IN, IR, IE, IT, JM,  \\ && JP, PH, PL, PT, PR, TW, US, YU \\
	 & -0.882 & CN, CO, DO, EC, SV, GT, NL, LA, MX, NI, GU-VI-etc, PE, SCT, ZA,\\ &&  TH, TT, VN 
	
	\end{tabular}
	
\caption{Coefficients of SCOPE model trained on the full dataset. Here, $\gamma = 100$ and $\lambda$ was selected by 5-fold cross-validation (with cross-validation error of 16.82\%). Countries, aside from those in the UK, are referred to by their (possibly historical) internet top-level domains. \qquad *Relation with which the subject lives.}
\label{fig:adultfullmodel}
\end{table}

\begin{figure}[h!]
\centering
\begin{minipage}[b]{\textwidth}
\includegraphics[width=\textwidth]{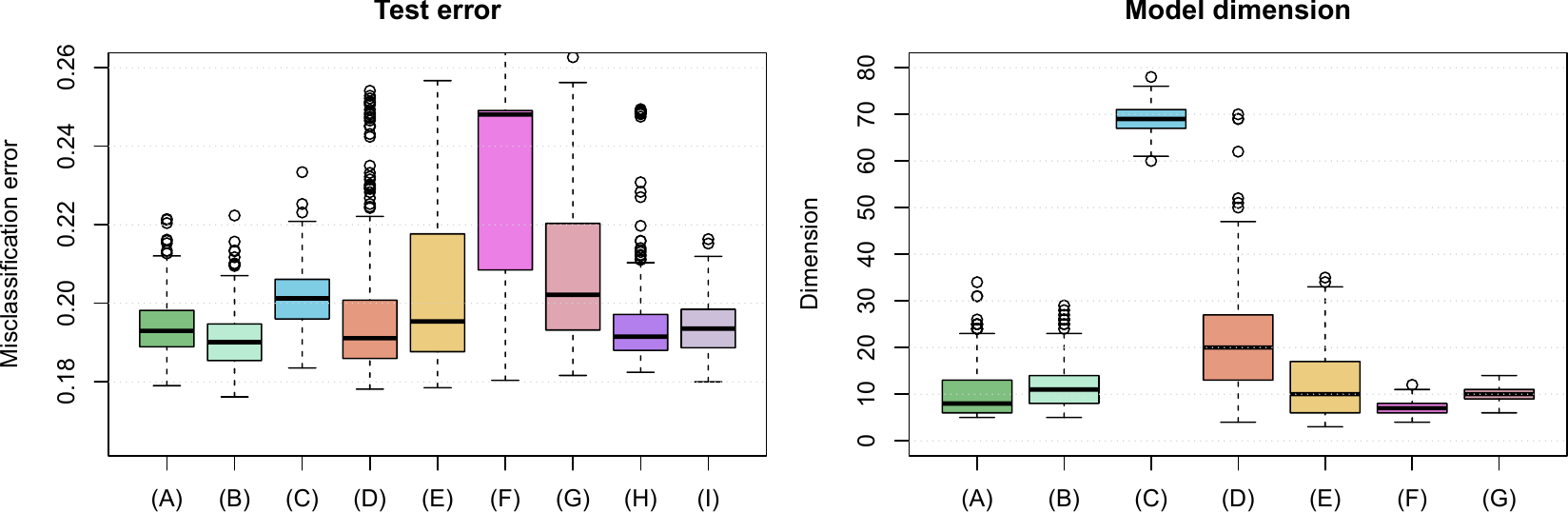}	
\caption{Prediction performance and fitted model dimension (respectively) of various methods on the Adult dataset: (A) SCOPE-100; (B) SCOPE-250; (C) Logistic regression; (D) CAS-ANOVA; (E) Adaptive CAS-ANOVA; (F) DMR; (G) BEF; (H) CART; (I) RF. \label{fig:adult_boxplots}} 
\end{minipage}
\end{figure}

\begin{table}[h!] 
\centering
\footnotesize
	\begin{tabular}{r S[table-format=0.3] S[table-format=2.1] S[table-format=4.2]}
	Method & \multicolumn{1}{c}{Misclassification error} & \multicolumn{1}{c}{Model dimension}  & \multicolumn{1}{c}{Computation time (s)} \\ \hline\hline
		SCOPE-100 & 0.194 & 10.5 & 467 \\ \hline
	SCOPE-250 & $\mathbf{0.191}$ & 11.8 & 450 \\ \hline
	Logistic regression & 0.202 & 68.9 & 0.04\\ \hline 
	CAS-ANOVA & 0.198 & 21.5 & 429 \\ \hline 
	Adaptive CAS-ANOVA & 0.205 & 11.7 & 8757 \\ \hline
	DMR & 0.235 & 6.9 & 11 \\ \hline
	BEF & 0.207 & 9.8 & 1713 \\ \hline
		CART & 0.196 &  & 0.01\\ \hline
	RF & 0.194 &  & 0.14 \\ 
 	\end{tabular}
 	\caption{Results of experiments on the Adult dataset.}
 	\label{tab:adultresults}
\end{table}
Here we assess performance in the challenging setting when the training set is quite small by randomly selecting 1\% (452) of the total observations for training, and using the remainder as a test set. Any observations containing levels not in the training set were removed.
Models were fitted with SCOPE-100, SCOPE-250, logistic regression, vanilla and adaptive CAS-ANOVA, DMR, Bayesian effect fusion, CART and random forests.

We see that both SCOPE-100 and SCOPE-250 are competitive, with CART and Random forests also performing well, though the latter two include interactions in their fits. CAS-ANOVA also performs fairly well, the misclassification error is larger that for both versions of SCOPE, and the average fitted model size is larger.

\FloatBarrier

\subsection{Adult dataset with artificially split levels} \label{sec:adultds2}

\begin{figure}[h!]
\centering
\begin{minipage}[b]{0.45\textwidth}
	\includegraphics[width=\textwidth]{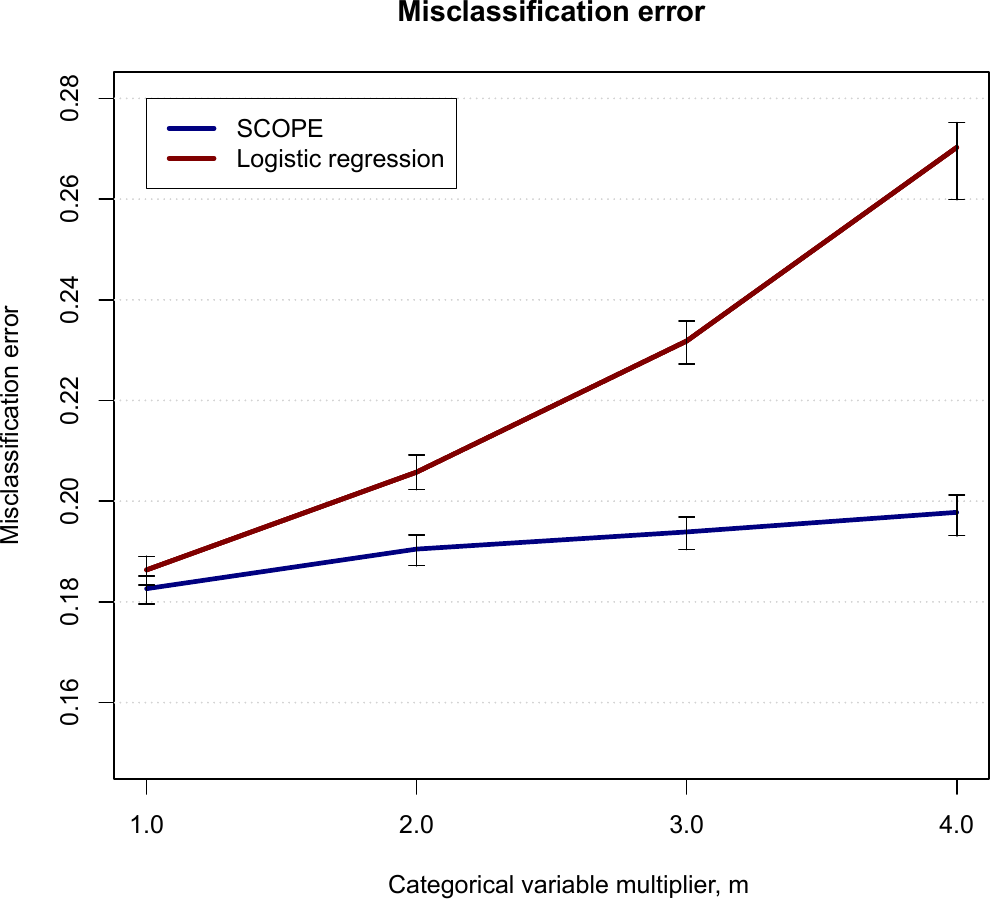}
\end{minipage}
\begin{minipage}[b]{0.45\textwidth}
	\includegraphics[width=\textwidth]{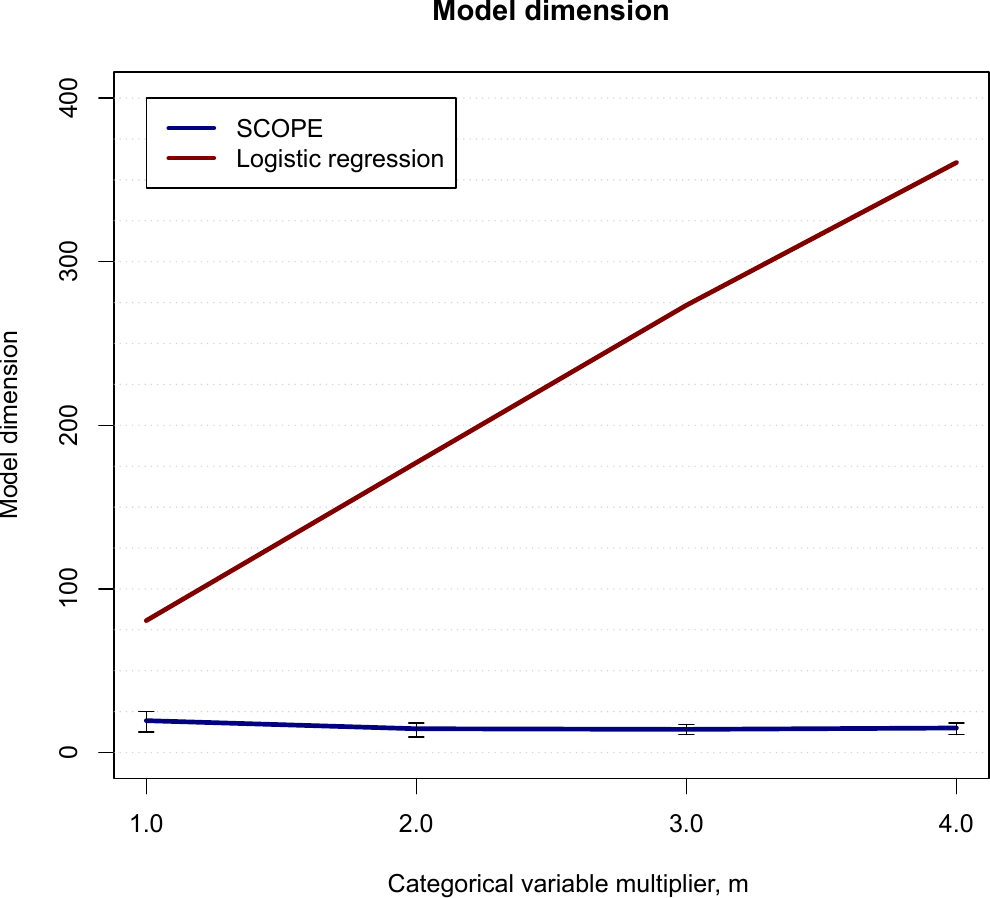}
\end{minipage}
\caption{Misclassification error and dimensions of models fitted on a sample of the \emph{Adult dataset} when levels have been artificially split $m$ times.}
\label{fig:categorynoisediagrams}
\end{figure}

To create a more challenging example, we artificially created additional levels in the \emph{Adult dataset} as follows. For each categorical variable we recursively selected a level with probability proportional to its prevalence in the data and then split it into two by appending ``-0'' or ``-1'' to the level for each observation independently and with equal probabilities. We repeated this until the total number of levels reached $m$ times the original number of levels for that variable for $m=2,3,4$. This process simulates for example responses to a survey, where different respondents might answer `US', `U.S.', `USA', `U.S.A.', `United States' or `United States of America' to a question, which would naively all be treated as different answers.

We used 2.5\% (1130) of the observations for training and the remainder for testing and applied SCOPE with $\gamma = 100$ and logistic regression. Results were averaged over 250 training and test splits. Figure~\ref{fig:categorynoisediagrams} shows that as the number of levels increases, the misclassification error of SCOPE increases only slightly and the fitted model dimension remains almost unchanged, whereas both increase with $m$ for logistic regression.

\subsection{Insurance data example} \label{sec:prudential}
The Prudential Life Insurance Assessment challenge was a prediction competition run on \citet{prudent}. By more accurately predicting risk, the burden of extensive tests and check-ups for life insurance policyholders could potentially be reduced. For this experiment, we use the training set that was provided for entrants of the competition.
\begin{figure}[h!]
\centering
\begin{minipage}[t]{0.7\textwidth}
\includegraphics[width=\textwidth]{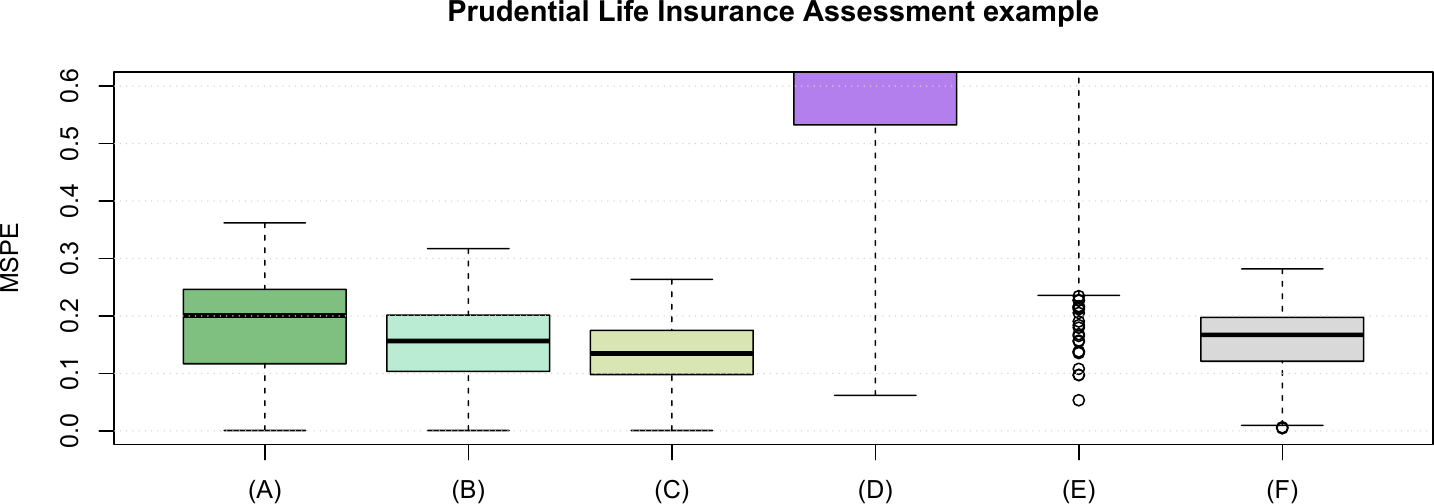}
\end{minipage}
\begin{minipage}[t]{0.25\textwidth}
\vspace{-3.3cm}
\footnotesize
	\begin{tabular}{r S[table-format=1.3] }
	\text{Method} & \text{MSPE} \\ \hline \hline 
		\text{SCOPE-8} & 0.265 \\ \hline
		\text{SCOPE-32} & 0.244 \\ \hline
		\text{SCOPE-CV} & $\mathbf{0.211}$\\ \hline
		\text{CART} & 0.851 \\ \hline
		 \text{RF} & 1.120 \\ \hline
		 \text{Lasso} & 0.244 \\ 
	\end{tabular}
\end{minipage}
\caption{Mean squared prediction error on the example based on the Prudential Life Insurance Assessment dataset. Methods used are: (A) SCOPE-8; (B) SCOPE-32; (C) SCOPE-CV; (D) CART; (E) RF; (F) Lasso.  \label{fig:prudential_results}}
\end{figure}

We removed a small number of variables due to excessive missingness, leaving 5 continuous variables and $108$ categorical variables, most with 2 or 3 levels but with some in the hundreds (and the largest with 579 levels). Rather than using the response from the original dataset, which is ordinal, 
to better suit the regression setting we are primarily concerned with in this work, we artificially generated a continuous response. To construct this signal, firstly $10$ of the categorical variables were selected at random, with probability proportional to the number of levels. For the $j$th of these, writing $K_j$ for the number of levels, we set $s_j := \lfloor 2 + \frac{1}{2} \log K_j \rfloor$ and assigned each level a coefficient in $1, \ldots, s_j$ uniformly at random, thus yielding $s_j$ true levels.
The coefficients for the $5$ continuous covariates were generated as draws from $\mathcal{N}_5(0, I)$.
The response was then scaled to have unit variance, after which standard normal noise was added.

We used 10\% ($n=5938$) of the $59\,381$ total number of observations for training, and the remainder to compute an estimated MSPE \eqref{eq:MSPE} by taking an average over these observations. We repeated this $1000$ times, sampling 10\% of the observations and generating the coefficients as above anew in each repetition. The average
mean squared prediction errors achieved by the various methods under comparison are given in Figure~\ref{fig:prudential_results}. 
We see that SCOPE with a cross-validated choice of $\gamma$ performs best, followed by the Lasso and SCOPE-32.

\subsection{COVID-19 Forecast Hub example} \label{sec:covidhub}
As well as the prediction performance experiments in the rest of this section, we include an exploratory data analysis example based on data relating to the ongoing (at time of writing) global COVID-19 pandemic. The \citet{covid19} \emph{`\ldots serves as a central repository of forecasts and predictions from over 50 international research groups.'} A collection of different research groups publish forecasts every week of case incidence in each US state for some number of weeks into the future. 

\begin{figure}[h!] 
\centering
\begin{minipage}[b]{\textwidth}
	\includegraphics[width=\textwidth]{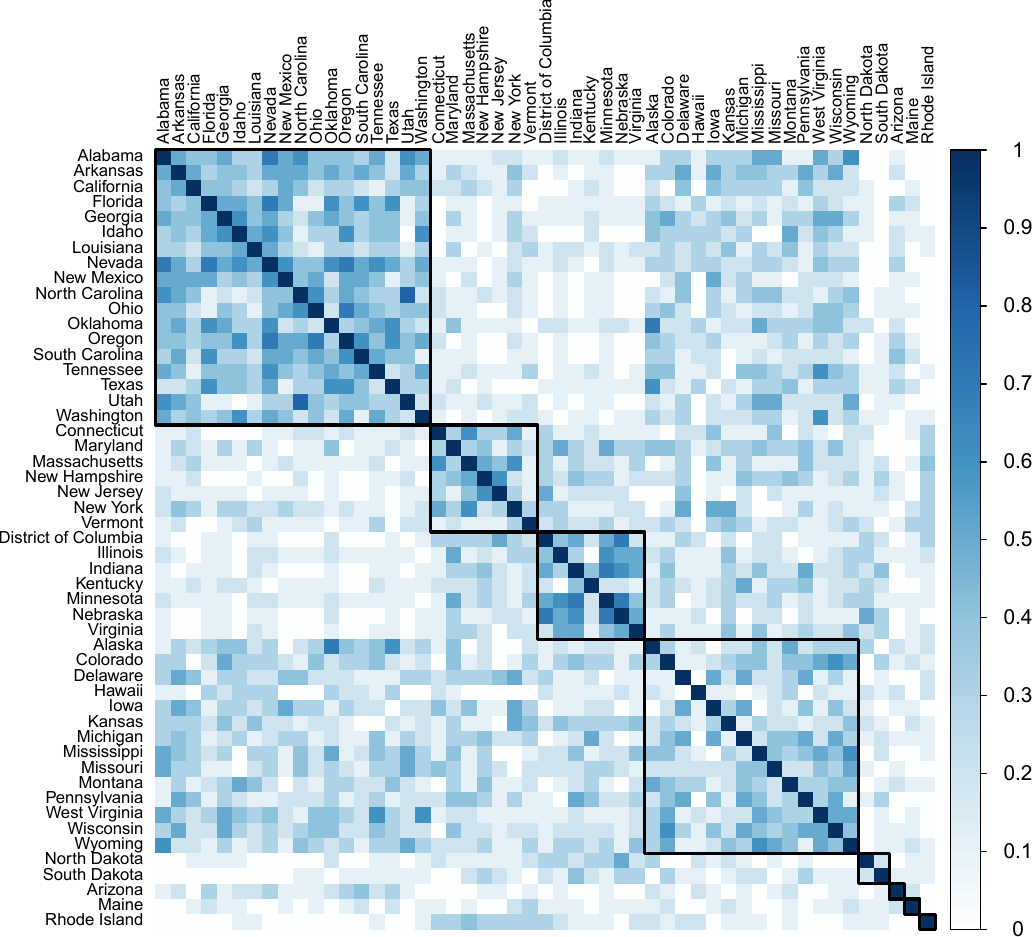}
\end{minipage}
\caption{Similarity matrix for US states computed based on data relating to the second `wave' of the COVID-19 pandemic in the US, taken to be from 26 June 2020 to 29 August 2020.}
\label{fig:covidsecondwave}
\end{figure}

In order to understand some of the difficulties of this challenging forecasting problem, we fitted an error decomposition model of the form
\begin{align}
\log \left(  \frac{1 + \text{cases}_{w, l}}{1 + \text{est.cases}_{m, t, w, l}} \right) = \alpha_0 + \alpha_{m, t} + \beta_{w, l} + \eta_{m, t, w, l},
\end{align}
where $w$ is the week that the forecast is for, $l$ is the state, $m$ indexes the forecasting model, $t$ is the `target' number of weeks in the future the forecast is for, $\eta_{m, t, w, l}$ is an error term, and $\text{cases}_{w, l}$ and $\text{est.cases}_{m, t, w, l}$ are the observed and estimated cases respectively. 
This decomposition allows an interaction term between time and location, which is important given that the pandemic was known to be more severe at different times for different areas. An interaction between model and forecasting distance was also included in order to capture the effect of some models potentially being more `short-sighted' than others. 
The inclusion of the $+1$ on the left-hand side is to avoid numerators or denominators of zero.

We used data from 6 April 2020 to 19 October 2020, giving a total of $100\,264$ $(m, t, w, l)$-tuples. We applied a SCOPE penalty with $\gamma = 8$ to $\beta_{w, l}$, which had 1428 levels. The $\alpha_{m, t}$ coefficients, which amounted to $170$ levels, were left unpenalised. The additional tuning parameter $\lambda$ was selected using the Extended Bayesian Information Criterion \citep{chen2008extended} rather than cross-validation, as it was more suited to this sort of exploratory analysis on data with a chronological structure.

The resulting estimates $\hat{\beta}_{w, l}$ had $8$ levels. We measured the `similarity' of two US states $l_a$ and $l_b$ over a period of time by computing the proportion of weeks at which their estimates $\hat{\beta}_{w, l_a} = \hat{\beta}_{w, l_b}$ coincided. The similarity matrix presented in Figure~\ref{fig:covidsecondwave} was constructed based on the second `wave' of the epidemic which occurred in Summer 2020, with clusters identified by applying spectral clustering on the similarity matrix and plotted in order of decreasing within-cluster median pairwise similarity. 

The resulting clusters are at once interpretable and interesting. Roughly speaking, the top 3 clusters (`top' when ordered according to median pairwise within-cluster agreement) correspond to states that experienced notable pandemic activity in the second, first, and third `waves' of the U.S. coronavirus pandemic, respectively. The first cluster features several southern States (e.g., Georgia, Florida, Texas) which experienced a surge of COVID cases in June--July. The second cluster features east coast states (e.g., New Jersey and New York) which experienced an enormous pandemic toll in April--May. And the third features midwestern states (e.g., Kentucky, Indiana, Nebraska) which had upticks most recently in September-October.  

\section{Discussion} \label{sec:discuss}
In this work we have introduced a new penalty-based method for performing regression on categorical data. An attractive feature of a penalty-based approach is that it can be integrated easily with existing methods for regression with continuous data, such as the Lasso.
Our penalty function is nonconvex, but in contrast to the use of nonconvex penalties in standard high-dimensional regression problems, the nonconvexity here is necessary in order to obtain sparse solutions, that is fusions of levels. Whilst computing the global optimum of nonconvex problems is typically very challenging, for the case with a single categorical variable with several hundred levels, our dynamic programming algorithm can typically solve the resulting optimisation problem in less than a second on a standard laptop computer. The algorithm is thus fast enough to be embedded within a block coordinate descent procedure for handling multiple categorical variables.

We give sufficient conditions for SCOPE to recover the oracle least squares solution when $p=1$ involving a minimal separation between unequal coefficients that is optimal up to constant factors. For the multivariate case where $p>1$, we show that oracle least squares is a fixed point of our block coordinate descent algorithm, with high probability.

Our work offers several avenues for further work. On the theoretical front, it would be interesting to obtain guarantees for block coordinate descent to converge to a local optimum with good statistical properties, a phenomenon that we observe empirically. On the methodology side, it would be useful to generalise the penalty to allow for clustering multivariate coefficient vectors; such clustering could be helpful in the context of mixtures of regressions models, for example.

\bibliographystyle{abbrvnat}

\appendix
\section{Appendix}

\subsection{Candidate minimiser functions} \label{sec:analyticminima}
In this section we give explicit forms of the functions $p_{k,r}$ as defined in Section~\ref{sec:univariate}.
We write $q_{k,r}(x)= a_r x^2 + b_r x + c_r$ for simplicity, suppressing the subscript $k$. 
For $S \subseteq \R$ and $a,b\in \R$, we write $a S + b$ for the set $\{ ax + b : x \in S \}$.

Recall from Section~\ref{sec:univariate} that
\[
u_{k, r, t}(\theta_{k+1}) := \tilde{\min_{\theta_k \in D_k : \theta_k < \theta_{k+1}}} \{ \tilde{q}_{k,r}(\theta_k) + \tilde{\rho}_t(\theta_{k+1} - \theta_k)\}.
\]
For a function $f : \R \rightarrow \R \cup \{ \infty \}$, we denote the \emph{effective domain} of $f$ by
\[
\text{dom } f := \{ x \in \R : f(x) < \infty \}.
\] 
For each $r = 1, \ldots, m(k)$, there are cases corresponding to $t=1$ and $t=2$. The formulas are as follows:
\begin{align*}
	u_{k, r, 1} (x) &= \frac{2a_r x^2 + 2(b_r - 2 a_r \gamma \lambda )x + \gamma(b_r-\lambda)^2}{2(1 - 2 a_r \gamma )} +c_r, \\
	\text{with}\quad  \text{dom } u_{k,r,1} &=
	\begin{cases}
		\left( (1 - 2 a_r \gamma ) I_{k, r} + \gamma (\lambda - b_r) \right) \cap \big[ \frac{\lambda - b_r}{2 a_r} ,\gamma \lambda - \frac{b_r}{2 a_r} \big) &\text{if}\quad  2a_r - 1 / \gamma > 0 \\
		\emptyset &\text{otherwise.}
	\end{cases}
\end{align*}
If $g_k ( \theta_{k +1}) = u_{k,r,1}(\theta_{k + 1})$, then 
\begin{align*}
	b_k ( \theta_{k + 1}) = \frac{\theta_{k + 1} + \gamma ( b_r - \lambda ) }{1 - 2 a_r \gamma}.
\end{align*}
The second case is
\begin{align*}
	u_{k, r, 2} (x) &= - \frac{b_r^2}{4a_r} + c + \frac{1}{2} \gamma \lambda^2, \\
	\text{with } \quad \text{dom } u_{k, r, 2} &= \begin{cases}
 \big[ - \frac{b_r}{2 a_r} + \gamma \lambda , \infty\big ) &\text{if } \quad a_r>0 \text{ and} -b_r / 2 a_r \in I_{k, r}\\
 \emptyset &\text{otherwise.}	
 \end{cases}
\end{align*}
Here, if $g_k ( \theta_{k + 1}) = u_{k, r, 2}(\theta_{k + 1})$, then 
\begin{align*}
	b_k (\theta_{ k + 1}) = -b_r / 2 a_r.
\end{align*}
Considering \eqref{eq:g_k_def}, we see that we can also have the case  where $g_k(\theta_{k+1}) = f_k(\theta_{k+1})$. Thus we can form the set of quadratics $p_{k,r}$ and associated intervals as the set of $u_{k,r,t}$ as above for $t=1,2$ and the $q_{k,r}$ themselves. Note that when $g_k(\theta_{k+1}) = q_{k,r}(\theta_{k+1})$, we have $b_k(\theta_{k+1}) = \theta_{k+1}$.

\subsection{Algorithm details} \label{sec:algdetail}
\begin{algorithm}
	\caption{Outline of procedure for computing $f_k$} \label{alg:fullalg}
	\begin{algorithmic}[1]
	\WHILE{$E ,N(x) \neq \emptyset$ }
		\IF{$\min \{ y \colon (y, r) \in N(x) \} < \min E$}
	\STATE $(y^*, r^*) = \argmin \{ y \colon (y, r) \in N(x) \}$\\
	  $U = U \cup \{([\tilde x, y^* ), r(x))\}$, 
	  $x = \tilde x = y^*$, 
	  $r(x) = r^*$ \\
	  $N(x) = \emptyset$, for any intersection between $p_{k-1,r(x)}$ and any $p_{k-1,r}$ with $r \in A ( x) \backslash \{ r(x) \}$ at location $y >  x$, set $N(x) = N(x) \cup \{ (y, r)\}$.
	\ELSE
	\STATE $y^* = \min E$, $E = E \backslash \{ y^* \}$, \\
	Update active set $A(y^*)$
		\IF{$r(x) \notin A(y^*)$}
	\STATE 	Set $r^*$ such that $p_{k-1, r^*} = \text{\textsf{ChooseFunction}}(A(y^*), y^*)$ \\
		$U = U \cup \{([\tilde x, y^* ), r(x))\}$, 
	  $x = \tilde x = y^*$, 
	  $r(x) = r^*$ \\
	  	  $N(x) = \emptyset$, for any intersection between $p_{k-1,r(x)}$ and any $p_{k-1,r}$ with $r \in A ( x) \backslash \{ r(x) \}$ at location $y \geq  x$, set $N(x) = N(x) \cup \{ (y, r)\}$.
		\ELSE 
\IF{$p_{k-1,r(x)} \neq p_{k-1,r^*} = \text{\textsf{ChooseFunction}}(A(y^*),y^*)$}
	\STATE 	$U = U \cup \{ ([\tilde x, y^* ), r(x))\}$, 
	  $x = \tilde x = y^*$, 
	  $r(x) = r^*$ \\
	  	  $N(x) = \emptyset$, for any intersection between $p_{k-1,r(x)}$ and any $p_{k-1,r}$ with $r \in A ( x) \backslash \{ r(x) \}$ at location $y >  x$, set $N(x) = N(x) \cup \{ (y, r)\}$.
	  	  \ELSE
	  	  \IF{$A(y^*) \neq A(x)$}
		\STATE For any intersection between $r(x)$ and any $r \in A ( y^*) \backslash A(x)$ at location $y >  x$, set $N(y^*) = N(y^*) \cup \{(y, r)\}$.\\
		For any $(y, r) \in N(x)$ with $r \notin A(y^*)$, set $N(y^*) = N(y^*) \backslash \{(y, r) \}$ \\ 
		$x = y^*$
		\ENDIF
		\ENDIF
	\ENDIF
	\ENDIF
	\ENDWHILE
	\end{algorithmic}
\end{algorithm}

\begin{algorithm}
	\caption{\textsf{ChooseFunction}$(H, x)$} \label{alg:ChooseFunction}
	\begin{algorithmic}[1]
	\REQUIRE{$H = \{ h_1, \ldots, h_n \}$ a set of functions, $x$ a real number}
	\STATE	Set $H_1 = \argmin \{h(x) : h \in H \}$
	\IF{$|H_1| = 1$}
	\STATE Select $h^* \in H_1$
	\ELSE
	\STATE Set $H_2 = \argmin \{ h'(x) : h \in H_1 \}$
	\IF{$|H_2| = 1$}
	\STATE Select $h^* \in H_2$
	\ELSE
	\STATE Set $H_3 = \argmin \{ h''(x) : h \in H_2 \}$\\
	Select $h^* \in H_3$ (choosing $h_i \in H_3$ with $i$ minimal if $|H_3| > 1$)
	\ENDIF
	\ENDIF
	\ENSURE{$h^*$}
	\end{algorithmic}
\end{algorithm}

Algorithm~\ref{alg:fullalg} describes in detail how the optimisation routine works. In the algorithm we make use of the following objects:
\begin{itemize}
  \item for $x \in \R$, $A(x)$ is the active set at $x$;
  \item $E$ is the set of points at which the active set changes;
  \item $N(x)$ is the intersection set at $x$;
  \item $U$ is a set of tuples $(I, r)$ where $I \subseteq \R$ is an interval and $r$ is an integer, which is dynamically updated as the algorithm progresses.
\end{itemize}
See Section~\ref{sec:comp_details} for definitions of the sets above.
We also use the convention that if $x = -\infty$ then $[x, y) = (-\infty, y)$.

All of the $p_{k, 1}, \ldots, p_{k, m(k)}$ and $J_{k, m}$ are computed at the start of each iterate $k$.
We then initialise
\begin{align*}
	E = \bigcup_{r = 1}^{n(k)} \partial J_{k - 1, r},
\end{align*}
the set of all of the end-points of the intervals $J_{k-1, 1} , \ldots, J_{k-1, n(k)}$.

Here $x$ can be thought of as the `current position' of the algorithm; $\tilde x$ is used to store when the minimising function $p_{k-1, r(x)}$ last changed. We initialise $\tilde x = -\infty$ and $x = - 1 + \max \{ y \in I_{k-1, 1} \colon f'_{k-1}(y_-) \leq 0 \}$.
This choice of $x$ ensures that the active set $A(x)$ contains only one element (as mentioned in Section~\ref{sec:univariate}); this will always be the index corresponding to the function $\tilde q_{k-1, 1}$.

We initialise the output set $U = \emptyset$, which by the end of this algorithm will be populated with the functions $\tilde q_{k, 1}, \ldots, \tilde q_{k, m(k)}$ and their corresponding intervals $I_{k,1}, \ldots, I_{k, m(k)}$ that partition $\R$.
Finally, we initialise the set $N(x)$ which will contain the intersections between $p_{k-1, r(x)}$ and other functions in the active set. As the active set begins with only one function, we set $N(x) = \emptyset$.

As mentioned in Section~\ref{sec:univariate}, there are several modifications that can speed up the algorithm. 
One such modification follows from the fact that for each $r$, $u_{k,r,2}$ is a constant function over its effective domain, and their effective domain is a semi-infinite interval (see Section~\ref{sec:analyticminima} of the Appendix for their expressions). Therefore, for a given point $x \in \R$, we can remove all such functions from $A(x)$ except for the one taking the minimal value.

We also note that in Algorithm~\ref{alg:fullalg}, the set $N(x)$ is not recomputed in its entirety at every point $x$ at which $A(x)$ is updated, as is described in Section~\ref{sec:univariate}. Line 13 shows how sometimes $N(x)$ can instead be updated by adding or removing elements from it.
Often, points 3 (i) and 3 (ii) from the description in the Section~\ref{sec:univariate} will coincide, and in such instances some calls to \textsf{ChooseFunction} (Algorithm~\ref{alg:ChooseFunction}) can be skipped.

\newpage
\begin{cbunit}
\resumesections
\setcounter{page}{1}
\section*{Supplementary material} 
This supplementary material is organised as follows. In Section~\ref{sec:extracompdetails} we include further details of our algorithm and the proofs of results in Sections~\ref{sec:scopeintro}~\&~\ref{sec:computation}.  The proofs of Theorems~\ref{thm:univarglobal} and \ref{thm:multivaroracle} along with a number of lemmas they require can be found in Section~\ref{sec:theoryproof}.  Section~\ref{sec:extraexp} contains information regarding simulation settings and additional results for the experiments in Section~\ref{sec:numerical} of the main paper. 
\section{Additional algorithmic details} \label{sec:extracompdetails}
\subsection{Remarks on constrained and unconstrained formulations of the univariate objective} \label{sec:constraintremarks}
It is clear why the identifiability constraint \eqref{eq:identifiability} is important when we consider the multivariate problem in Section~\ref{sec:bcd}. However, for the univariate problem, both constrained and unconstrained formulations of the objective can be clearly defined:
	\begin{align}
		\hat{\mbb\theta}^\text{c} &\in \argmin_{ \mbb\theta \in \Theta} \frac{1}{2} \sum_{k=1}^K w_k \left( \bar{Y}_k - \hat{\mu} - \theta_k \right)^2 + \sum_{k=1}^{K-1} \rho(\theta_{(k+1)} - \theta_{(k)} ), \label{eq:constrain1} \\
		\hat{\mbb\theta}^\text{u} &\in \argmin_{\mbb\theta \in \R^K} \frac{1}{2} \sum_{k=1}^K w_k \left( \bar{Y}_k - \theta_k \right)^2 + \sum_{k=1}^{K-1} \rho(\theta_{(k+1)} - \theta_{(k)}). \label{eq:constrain2}
	\end{align}
	As discussed in Section~\ref{sec:preliminaries}, we can enlarge the feasible set in \eqref{eq:constrain1} to be all of $\R^K$: similarly to the observation that $\sum_k w_k \hat{\theta}^\text{u}_k = \hat{\mu} = \sum_k w_k \bar{Y}_k$, the minimiser of \eqref{eq:constrain1} over all of $\R^K$ will always be in $\Theta$. This can be shown by following the argument at the beginning of the proof of Lemma~\ref{thm:nullconsistency}. Therefore the algorithm defined in Section~\ref{sec:univariate} can also be applied to the unconstrained formulation of the objective.
	
	It is clear that these problems are essentially identical, as $\hat{\mbb\theta}^\text{u}$ is a minimiser of the unconstrained objective if and only if $\hat{\mbb\theta}^\text{u} - \hat{\mu} \mbb 1$ is a minimiser of the constrained objective.
	Observe that while $\hat{\mbb\theta}^\text{u} \in \R^K$, the solution to the constrained objective is in fact $( \hat{\mu}, \hat{\mbb\theta}^\text{c}) \in \R \times \Theta$, which is the same $K$-dimensional space only with a different parametersation. In particular, $\hat{\mbb\theta}^\text{c}$ is non-unique if and only if $\hat{\mbb\theta}^\text{u}$ is non-unique.
	
	Since one can obtain the solution to the constrained objective by solving the unconstrained one and then reparameterising (and vice versa), we are free to assume without loss of generality that $w^T \bar{Y} = 0$, so $\hat{\mu}=0$, when solving the univariate problem, and will remark where we do this.

\subsection{Proofs of results in Sections~\ref{sec:scopeintro}~\&~\ref{sec:computation}} \label{sec:algoproofs}

\begin{proof}[Proof of Proposition~\ref{prop:convexpenalty}]
Assume, without loss of generality, that $\hat{\mu}=0$.
	Suppose that there exists $l \neq k$ such that $\hat{\theta}_{k} = \hat{\theta}_{l}$. 
	Without loss of generality we have that $\bar{Y}_{k} \neq \hat{\theta}_{k}$ (if $\bar{Y}_{k} = \hat{\theta}_{k}$ then $\bar{Y}_{l} \neq \hat{\theta}_{l}$ and it can be seen that $\hat{\theta}_{(1)} < \bar{Y}_{l} < \hat{\theta}_{K}$, in which case swap labels).
	
	Now we construct $\tilde{\mbb\theta}$ by setting $\tilde{\theta}_{r} = \hat{\theta}_r \wedge \bar{Y}_{k}$ for $r = 1, \ldots, k$, and $\tilde{\theta}_r = \hat{\theta}_r$ otherwise. We have $\ell ( \hat{\mu}, \tilde{\mbb\theta}) < \ell ( \hat{\mu}, \hat{\mbb\theta})$ and, by convexity of $\rho$, it follows that
			\begin{align*}
			\sum_{r=1}^{K-1} \rho( \tilde{\theta}_{(r+1)} - \tilde{\theta}_{(r)} ) \leq \sum_{r=1}^{K-1} \rho( \hat{\theta}_{(r+1)} - \hat{\theta}_{(r)} ).
		\end{align*}
		This gives the conclusion $Q(\tilde{\mbb\theta}) < Q(\hat{\mbb\theta})$, contradicting the optimality of $\hat{\theta}$.	
\end{proof}

\begin{proof}[Proof of Proposition~\ref{thm:orderprop}]
	Suppose, for a contradiction, that $\hat{\theta}_{k} < \hat{\theta}_{l}$.
	Then at least one of the following must be true:
	\begin{align}
	\left| \hat{\mu} + \hat{\theta}_{k} - \bar{Y}_{k} \right| &> \left| \hat{\mu} + \hat{\theta}_{l} - \bar{Y}_{k} \right| \label{eq:order1} \\
	\left| \hat{\mu} + \hat{\theta}_{l} - \bar{Y}_{l} \right| &> \left| \hat{\mu} + \hat{\theta}_{k} - \bar{Y}_{l} \right|. \label{eq:order2}
	\end{align}

Let $\tilde{\mbb\theta}$ be defined as follows. Set $\tilde{\theta}_r = \hat{\theta}_r$ for all $r \neq k, l$. If \eqref{eq:order1} holds set $\tilde{\theta}_{k} = \hat{\theta}_{l}$ and if \eqref{eq:order2} holds set $\tilde{\theta}_{l} = \hat{\theta}_{k}$.
Observe that 
\begin{align*}
	\sum_{r=1}^n \rho ( \hat{\theta}_{(r+1)} - \hat{\theta}_{(r)} ) \geq \sum_{r=1}^n \rho ( \tilde{\theta}_{(r+1)} - \tilde{\theta}_{(r)} )
\end{align*}
and that the squared loss of $\tilde{\mbb\theta}$ is strictly smaller than the squared loss of $\hat{\mbb\theta}$, thus contradicting optimality of $\hat{\mbb\theta}$.
\end{proof}

\begin{proof} [Proof of Proposition~\ref{prop:uniquewhp}]
In this proof we consider the unconstrained formulation of the objective \eqref{eq:constrain2} discussed in Section~\ref{sec:constraintremarks}. 
Suppose that $(\bar{Y}_k)_{k=1}^K$ is such that there are two distinct solutions to \eqref{eq:univarobjordered}, $\hat{\mbb\theta}^{(1)} \neq \hat{\mbb\theta}^{(2)}$. Let us assume that the levels are indexed such that $\bar{Y}_1 \leq \cdots \leq \bar{Y}_K$. Define $k^* = {\max \{ k : \hat{\theta}^{(1)}_k \neq \hat{\theta}^{(2)}_k\}}$ to be the largest index at which the two solutions take different values and note that we must have $\hat{\theta}^{(r)}_1 \leq \cdots \leq \hat{\theta}^{(r)}_K$.

First consider the case where $k^* < K$. Then
 \[
 S_r := \{k : \hat{\theta}^{(r)}_k = \hat{\theta}^{(r)}_{k^*+1} \} \subseteq \{k^*+1, k^* +2,\ldots,K\},
 \]
 for $r=1, 2$.
We now argue that we must have $\hat{\theta}^{(1)}_{k^* + 1} = \hat{\theta}^{(2)}_{k^* + 1} =:t^* \geq (\hat{\theta}^{(1)}_{k^*} \vee \hat{\theta}^{(2)}_{k^*}) + \gamma \lambda$. Indeed, suppose not, and suppose that without loss of generality $\hat{\theta}^{(2)}_{k^*} > \hat{\theta}^{(1)}_{k^*}$. Fix $r \in \{1,2\}$. The directional derivative of the objective in the direction of the binary vector with ones at the indices given by $S_r$ and zeroes elsewhere evaluated at $\hat{\mbb\theta}^{(r)}$ must be $0$. But comparing these for $r=1, 2$, we see they are identical except for the term $\rho'(\theta_{k^*+1} - \hat{\theta}^{(r)}_{k^*})$, which will be strictly larger for $r=2$, giving a contradiction. This then implies that both $\hat{\theta}^{(1)}_{k^*}$ and $\hat{\theta}^{(2)}_{k^*}$ must minimise $f_{k^*}$ over $\theta \leq t^* - \gamma\lambda$ since
the full objective value is
\[
Q(\hat{\mbb\theta}^{(r)}) = f_{k^*}(\hat{\theta}^{(r)}_{k^*})  + \frac{1}{2}\gamma \lambda^2 + \text{(terms featuring only index $k^* + 1$ or higher)}
\] for $r=1,2$.
We also have that when $k^*=K$, both $\hat{\theta}^{(1)}_{k^*}$ and $\hat{\theta}^{(2)}_{k^*}$ must minimise $f_{k^*}$.

Using the functions $g_{k-1}$ as defined in \eqref{eq:gfunction}, we have the simple relationship that $g_{k-1}(\theta_k) = f_k ( \theta_k) - \frac{1}{2}w_k ( \bar{Y}_k - \theta_k)^2$. In particular, properties (i) and (iii) of Lemma~\ref{lem:optim_prop} hold with $f_k$ replaced by $g_{k-1}$. These can be characterised as $g_{k-1}(\theta_k) = \check{q}_{k,r}(\theta_k)$ for $\theta_k \in I_{k,r}$, where $I_{k,r}$ are the intervals associated with $f_k$ and $\check{q}_{k,r}(\theta_k) := q_{k,r}(\theta_k) - \frac{1}{2}w_k ( \bar{Y}_k - \theta_k)^2$. Note that for each $r$, $\check{q}_{k,r}$ depends on the values of $\bar{Y}_{1},\ldots, \bar{Y}_{k-1}$ but not that of $\bar{Y}_k$ (observe that $q_{k,r}(\theta_k)$ includes a term $\frac{1}{2}w_k ( \bar{Y}_k - \theta_k)^2$; see \eqref{eq:f_k_def}).

Now as $\hat{\theta}^{(1)}_{k^*} \leq \hat{\theta}^{(1)}_{k^* + 1} - \gamma \lambda$ and $\hat{\theta}^{(2)}_{k^*}  \leq \hat{\theta}^{(2)}_{k^* + 1} - \gamma \lambda$ (if $k^* < K$), by Lemma~\ref{lem:optim_prop} (iii) both must be local minima of $f_{k^*}$, and we have that there must exist distinct $r_1 \neq r_2$ such that $\hat{\theta}^{(1)}_{k^*} \in I_{k^*,r_1}$ and $\hat{\theta}^{(2)}_{k^*} \in I_{k^*,r_2}$. Let
\begin{align*}
		\check{q}_{k^*,r_1}(x) &= a_1 x^2 + b_1 x + c_1, \\
	\check{q}_{k^*,r_2}(x) &= a_2 x^2 + b_2 x + c_2.
\end{align*}
Since $\hat{\theta}^{(1)}_{k^*}$ must be the minimum of $\check{q}_{k^*,r_1}(\theta_{k^*}) + \frac{1}{2}w_{k^*}(\bar{Y}_{k^*} - \theta_{k^*})^2$ (and similarly for $\hat{\theta}^{(2)}_{k^*}$), we must have that
\begin{align}
	\min_x \left\{ a_1 x^2 + b_1 x + c_1 + \frac{1}{2}w_{k^*}(\bar{Y}_{k^*} - x)^2 \right\} &= \min_x \left\{ a_2 x^2 + b_2 x + c_2 + \frac{1}{2}w_{k^*}(\bar{Y}_{k^*} - x)^2 \right\} \notag\\ 
	\implies c_1 - \frac{(b_1 - w_{k^*} \bar{Y}_{k^*})^2}{4a_1 + 2w_{k^*}} &= c_2 -  \frac{(b_2 - w_{k^*} \bar{Y}_{k^*})^2}{4a_2 + 2w_{k^*}}. \label{eq:unique1}
\end{align}
This is a quadratic equation in $\bar{Y}_{k^*}$, so there are at most two values for which \eqref{eq:unique1} holds. Considering all pairs $r_1, r_2$, we see that in order for there to exist two solutions $\hat{\mbb\theta}^{(1)} \neq \hat{\mbb\theta}^{(2)}$, $\bar{Y}_{k^*}$ must take values in a set of size at most $c(K)$, for some function $c : \mathbb{N} \rightarrow \mathbb{N}$. 

Now let
\[
\mathcal{S} := \{ (\bar{Y}_k)_{k=1}^K : \text{ the minimiser of the objective is not unique}\} \subseteq \R^K.
\]
What we have shown, is that associated with each element $(\bar{Y}_k)_{k=1}^K \in \mathcal{S}$, there is at least one $k^*$ such that
\[
|\{(\bar{Y}'_k)_{k=1}^K \in \mathcal{S} : \bar{Y}'_k = \bar{Y}_k \text{ for all } k  \neq k^* \}| 
\]
is bounded above by $c(K)$. Now for each $j=1,\ldots,K$, let $\mathcal{S}_j$ be the set of $(\bar{Y}_k)_{k=1}^K \in \mathcal{S}$ for which the there exists a $k^*$ with the property above and $k^*=j$. Note that $\cup_j \mathcal{S}_j = \mathcal{S}$. Now $\mathcal{S}_j \subset \R^K$ has Lebesgue measure zero as a finite union of graphs of measurable functions $f : \R^{K-1} \to \R$. Thus $\mathcal{S}$ has Lebesgue measure zero.
\end{proof}

\begin{proof}[Proof of Lemma~\ref{lem:optim_prop}.]
Assume, without loss of generality, that $\hat{\mu} = 0$.
We proceed inductively, assuming that the properties (i) and (iii) hold for $f_k$, and (ii) holds for $b_{k+1}$. Additionally we include in our inductive hypothesis that for all $x$, $f'_k(x_-) \geq f'_k(x_+)$, where we define $f_k'(x_-)$ and $f'_k(x_+)$ to be the left-derivative and right-derivative of $f_k$ at $x$, respectively. We note that these trivially hold for the base case $f_1$, and the case $b_2$ can be checked by direct calculation.

We first prove (i), that $f_{k+1}$ is continuous, coercive, and piecewise quadratic and with finitely many pieces. We then show that $f'_{k+1}(x_-) \geq f'_{k+1}(x_+)$ for all $x$, which allows us to show that (iii) holds for $f_{k+1}$. Finally, we use these results to show that (ii) holds for $b_{k+2}$.

We now show that $f_{k+1}$ is coercive and continuous. Clearly $g_k(x) \geq \min_{y\leq x} f_k(y)$, so it follows that $g_k(x) \rightarrow \infty$ as $x \rightarrow -\infty$ as $f_k$ is coercive. Furthermore $g_k$ is bounded from below as $f_k$ is coercive and continuous. Thus since $f_{k+1}(x) = g_k(x) + \frac{1}{2} w_{k+1}(\bar{Y}_{k+1} - x)^2$, it follows that $f_{k+1}$ is coercive. Next as $g_k(x) = \min_{y \leq x} f_k(y) + \rho(y - x)$, and $f_k$ and $\rho$ are continuous, it follows that $g_k$ is continuous and therefore that $f_{k+1}$ is continuous. 
	
	To see why $f_{k+1}$ is piecewise quadratic with finitely many pieces, we observe that it can be written $f_{k+1}(x) = f_k(b_{k+1}(x)) + \rho(x - b_{k+1}(x)) + \frac{1}{2}w_{k+1}(\bar{Y}_{k+1} - x)^2$. We have by our inductive hypothesis that $f_k$ is piecewise quadratic and $b_{k+1}(x)$ is piecewise linear, both with finitely many pieces. Since the composition of a piecewise linear function inside a piecewise quadratic function is piecewise quadratic, the remainder of (i) is shown.
	
	We now turn our attention to (iii), and define for $x \in \R$:
	\begin{align*}
		y_*(x) &= \sargmin_{y \leq x} f_k(y) + \rho(x - y), \\
				y^*(x) &= \sargmin_{y \leq x} f_{k+1}(y) + \rho(x - y).
	\end{align*}
	We will first show that $f'_{k+1}(x_+) \leq f'_{k+1}(x_-)$ for all $x \in \R$. Suppose that we are increasing $x$ and we have reached a point where $g_k(x)$ is not differentiable (that is, the left-derivative and the right-derivative do not match). By assumption (ii) for $b_{k+1}$, we can assume that there is some window $\delta > 0$ such that $y_*(t)$ is linear for $t \in (x - \delta, x)$, say $y_*(t) = \alpha + \beta t$.
	
	In order to proceed with the following argument, we must show that for sufficiently small $\epsilon > 0$, we have $\alpha + \beta(x + \epsilon) \leq x + \epsilon$. If $\alpha + \beta x < x$, this is immediate. Therefore it remains to consider the case $\alpha + \beta x = x$, for which we show that we must have $\alpha = 0$ and $\beta = 1$, i.e\ $y_*(t) = t$ for $t \in (x - \delta, x)$. This follows from the observation that if $y_*(t) < t$, then for all $t_1 > t$ we have $y_*(t_1) \notin ( y_*(t), t]$. Indeed, suppose not, then 
	\begin{align*}
		f_k(y_*(t_1)) + \rho(t_1 - y_*(t_1)) &< f_k(y_*(t)) + \rho(t_1 - y_*(t)) \\
		\implies \quad f_k(y_*(t_1)) + \rho(t - y_*(t_1)) &< f_k(y_*(t)) + \rho(t_1 - y_*(t))   + \rho(t - y_*(t_1)) - \rho(t_1 - y_*(t_1)) \\
		&\leq f_k(y_*(t)) + \rho(t - y_*(t)),
	\end{align*}
 contradicting the definition of $y_*(t)$. The last line uses $\rho(t_1 - y_*(t)) - \rho(t_1 - y_*(t_1)) \leq \rho(t - y_*(t)) - \rho(t - y_*(t_1))$, which follows from concavity of $\rho$ and $y_*(t) < y_*(t_1) \leq t < t_1$. 

  	With this established, we have that:
  	\begin{align*}
  	g_k(x - \epsilon) &= f_k(\alpha + \beta (x - \epsilon)) + \rho(x-\epsilon - (\alpha + \beta(x - \epsilon))) \\
  		g_k(x + \epsilon) &= f_k(y_*(x + \epsilon)) + \rho(x + \epsilon - y_*(x + \epsilon)) \\
  		&\leq f_k(\alpha + \beta (x + \epsilon)) + \rho(x + \epsilon - (\alpha + \beta(x + \epsilon))).
  	\end{align*}
Note that $f_k$ has both left-derivatives and right-derivatives at every point in $\R$. Suppose first that $\beta \geq 0$, and we observe that
\begin{align*}
	g_k'(x_-) &= \beta f_k'(y_*(x)_{-}) + (1 - \beta )\rho'(x - y_*(x))
\end{align*}
Then by the basic definition of the right-derivative, 
\begin{align*}
	g'_k(x_+) &= \lim_{\epsilon \rightarrow 0^+} \frac{ f_k(y_*(x + \epsilon)) + \rho(x + \epsilon - y_*(x + \epsilon)) - f_k(y_*(x)) - \rho(x - y_*(x))}{\epsilon}\\
	&\leq \lim_{\epsilon \rightarrow 0^+} \frac{1}{\epsilon}\bigg[ f_k(\alpha + \beta (x + \epsilon)) + \rho(x+\epsilon - (\alpha + \beta(x + \epsilon))) \\
	&\hspace{1.5cm} - f_k(\alpha + \beta x ) - \rho(x- (\alpha + \beta x)) \bigg]\\
	&= \beta f'_k( y_*(x)_{+}) + (1 - \beta) \rho'(x - y_*(x)) \\
	&= g'_k(x_-) + \beta( f'_k( y_*(x)_{+}) - f'_k( y_*(x)_{-})) \\
	&\leq g'_k(x_-),
\end{align*}
where the last inequality follows from our inductive hypothesis that $f'_k(y_{+}) \leq f'_k(y_{-})$ for all $y \in \R$. An analogous argument shows that the same conclusion holds when $\beta < 0$.

Now we use this to prove the claim. Because there are no points of $f_{k+1}$ at which the left-derivative is less than the right-derivative, without loss of generality we claim that $f_{k+1}$ is differentiable at $y^*(x)$ for all $x$, unless $y^*(x) = x$. Indeed, suppose not, then we have that $f'_{k+1}(y^*(x)_-) > f'_{k+1}(y^*(x)_+)$ and necessarily that defining $h(y) := f_{k+1}(y) + \rho(x-y)$, we have $0 \in \partial h(y^*(x))$.
But since $h(y^*(x)_+) < h(y^*(x)_-)$, we contradict the optimality of $y^*(x)$ as this point is in fact a local maximum.

We finally consider claim (ii). By (iii), we have that for every point $x$, $y^*(x)$ is either $x$ or at the minimum of one of the quadratic pieces of $f_{k+1}(\cdot) + \rho(x - \cdot)$. In either case, we have that $y^*(x)$ is linear in $x$ and thus $f_{k+1}(y^*(x)) + \rho(x - y^*(x))$ is quadratic in $x$. We can define $g_{k+1}(x)$ pointwise as the minimum of this finite set of quadratic functions of $x$, whose expressions are given in Appendix~\ref{sec:analyticminima}. Importantly, the coefficients in the linear expression $y^*(x)$ of $x$ depend only on which of these functions is the minimum at $x$. As the number of intersections between elements in this set of quadratic functions is bounded above by twice the square of the size of the set, we can conclude that $b_{k+2}(x)$ is piecewise linear and with a finite number of pieces, thus concluding the proof.
\end{proof}

\subsection{Computation time experiments} \label{sec:comptimeexp}
A small experiment was performed to demonstrate the runtimes one can expect in practice for the univariate problem. Note that this clustering is applied iteratively in the block coordinate descent procedure we propose to use in multivariate settings.
We considered 3 settings: one with no signal, one with 2 true clusters and one with 5 true clusters. 
Independent and identically distributed Gaussian noise was added to each of the subaverages.
As in Section~\ref{sec:adultds2} the number of categories was increased by random splitting of the levels.
Each of these tests were repeated 25 times, on a computer with a 3.2GHz processor. The results are shown in Figure~\ref{fig:timeplot}.
\begin{figure}[h!]
\centering
\begin{minipage}[b]{0.95\textwidth} 
\includegraphics[width=\textwidth]{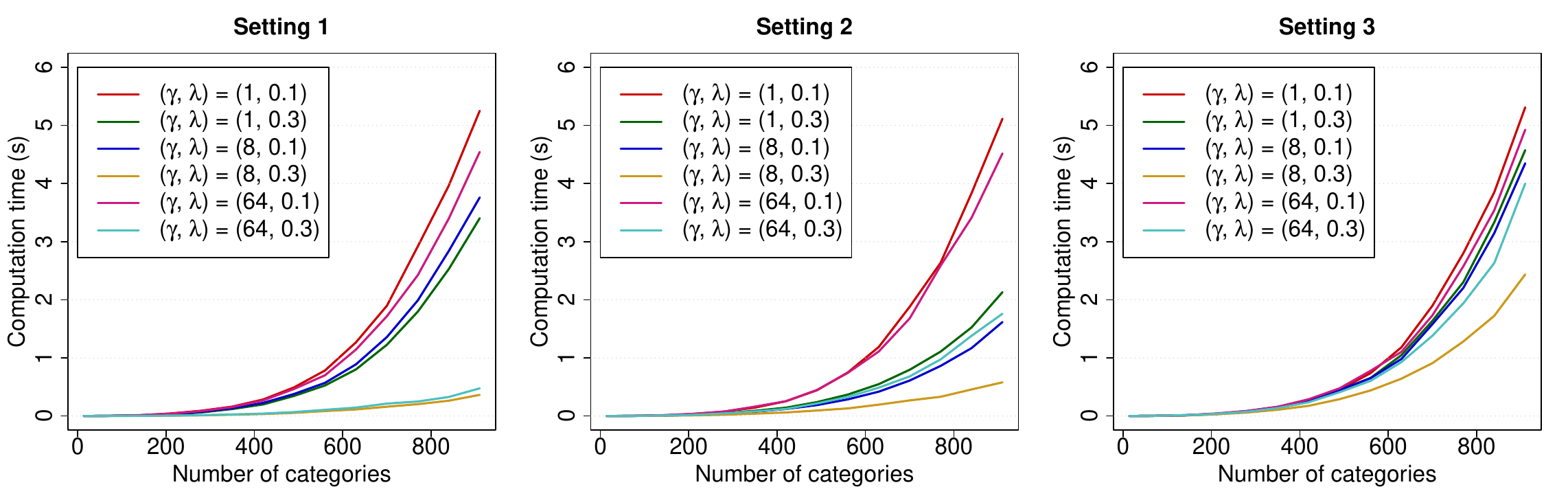}
\end{minipage}
\caption{Computation times for solving the univariate problem.}	
 \label{fig:timeplot}
\end{figure}

\subsection{Discretised algorithm} \label{sec:discretealg}
For very large-scale problems, speed can be improved if we only allow coefficients to take values in some fixed finite grid, rather than any real value. Below we describe how such an algorithm would approximately solve the univariate objective \eqref{eq:univarobjordered}. We will use the unconstrained objective as discussed in Section~\ref{sec:constraintremarks}.
We would first fix $L$ grid points $\vartheta_1 < \cdots < \vartheta_L$, and then proceed as described in Algorithm~\ref{alg:discretealg}. 

\begin{algorithm}
	\caption{Discrete algorithm for computing approximate solution to \eqref{eq:univarobjordered}} \label{alg:discretealg}
	\begin{algorithmic}[1]
	\FOR{$l = 1, \ldots, L$}
	\STATE Set $F_\text{new}(l) = \frac{1}{2} w_1 ( \bar{Y}_1 - \vartheta_l)^2$
	\STATE Set $B(1,l) = l$
	\ENDFOR
	\FOR{$k = 2, \ldots, K$}
	\STATE Set $F_\text{old} = F_\text{new}$
	\FOR{$l = 1, \ldots, L$} 
	\STATE Set $B(k, l) = \argmin_{l' \in \{ 1, \ldots, l \} } F_\text{old}(l') + \rho(\vartheta_l - \vartheta_{l'}) + \frac{1}{2} w_k ( \bar{Y}_k - \vartheta_l)^2$
	\STATE Set $F_\text{new}(l) = F_\text{old}(B(k,l)) + \rho(\vartheta_l - \vartheta_{B(k,l)}) + \frac{1}{2} w_k ( \bar{Y}_k - \vartheta_l)^2$
	\ENDFOR
	\ENDFOR
	\STATE Set $B^*(K) = \argmin F_\text{new}$, and $\hat{\theta}_K = \vartheta_{B^*(K)}$
	\FOR{$k = K - 1, \ldots, 1$}
	\STATE Set $B^*(k) = B(k+1, B^*(k+1))$, and $\hat{\theta}_k = \vartheta_{B^*(k)}$
	\ENDFOR
	\end{algorithmic}
\end{algorithm}

This algorithm has the same basic structure to the approach we use in Section~\ref{sec:univariate} for computing the exact global optimum.
The difference is that now, instead of as in \eqref{eq:fkfull}, we define $f_k$ in the following way:
\[
	f_k(\theta_k) := \min_{\substack{(\theta_1,\ldots,\theta_{k-1})^T \in \{ \vartheta_1,\ldots,\vartheta_L \}^{k-1} \\ \theta_1 \leq \cdots \leq \theta_{k-1} \leq \theta_k } } \bigg\{ \frac{1}{2} \sum_{l=1}^{k} w_l(\bar{Y}_l - \theta_l)^2 + \sum_{l=1}^{k-1}\rho(\theta_{l+1} - \theta_l) \bigg\}.
\]
The objects $F$ and $B$ play analogous roles to $f_k$ and $b_k$ in Section~\ref{sec:univariate}.
Since we restrict $\theta_k \in \{ \vartheta_1, \ldots, \vartheta_L \}$, we only need to store the values that $f_k$ takes at these $L$ values; this is the purpose of the vector $F$ in Algorithm~\ref{alg:discretealg}. 
Similarly, the rows $B(k, \cdot)$ serve the same purpose as the functions $b_k$ where, again, we only need to store $L$ values corresponding to the different options for $\theta_k$.

This algorithm returns the optimal solution $\hat{\mbb \theta}$ to the objective where each of the coefficients are restricted to take values only in $\{ \vartheta_1, \ldots, \vartheta_L \}$. We must ensure that the grid of values has fine enough resolution that interesting answers can be obtained, which requires $L$ being sufficiently large. The number of clusters obtained by this approximate algorithm is bounded above by $L$, so this must not be chosen too small.

One can see that the computational complexity of this algorithm is linear in $K$, with a total of $O(K L^2)$ operations required. This is of course in addition to the $O(n)$ operations needed to compute $w_1,\ldots,w_K$ and $\bar{Y}_1, \ldots, \bar{Y}_K$ beforehand. In particular, choosing $L \lesssim \sqrt{K}$ guarantees that the complexity of this algorithm is at worst quadratic in $K$.

\section{Proofs of results in Section~\ref{sec:theory}} \label{sec:theoryproof}

\subsection{Proof of Theorem~\ref{thm:univarglobal}} \label{sec:uniproof}
The proof of Theorem~\ref{thm:univarglobal} requires a number of auxiliary lemmas, which can be found in Section~\ref{sec:unilemmas}.

\begin{proof}[\unskip\nopunct]
Let us define $R_i = Y_i - \hat{\mu}$ for $i  = 1, \ldots, n$, and $\bar{R}_k = \frac{1}{n_k} \sum_{i=1}^n \ind_{\{ X_i = k \}} R_i$ for $k = 1, \ldots, K$. 
Note that
\begin{align*}
	R_i = \sum_{k=1}^K \ind_{\{ X_i = k \}}\theta^0_k + (P \varepsilon)_i
\end{align*}
where $P = I - \mathbf{1}\mathbf{1}^T /n$. 

For each $k = 1, \ldots, K$, we define the event
	\begin{align*}
	 \Lambda_k = \left\{ 	\left\vert \frac{1}{n_k} \sum_{i=1}^n \ind_{\lbrace X_i = k \rbrace} (P\varepsilon)_i  \right\vert < \frac{1}{2} \sqrt{ \eta \gamma_* s} \lambda \right\}.
	\end{align*}
	 By a union bound, we have that $\mathbb{P}( \cap_{k=1}^K \Lambda_k) \geq 1 - \sum_{k=1}^K \mathbb{P}(\Lambda_k^c)$. Now observe we can write \[
	\frac{1}{n_k} \sum_{i=1}^n \ind_{\{ X_i = k \} } (P \varepsilon)_i =  {v^{(k)}}^T P \varepsilon,
	\]
	where we define $v^{(k)} \in \R^n$ by $v^{(k)}_i = \frac{1}{n_k} \ind_{\{ X_i = k \}} $.
	 Since $P$ is an orthogonal projection matrix, we have that $\| P v^{(k)} \|_2 \leq \| v^{(k)} \|_2 = \frac{1}{\sqrt{n_k}}$. It follows that $ {v^{(k)}}^T P \varepsilon $ is sub-Gaussian with parameter $\sigma / \sqrt{n_k}$. Applying the standard sub-Gaussian tail bound, we obtain
	\begin{align*}
		\mathbb{P}(\Lambda_k^c)  &= \mathbb{P}\left(  \left\vert \frac{1}{n_k} \sum_{i=1}^n  \ind_{\lbrace X_i = k \rbrace}  (P \varepsilon)_i \right\vert \geq \frac{1}{2} \sqrt{ \eta \gamma_* s} \lambda \right) \\
		&\leq  2 \exp\left( - \frac{n w_k \eta \gamma_* s \lambda^2}{8 \sigma^2}   \right),
	\end{align*}
	where recall that $w_k = n_k / n$.
	Therefore, we have that
	\begin{align}
		\mathbb{P} \left( \cap_{k=1}^K \Lambda_k \right) \geq 1 - 2 \sum_{k=1}^K \exp \left( - \frac{n w_k \eta \gamma_* s \lambda^2}{8 \sigma^2}  \right) 
		\geq 1 - 2  \exp \left( - \frac{n w_{\text{min}} \eta \gamma_* s \lambda^2}{8 \sigma^2}  + \log ( K ) \right). \label{eq:univartailboundcompute}
	\end{align}
In the following we work on the intersection $\Lambda := \cap_{k=1}^K \Lambda_k$. This entails that for each $k$, $| \bar{R}_k - \theta^0_k | < \sqrt{\eta \gamma_* s} \lambda / 2$.
We now relabel indices such that $\bar{R}_1 \leq \cdots \leq \bar{R}_K$, and so from Proposition~\ref{thm:orderprop} that $\hat{\theta}_1 \leq \cdots \leq \hat{\theta}_K$. Since our assumption \eqref{eq:univarsigstrength} implies $\Delta( \mbb\theta^0) \geq \sqrt{\eta \gamma_* s} \lambda$, it follows that on $\Lambda$ the observed ordering is consistent with the ordering of the true coefficients, i.e.\ there exist $0=k_0 < k_1 < \cdots < k_s=K$ such that
\begin{align}
	\theta^0_1 = \cdots= \theta^0_{k_1} < \theta^0_{k_1 + 1} = \cdots = \theta^0_{k_2} < \cdots < \theta^0_{k_{s-1} + 1} = \cdots = \theta^0_{k_s}. \label{eq:oracle_indices}
\end{align}
Indeed, observe that for $j = 1, \ldots, s - 1$, we have by the triangle inequality and \eqref{eq:univarsigstrength}, the stronger property that
\begin{align}
	\bar{R}_{k_j + 1} - \bar{R}_{k_{j}} &> 3 \left(1 + \frac{\sqrt{2}}{\eta} \right) \sqrt{\gamma \gamma^*} \lambda - \sqrt{\eta \gamma_* s}\lambda \notag\\
	&> \gamma \lambda + 2 ( \sqrt{2 s / \eta} \sqrt{\gamma} \lambda \vee \gamma \lambda) + 2 \sqrt{\eta \gamma_* s}\lambda. \label{eq:univar_sep_established}
\end{align}
Our optimisation objective is therefore
\begin{align}
	\hat{\mbb\theta} \in \argmin_{\mbb\theta \in \Theta} \frac{1}{2} \sum_{k=1}^K w_k ( \bar{R}_k - \theta_k)^2 + \sum_{k=1}^K \rho(\theta_{k+1} - \theta_k). \label{eq:univarproof_innerobjective}
\end{align}
Since $\bar{R}_{k_j} - \bar{R}_{k_{j-1}+ 1} < \sqrt{\eta \gamma_* s} \lambda$ for $j = 1, \ldots, s$, it follows from Lemma~\ref{thm:oraclelselemma}
that $ \hat{\theta}_{k_j + 1} - \hat{\theta}_{k_j} \geq \gamma \lambda$ for $j = 1, \ldots , s - 1$, so
\begin{align}
	Q(\hat{\mbb\theta}) &= \frac{1}{2}\sum_{k=1}^K w_k (\bar{R}_k - \hat{\theta}_k)^2 + \sum_{k = 1}^{K-1} \rho ( \hat{\theta}_{k+1} - \hat{\theta}
	_{k}) \notag \\
	&= \frac{1}{2} \sum_{k = 1}^K w_k (\bar{R}_k - \hat{\theta}_k)^2 + \sum_{j = 1}^s \sum_{k = k_{j-1}+1}^{k_j - 1} \rho(\hat{\theta}_{k+1} - \hat{\theta}_k) + \frac{s - 1}{2} \gamma \lambda^2 \\ 
	&= \min_{\mbb\theta \in \R^K} \left(\frac{1}{2} \sum_{k = 1}^K w_k (\bar{R}_k - \theta_k)^2 + \sum_{j = 1}^s \sum_{k = k_{j-1}+1}^{k_j - 1} \rho(\theta_{k+1} - \theta_k)\right) + \frac{s - 1}{2} \gamma \lambda^2. \label{eq:univarsepobj}
\end{align}
Observe that we can have $ k_{j-1} + 1 > k_j - 1$ for some $j$, in which case we take the sum over that range to be zero.
Note that \eqref{eq:univarsepobj} can be optimised over $(\theta_{k_{j-1}+1}, \ldots, \theta_{k_j})$ separately for each $j = 1, \ldots, s$.
If $s = 1$, i.e.\ the true signal is zero, then the result follows from Lemma~\ref{thm:nullconsistency}. Now we see what happens when $s > 1$.

Without loss of generality, consider $j = 1$  and note that if $k_1 = 1$ it is immediate that $\hat{\theta}_1 = \hat{\theta}^0_1$. Hence, we can assume that $k_1 > 1$. We note that $\hat{\theta}^0_1 = \sum_{k=1}^{k_1} w_k \bar{R}_k / w^0_1$, where we define $w^0_k = n^0_k / n$. We see that our goal is to compute
\begin{align}
	&\argmin_{\mbb\theta \in \R^{k_1}} \frac{1}{2} \sum_{k=1}^{k_1} w_k ( \bar{R}_k - \theta_{k})^2 + \sum_{k = 1}^{k_1 - 1} \rho
	( \theta_{k + 1} - \theta_k ) \notag \\
	= \hat{\theta}^0_1\mathbf{1} \: + \: &\argmin_{\mbb \theta \in \R^{k_1}} \frac{1}{2} \sum_{k = 1}^{k_1} w_k ( \tilde{R}_k - \theta_k )^2 - \sum_{k = 1}^{k_1 - 1} \rho
	 (  \theta_{k+1} - \theta_k ), \label{eq:univarolserepar}
\end{align}
where $\mathbf{1} \in \R^{k_1}$ is a vector of ones and $\tilde{R}_k := \bar{R}_k - \hat{\theta}^0_1$ for $k = 1, \ldots, k_1$.
Note that we subtract $\hat{\theta}^0_1$ to ensure that
\[
\sum_{k=1}^{k_1} w_k \tilde{R}_k =0,
\]
as required for application of Lemma~\ref{thm:nullconsistency}.
We have by assumption that for $k \in 1,\ldots, k_1$, $| \tilde{R}_k | \leq \sqrt{ \eta \gamma_* s} \lambda / 2 \leq (2 \wedge \sqrt{w^0_1 \gamma}) \lambda / w^0_1$.
Thus, Lemma~\ref{thm:nullconsistency} can be applied with $\check{w} = w^0_1$ and it follows that $\hat{\theta}_k = \hat{\theta}^0_1$ for $k = 1, \ldots, k_1$.
\end{proof}

\subsubsection{Auxiliary lemmas} \label{sec:unilemmas}
Here we prove a number of results required to obtain conditions for recovering the oracle least squares estimate in the univariate case.
Lemma~\ref{thm:nullconsistency} gives conditions for recovery of the true solution, in the case where there is zero signal.
Lemmas~\ref{thm:oraclelselemma} and \ref{thm:betalexist} ensure that the true levels are far enough apart that they can be separated.
Once we have this separation, we apply Lemma~\ref{thm:nullconsistency} on each of the levels to obtain the solution.

\begin{lem} \label{thm:quadoptlem}
Consider the optimisation problem
\begin{align*}
	x^* = \argmin_{x \geq 0} \frac{\kappa}{2} ( 2 \tau - x)^2 + \rho ( x),
\end{align*}
	where $\tau > 0$ and $\kappa \in (0, 1]$. Suppose further that $\tau < ( 1 \wedge \sqrt{ \kappa \gamma } )  \lambda / 2 \kappa$.
	Then $x^* = 0$ is the unique optimum.
\end{lem}
\begin{proof}
	We first observe that
	\begin{align*}
		x^* = \argmin_{x \geq 0} \frac{\kappa}{2} ( 2 \tau - x)^2 + \rho_{\gamma, \lambda} ( x) = \argmin_{x \geq 0} \frac{1}{2} ( 2 \tau - x)^2 + \rho_{\kappa \gamma, \lambda / \kappa} ( x)  .
	\end{align*}
	For convenience, we define $ F(x) := \ ( 2 \tau - x)^2 / 2 + \rho_{\kappa \gamma, \lambda / \kappa} ( x)$.
	It now suffices to show that $F$ is uniquely minimised at $0$ provided $\tau < ( 1 \wedge \sqrt{ \kappa \gamma } )  \lambda / 2 \kappa$. 
	We can clearly see that $x^* \in [0, 2 \tau ]$.
	Equation (2.3) of \citet{breheny2011coordinate} gives the result when $\kappa \gamma \geq 1$.
	
	When $\kappa \gamma < 1$, we see that any stationary point of $F$ in $[0, \gamma \lambda \wedge 2 \tau ]$ must be a maximum, since on this interval $F(x)$ is a quadratic function with a negative coefficient of $x^2$. Therefore its minimum over $[0, \gamma\lambda]$ is attained at either $x = 0$ or $x = \gamma \lambda \wedge 2 \tau$.
	If $ 2 \tau \leq \gamma \lambda$, then it suffices to check that $F(0) < F(2 \tau)$.
	This holds if and only if $\tau < \gamma \lambda / ( \gamma \kappa + 1 )$, but since we are assuming $\tau \leq \gamma \lambda / 2$ and $\kappa \gamma < 1$, this is always satisfied.
	
	If $\gamma \lambda < 2 \tau$, then we can see that the minimum of $F$ over $[ \gamma \lambda, 2 \tau ]$ will be attained at exactly $2 \tau$.
	Thus, here it also suffices to check $F(0) < F(2 \tau )$, which holds if and only if $ \tau < \sqrt{ \gamma / \kappa } \lambda / 2$.
The final bound $\tau < ( 1 \wedge \sqrt{\kappa \gamma } ) \lambda / 2 \kappa$ follows from combining the results for these cases.
\end{proof}
The following is a deterministic result to establish separation between groups of coefficients.
\begin{lem} \label{thm:oraclelselemma}
Consider the setup of Theorem~\ref{thm:univarglobal}, and assume that $\hat{\mu} = 0$.
Suppose that $\bar{Y}_1 \leq \cdots \leq \bar{Y}_K$, and that for $j = 1, \ldots, s$ we have
\begin{align}
	\bar{Y}_{k_j} - \bar{Y}_{k_{j-1}+1} < \sqrt{\eta \gamma_* s} \lambda, \label{eq:oraclelsenc1}
\end{align}
where $k_j$ and $k_{j - 1}$ are as defined in \eqref{eq:oracle_indices}.
Suppose further that for $j = 1, \ldots, s - 1$,
\begin{align}
	\bar{Y}_{k_j + 1} - \bar{Y}_{k_j} \geq  \gamma \lambda + 2  ( \sqrt{2 s / \eta} \sqrt{\gamma} \lambda \vee \gamma \lambda) + 2 \sqrt{ \eta \gamma_* s} \lambda .\label{eq:oraclelselemsigsep}
\end{align}
Then for $j = 1, \ldots, s$, we have $\bar{Y}_{k_{j-1}+1} \leq \hat{\theta}_{k_{j - 1} + 1} \leq \hat{\theta}_{k_j} \leq \bar{Y}_{k_j}$.
\end{lem}

\begin{proof}
For convenience, within this lemma we define $\zeta := \sqrt{\eta \gamma_* s} \lambda$. Recall that the objective function which $\hat{\mbb\theta}$ optimises takes the form
\begin{align*}
	Q(\mbb\theta) = \frac{1}{2} \sum_{k = 1}^K w_k ( \bar{Y}_k - \theta_k)^2 + \sum_{k = 1}^{K-1} \rho ( \theta_{k+1} - \theta_k ).
\end{align*}
	We first claim that $\hat{\theta}_k \in [ \bar{Y}_1 , \bar{Y}_K ]$ for $k = 1, \ldots, K$.
	To see this, suppose that this is not the case and define $\check{\mbb\theta}$ by projecting $\hat{\mbb\theta}$ onto $[ \bar{Y}_1, \bar{Y}_K ]^K$ (i.e.\ $\check{\theta}_k = \bar{Y}_K \wedge ( \bar{Y}_1 \vee \hat{\theta}_k )$ for $k = 1, \ldots, K$). 
	The penalty contribution from $\check{\mbb\theta}$ is no larger than that of $\hat{\mbb\theta}$, and the loss contribution is strictly smaller, so we obtain the contradiction $Q( \check{\mbb\theta} ) < Q( \hat{\mbb\theta} )$.
	
	We now proceed to show that for $j = 1, \ldots, s - 1$, we have $\hat{\theta}_{k_j} \leq \bar{Y}_{k_j}$ and $\hat{\theta}_{k_j + 1} \geq \bar{Y}_{k_j + 1}$.
	We prove the first of these sets of inequalities, since the second follows similarly by considering the problem with $- \hat{\mbb\theta}$, $- \bar{Y}$ and reversing the indices.
	Suppose, for contradiction, that there exists some $j$ in $\{ 1, \ldots, s - 1 \}$ with $\hat{\theta}_{k_j} > \bar{Y}_{k_j}$. 
	Let this $j$ be minimal, such that for all $l < j$ we have $\hat{\theta}_{k_l} \leq \bar{Y}_{k_l}$.
	
	Next define $l_1$ to be the maximal element of $\{ k_{j - 1} + 1, \ldots, k_j - 1 \} $ such that $\hat{\theta}_{l_1} \leq \bar{Y}_{k_j}$.
Similarly, we define $l_2\in\{ k_j + 1, \ldots, k_{j+1} \}$ to be minimal such that $\hat{\theta}_{l_2} \geq \bar{Y}_{k_j + 1}$. The existence of $l_1$ and $l_2$ is guaranteed by Lemma~\ref{thm:betalexist}. 

We note that for $l = l_1 + 1, \ldots, k_j$, $\hat{\theta}_l = \hat{\theta}_{k_j}$ and hence $(\bar{Y}_l - \hat{\theta}_l)^2 \geq ( \bar{Y}_{k_j} - \hat{\theta}_{l})^2 = ( \bar{Y}_{k_j} - \hat{\theta}_{k_j})^2$. This can be shown by contradiction, as in \eqref{eq:2pointsinmid}.
For such $l$, we have from optimality of $\hat{\theta}$ that $\bar{Y}_l - \hat{\theta}_{l_1} \geq \hat{\theta}_{k_j} - \bar{Y}_l$ (otherwise one could improve the objective by setting $\hat{\theta}_{l_1} = \hat{\theta}_l$) which implies that $\hat{\theta}_{l_1} < \bar{Y}_l$. From this it follows that $( \bar{Y}_l - \hat{\theta}_{l_1} )^2 \leq (\bar{Y}_{k_j} - \hat{\theta}_{l_1})^2$, since $\hat{\theta}_{l_1} < \bar{Y}_l \leq \bar{Y}_{k_j}$.

Similarly, if $l_2 > k_j + 1$, then for $l = k_j + 1, \ldots, l_2 - 1$ we have $\hat{\theta}_l = \hat{\theta}_{k_j + 1}$ and hence $(\bar{Y}_l - \hat{\theta}_l)^2 \geq ( \bar{Y}_{k_j + 1} - \hat{\theta}_{l})^2 = ( \bar{Y}_{k_j+1} - \hat{\theta}_{k_j +1})^2$.
For such $l$, it follows that $\hat{\theta}_{l_2} > \bar{Y}_l$ and therefore that $( \bar{Y}_l - \hat{\theta}_{l_2} )^2 \leq (\bar{Y}_{k_j + 1} - \hat{\theta}_{l_2})^2$. 

Now, we define
\begin{align*}
	\tilde{w}_{k_j} &:= \sum_{l \leq k_j \colon \hat{\theta}_l = \hat{\theta}_{k_j}} w_l \\
	\llap{\text{ and, if $l_2 > k_j + 1$, } \quad} \tilde{w}_{k_j + 1} &:= \sum_{l \geq k_{j} + 1 \colon \hat{\theta}_l = \hat{\theta}_{k_j + 1}} w_l.
\end{align*}
We also define $\tilde{\mbb\theta} \in \R^K$ according to
		\[ 
	\tilde{\theta}_l = \begin{cases}
		\hat{\theta}_l \wedge \hat{\theta}_{l_1} & \text{ for } l \leq k_j \\
		\hat{\theta}_l \vee \hat{\theta}_{l_2} & \text{ for } l > k_j .
	\end{cases}
\]

We note that by assumption, both $\tilde{w}_{k_j} < 1 / \eta s$ and $\tilde{w}_{k_j + 1} < 1 / \eta s$. 
We now consider two cases: (A) where $l_2 = k_j + 1$, so $\hat{\theta}_{k_j + 1} \geq\bar{Y}_{k_j + 1}$, and (B) where $l_2 > k_j + 1$, so $\hat{\theta}_{k_j + 1} < \bar{Y}_{k_j + 1}$.
 
 We first consider case (A), where the penalty terms between $l_1$ and $l_2$ in $Q( \hat{\mbb\theta})$ are
\begin{align*}
	\sum_{ l = l_1}^{l_2 - 1} \rho ( \hat{\theta}_{l + 1} - \hat{\theta}_l ) = \rho ( \hat{\theta}_{l_2} - \hat{\theta}_{k_j}) + \rho ( \hat{\theta}_{k_j} - \hat{\theta}_{l_1} ).
\end{align*}
Thus,
		\begin{align}
		Q(\hat{\mbb\theta}) - Q(\tilde{\mbb\theta}) = &  \sum_{l \leq k_j \colon \hat{\theta}_l = \hat{\theta}_{k_j}} \frac{w_l}{2} (\bar{Y}_l - \hat{\theta}_l)^2 -  \sum_{l \leq k_j \colon \hat{\theta}_l = \hat{\theta}_{k_j}} \frac{w_l}{2} (\bar{Y}_l - \hat{\theta}_{l_1})^2 \notag \\
		& +  \rho ( \hat{\theta}_{l_2} - \hat{\theta}_{k_j} ) + \rho(\hat{\theta}_{k_j} - \hat{\theta}_{l_1} ) - \frac{1}{2} \gamma \lambda^2 \notag \\[1.5ex]
		\geq & \frac{\tilde{w}_{k_j}}{2} ( \bar{Y}_{k_j} - \hat{\theta}_{k_j})^2 - \frac{\tilde{w}_{k_j}}{2} ( \bar{Y}_{k_j} - \hat{\theta}_{l_1})^2 \notag \\
		& + \rho ( \hat{\theta}_{l_2} - \hat{\theta}_{k_j} ) + \rho(\hat{\theta}_{k_j} - \hat{\theta}_{l_1} ) - \frac{1}{2} \gamma \lambda^2 \label{eq:oraclelsecaseaeq1} \\[1.5ex]
		\geq &\inf_{ \bar{Y}_{k_j} < a \leq \hat{\theta}_{l_2} } \frac{\tilde{w}_{k_j}}{2} ( \bar{Y}_{k_j} - a)^2 - \frac{\tilde{w}_{k_j}}{2} ( \bar{Y}_{k_j} - \hat{\theta}_{l_1})^2  \notag \\ 
		& + \rho ( \hat{\theta}_{l_2} - a ) + \rho(a - \hat{\theta}_{l_1} ) - \frac{1}{2} \gamma \lambda^2 \label{eq:oraclelsecaseaeq2}.
		\end{align}
We specify the infimum in \eqref{eq:oraclelsecasea1eq2} because $( \bar{Y}_{k_j} , \hat{\theta}_{l_2}]$ is not closed, and let $(a_m)$ be a convergent sequence in $( \bar{Y}_{k_j}, \hat{\theta}_{l_2}]$ whose limit attains this infimum. We define $a^* = \lim_{m \rightarrow \infty} a_m$.

By assumption \eqref{eq:oraclelselemsigsep}, at least one of $(a^* - \hat{\theta}_{l_1})$ and $( \hat{\theta}_{l_2} - a^*)$ is greater than or equal to $\gamma \lambda$.
Here, we use that the separation $\eqref{eq:oraclelselemsigsep} \geq 2 \gamma \lambda$.
If $\hat{\theta}_{l_2} -a^* \geq \gamma \lambda$ then we denote this case (A1) and \eqref{eq:oraclelsecaseaeq2} becomes
\begin{align}
	Q(\hat{\mbb\theta}) - Q(\tilde{\mbb\theta}) &\geq \inf_{ \bar{Y}_{k_j} < a \leq \hat{\theta}_{l_2} - \gamma \lambda} \frac{\tilde{w}_{k_j}}{2} ( \bar{Y}_{k_j} - a)^2 - \frac{\tilde{w}_{k_j}}{2} ( \bar{Y}_{k_j} - \hat{\theta}_{l_1})^2 + \rho ( a- \hat{\theta}_{l_1}) \label{eq:oraclelsecasea1eq1} \\
	&\geq \min_{ \hat{\theta}_{l_1} \leq \tilde a \leq \hat{\theta}_{l_2} - \gamma \lambda} \frac{\tilde{w}_{k_j}}{2} ( \bar{Y}_{k_j} - \tilde a)^2 - \frac{\tilde{w}_{k_j}}{2} ( \bar{Y}_{k_j} - \hat{\theta}_{l_1})^2 + \rho ( \tilde a- \hat{\theta}_{l_1}) \label{eq:oraclelsecasea1eq2}.
\end{align}
We define $ \tilde a^*$ to be the minimiser over $\tilde{a}$ of \eqref{eq:oraclelsecasea1eq2}. We can observe that since $\bar{Y}_{k_j} - \hat{\theta}_{l_1} < \zeta$ and ${ \zeta < ( 1 \wedge \sqrt{\gamma \tilde{w}_{k_j}} ) \lambda / \tilde{w}_{k_j} }$, we have ${ \bar{Y}_{k_j} - \hat{\theta}_{l_1} < ( 1 \wedge \sqrt{\gamma \tilde{w}_{k_j}} ) \lambda / \tilde{w}_{k_j} }$. Thus, we have by Lemma~\ref{thm:quadoptlem} that the uniquely optimal $\tilde a^*= \hat{ \theta}_{l_1}$. This gives that the value of \eqref{eq:oraclelsecasea1eq2} is zero.

It is straightforward to see from \eqref{eq:oraclelsecasea1eq1} that $a^* = \bar{Y}_{k_j}$ must be the unique limit of $(a_m)$.
As we have assumed that $\hat{\theta}_{k_j} > \bar{Y}_{k_j}$ and the infimum is not attained in $(\bar{Y}_{k_j}, \bar{Y}_{k_j + 1} )$, the inequality in line \eqref{eq:oraclelsecasea1eq1} can be made strict.
It follows that $Q(\hat{\mbb\theta}) > Q( \tilde{\mbb\theta})$.

Thus, it remains for us to consider the case where $ \hat{\theta}_{l_2} - a^* < \gamma \lambda$, which implies that $a^* - \hat{\theta}_{l_1} \geq \gamma \lambda$.
We denote this case (A2). 
Now, from \eqref{eq:oraclelsecaseaeq2} we can obtain
\begin{align}
		Q(\hat{\mbb\theta}) - Q(\tilde{\mbb\theta}) \geq \min_{\hat{\theta}_{l_2} - \gamma \lambda < \tilde a \leq \hat{\theta}_{l_2} } \frac{\tilde{w}_{k_j}}{2} ( \bar{Y}_{k_j} - \tilde a)^2 - \frac{\tilde{w}_{k_j}}{2} ( \bar{Y}_{k_j} - \hat{\theta}_{l_1})^2 + \rho ( \hat{\theta}_{l_2} - \tilde a) \label{eq:oraclelsecasea2eq1}.
\end{align}

The objective is piecewise quadratic (and continuously differentiable), with two pieces: $[ \hat{\theta}_{l_1}, \hat{\theta}_{l_2} - \gamma \lambda ]$ and $( \hat{\theta}_{l_2} - \gamma \lambda, \hat{\theta}_{l_2}]$.
On the first region, the objective is a convex quadratic with minimum at $\bar{Y}_{k_j} \in [ \hat{\theta}_{l_1}, \hat{\theta}_{l_2} - \gamma \lambda ]$. 

By the assumption that $a^* > \hat{\theta}_{l_2} - \gamma \lambda $, we know that the objective must be concave on $(\hat{\theta}_{l_2} - \gamma \lambda, \hat{\theta}_{l_2} ] $.
It is clear that the derivative of the objective at $\hat{\theta}_{l_2} - \gamma \lambda$ is positive.
Hence, if $ \tilde a^*= \hat{\theta}_{l_2} - \gamma \lambda$, then the objective will take a strictly lower value at some $\tilde a^* \in ( \hat{\theta}_{l_2} - \gamma \lambda - \epsilon, \hat{\theta}_{l_2} - \gamma \lambda)$ (for some small $\epsilon > 0$), contradicting optimality of $\tilde a^*$.
It therefore follows that $\tilde a^* = \hat{\theta}_{l_2}$.

With this knowledge, we can further simplify \eqref{eq:oraclelsecasea2eq1} to obtain
\begin{align*}
			Q(\hat{\mbb\theta}) - Q(\tilde{\mbb\theta}) \geq  \frac{\tilde{w}_{k_j}}{2} ( \bar{Y}_{k_j} - \hat{\theta}_{l_2})^2 - \frac{\tilde{w}_{k_j}}{2} ( \bar{Y}_{k_j} - \hat{\theta}_{l_1})^2  > 0.
\end{align*}
The second inequality follows from $\bar{Y}_{k_j} - \hat{\theta}_{l_1} \leq \zeta$ and $\hat{\theta}_{l_2} - \bar{Y}_{k_j} > \zeta$.
Hence, we obtain that ${ Q( \hat{\mbb\theta}) > Q(\tilde{\mbb\theta}) }$.

We now we direct our attention towards case (B), where similarly to before we observe that the penalty contributions between $l_1$ and $l_2$ in $Q(\hat{\mbb\theta})$ are
\begin{align*}
	\sum_{ l = l_1}^{l_2 - 1} \rho ( \hat{\theta}_{l + 1} - \hat{\theta}_l ) = \rho ( \hat{\theta}_{l_2} - \hat{\theta}_{k_j + 1}) + \rho ( \hat{\theta}_{k_j + 1} - \hat{\theta}_{k_j} ) + \rho ( \hat{\theta}_{k_j} - \hat{\theta}_{l_1} ).
\end{align*}
Similarly to \eqref{eq:oraclelsecaseaeq1} in case (A), we obtain
		\begin{align}
			Q(\hat{\mbb\theta}) - Q(\tilde{\mbb\theta}) \geq & \frac{\tilde{w}_{k_j}}{2} ( \bar{Y}_{k_j} - \hat{\theta}_{k_j})^2 + \frac{\tilde{w}_{k_j + 1}}{2} ( \bar{Y}_{k_j + 1} - \hat{\theta}_{k_j + 1})^2 \notag \\
						&-   \frac{\tilde{w}_{k_j}}{2}( \bar{Y}_{k_j} - \hat{\theta}_{l_1})^2 -  \frac{\tilde{w}_{k_j + 1}}{2}( \bar{Y}_{k_j + 1} - \hat{\theta}_{l_2})^2 \notag \\			
			&+  \rho ( \hat{\theta}_{l_2} - \hat{\theta}_{k_j + 1}) + \rho ( \hat{\theta}_{k_j + 1} - \hat{\theta}_{k_j} ) + \rho ( \hat{\theta}_{k_j} - \hat{\theta}_{l_1} ) - \frac{1}{2} \gamma \lambda^2   \label{eq:oraclelsecasebeq1}\\[1.5ex]
			\geq & \inf_{ \bar{Y}_{k_j} < a \leq b < \bar{Y}_{k_j + 1} }  \frac{\tilde{w}_{k_j}}{2} ( \bar{Y}_{k_j} - a)^2 + \frac{\tilde{w}_{k_j + 1}}{2} ( \bar{Y}_{k_j + 1} - b)^2 \notag \\
						&-   \frac{\tilde{w}_{k_j}}{2}( \bar{Y}_{k_j} - \hat{\theta}_{l_1})^2 -  \frac{\tilde{w}_{k_j + 1}}{2}(  \bar{Y}_{k_j + 1} - \hat{\theta}_{l_2} )^2  \notag \\		
			&+  \rho ( \hat{\theta}_{l_2} - b) + \rho ( b - a ) + \rho ( a - \hat{\theta}_{l_1} ) - \frac{1}{2} \gamma \lambda^2. \label{eq:oraclelsecasebeq2}
		\end{align}
 We specify the infimum in \eqref{eq:oraclelsecasebeq2} because $(\bar{Y}_{k_j}, \bar{Y}_{k_j + 1} )$ is not closed and therefore a minimum may not exist. Let $( a_m, b_m )$ be a convergent sequence in $(\bar{Y}_{k_j}, \bar{Y}_{k_j + 1} )$ whose limit achieves this infimum. We now define $( a^*, b^*) = \lim_{m \rightarrow \infty} (a_m, b_m)$.
By assumption \eqref{eq:oraclelselemsigsep}, we know that $\bar{Y}_{k_j + 1} - \bar{Y}_{k_j} \geq 3 \gamma \lambda$, which implies that $\hat{\theta}_{l_2} - \hat{\theta}_{l_1} \geq 3 \gamma \lambda$. 
Thus, one of $ \{ ( \hat{\theta}_{l_2} -  b^*) , { ( b^* - a^* ) } , {( a^* - \hat{\theta}_{l_1}) } \}$ must be at least $\gamma \lambda$.

We first consider if $ b^* - a^* \geq \gamma \lambda$, and denote this case (B1). Here, \eqref{eq:oraclelsecasebeq2} becomes
\begin{align}
	Q( \hat{\mbb\theta}) - Q( \tilde{\mbb\theta}) \geq &\inf_{ \bar{Y}_{k_j} < a \leq b < \bar{Y}_{k_j + 1} }  \frac{\tilde{w}_{k_j}}{2} ( \bar{Y}_{k_j} - a)^2 + \frac{\tilde{w}_{k_j + 1}}{2} ( \bar{Y}_{k_j + 1} - b)^2 \notag \\
	&-   \frac{\tilde{w}_{k_j}}{2}( \bar{Y}_{k_j} - \hat{\theta}_{l_1})^2 -  \frac{\tilde{w}_{k_j + 1}}{2}(  \bar{Y}_{k_j + 1} - \hat{\theta}_{l_2} )^2		
			+  \rho ( \hat{\theta}_{l_2} - b)  + \rho ( a - \hat{\theta}_{l_1} ) \label{eq:oraclelsecaseb1eq1}\\[1.5ex]
			= &\inf_{ a \in ( \bar{Y}_{k_j}, \bar{Y}_{k_j + 1}) }  \frac{\tilde{w}_{k_j}}{2} ( \bar{Y}_{k_j} - a)^2  - \frac{\tilde{w}_{k_j}}{2} ( \bar{Y}_{k_j} - \hat{\theta}_{l_1})^2 + \rho ( a - \hat{\theta}_{l_1}) \notag \\
			&+ \inf_{b \in ( \bar{Y}_{k_j}, \bar{Y}_{k_j + 1})}   \frac{\tilde{w}_{k_j + 1}}{2} ( \bar{Y}_{k_j + 1} - b)^2  - \frac{\tilde{w}_{k_j + 1}}{2} ( \bar{Y}_{k_j + 1} - \hat{\theta}_{l_2})^2 + \rho ( \hat{\theta}_{l_2} - b)
			 \label{eq:oraclelsecaseb1eq2}	
\end{align}
We can observe that \eqref{eq:oraclelsecaseb1eq2} is the sum of two copies of \eqref{eq:oraclelsecasea1eq1} in case (A1).
Hence, by following the same arguments as before, we see that $Q( \hat{\mbb\theta}) > Q( \tilde{\mbb\theta})$.

It therefore remains for us to obtain the result in the case that $ b^* - a^* < \gamma \lambda$, and we denote this case (B2).
Using that the separation $\eqref{eq:oraclelselemsigsep} \geq 3 \gamma \lambda + 2 \zeta$, it is straightforward to see that one of $(\bar{Y}_{k_j + 1} - b^*)$ and $( a^* - \bar{Y}_{k_j})$ must be at least $\gamma \lambda + \zeta$.
By the symmetry of the problem, it is sufficient for us to consider the case where $\bar{Y}_{k_j + 1} - b^* \geq \gamma \lambda + \zeta$.
In this case, we can obtain from \eqref{eq:oraclelsecasebeq2} that
\begin{align}
	Q( \hat{\mbb\theta} ) - Q ( \tilde{\mbb\theta} ) \geq \min_{	(\tilde a, \tilde b) \in \mathcal{B} }  & \frac{\tilde{w}_{k_j}}{2} ( \bar{Y}_{k_j} - \tilde a)^2 + \frac{\tilde{w}_{k_j + 1}}{2}( \bar{Y}_{k_j + 1} - \tilde b)^2 \notag \\
	&- \frac{\tilde{w}_{k_j}}{2} ( \bar{Y}_{k_j} - \hat{\theta}_{l_1})^2 - \frac{\tilde{w}_{k_j + 1}}{2}( \bar{Y}_{k_j + 1} - \hat{\theta}_{l_2})^2 \notag \\
	& + \rho ( \tilde b - \tilde a ) + \rho ( \tilde a - \hat{\theta}_{l_1} ), \label{eq:oraclelsecaseb2eq1}
\end{align}
where $\mathcal{B} =  \left\{ (\tilde a, \tilde b) \colon \hat{\theta}_{l_1} \leq \tilde a \leq \tilde b \leq \bar{Y}_{k_j + 1} - \gamma \lambda - \zeta, \: \tilde b - \tilde a < \gamma \lambda \right\}$.
From this, we can extract the terms dependent on $\tilde b$ to obtain 
\begin{align}
	\tilde b^* =\argmin_{\tilde a^* \leq \tilde b < \tilde a^* + \gamma \lambda} \frac{\tilde{w}_{k_j + 1}}{2} ( \bar{Y}_{k_j + 1} - \tilde b)^2 + \rho ( \tilde b - \tilde a^* ). \label{eq:oraclelsefusecase}
\end{align}
This objective is piecewise quadratic (and continuously differentiable), with two pieces; $[\tilde a^*, \tilde a^* + \gamma \lambda )$ and $[\tilde a^* + \gamma \lambda, \hat{\theta}_{l_2}]$.
Over the second region, the objective is a convex quadratic with minimum at $\bar{Y}_{k_j + 1} \in [\tilde a^* + \gamma \lambda, \hat{\theta}_{l_2} ]$.
By following the same argument as for \eqref{eq:oraclelsecasea2eq1} in case (A2), we see that $\tilde b^* = \tilde a^*$.

With this knowledge, we can further simplify \eqref{eq:oraclelsecaseb2eq1} to obtain
\begin{align*}
	Q( \hat{\mbb\theta} ) - Q ( \tilde{\mbb\theta} ) \geq \min_{\hat{\theta}_{l_1} \leq \tilde a \leq \bar{Y}_{k_j + 1} - \gamma \lambda - \zeta}  & \frac{\tilde{w}_{k_j}}{2} ( \bar{Y}_{k_j} - \tilde a)^2 + \frac{\tilde{w}_{k_j + 1}}{2}( \bar{Y}_{k_j + 1} - \tilde a)^2 \\
	&- \frac{\tilde{w}_{k_j}}{2} ( \bar{Y}_{k_j} - \hat{\theta}_{l_1})^2 - \frac{\tilde{w}_{k_j + 1}}{2}( \bar{Y}_{k_j + 1} - \hat{\theta}_{l_2})^2 + \rho ( \tilde a - \hat{\theta}_{l_1} ).
\end{align*}
Since $ \bar{Y}_{k_j + 1 } - \tilde a^* > \zeta$, we can see that $(\bar{Y}_{k_j + 1}- \tilde a^*)^2 - ( \bar{Y}_{k_j + 1} - \hat{\theta}_{l_2} )^2 > 0$.
Thus, it suffices for us to show that 
\begin{align*}
	\min_{\hat{\theta}_{l_1} \leq \tilde a \leq \bar{Y}_{k_j + 1} - \gamma \lambda - \zeta}  & \frac{\tilde{w}_{k_j}}{2} ( \bar{Y}_{k_j} - \tilde a)^2  - \frac{\tilde{w}_{k_j}}{2} ( \bar{Y}_{k_j} - \hat{\theta}_{l_1} )^2 + \rho ( \tilde a - \hat{\theta}_{l_1} )  \geq 0.
\end{align*}
This objective is exactly as in \eqref{eq:oraclelsecasea1eq2} in case (A1), minimised over a smaller feasible set. 
Hence, it follows immediately that this holds and we can conclude that $Q ( \hat{\mbb\theta} ) > Q ( \tilde{\mbb\theta})$.

We now have for all cases that $Q(\hat{\mbb\theta}) > Q( \tilde{\mbb\theta})$, which contradicts the optimality of $\hat{\theta}$.
 Thus, we can conclude that for $j = 1, \ldots, s$, $\hat{\theta}_{k_j} \leq \bar{Y}_{k_j}$ and $\hat{\theta}_{k_{j - 1} + 1} \geq \bar{Y}_{k_{j-1}+1}$.
\end{proof}

\begin{lem} \label{thm:betalexist}
Consider the setup of Lemma~\ref{thm:oraclelselemma}.
For each $j = 1, \ldots, s$, there exists $k^*_j$ in $\{ k_{j - 1}+ 1, \ldots, k_j \}$ such that $\hat{\theta}_{k^*_j} \in [ \bar{Y}_{k_{j - 1} + 1}, \bar{Y}_{k_j} ]$.
\end{lem}
\begin{proof}
	We first show that if $\hat{\theta}_{k_j} > \bar{Y}_{k_j}$, then for any $k$ with $k_{j - 1} + 1 \leq k \leq k_j$, if $\hat{\theta}_k > \bar{Y}_{k_j}$ then $\hat{\theta}_k = \hat{\theta}_{k_j}$. 

	We prove the first case since the proof for the second is identical.
	Suppose that this does not hold, i.e.\ $\hat{\theta}_{k_j} > \bar{Y}_{k_j}$ and there exists some (minimal) $k$ in $\{ k_{j-1} + 1, \ldots, k_j - 1 \}$ with $\bar{Y}_{k_j} < \hat{\theta}_k < \hat{\theta}_{k_j} $.
	Then we construct $\check{\mbb\theta}$ by 
\begin{equation}
	\check{\theta}_l = \begin{cases}
		\hat{\theta}_k & \text{ for } l = k, k+1, \ldots, k_j \\
		\hat{\theta}_l & \text{ otherwise. } 
	\end{cases} \label{eq:2pointsinmid}
\end{equation}
	We observe that the penalty contribution from $\check{\mbb\theta}$ is no more than that of $\hat{\mbb\theta}$ and that the quadratic loss for $\check{\mbb\theta}$ will be strictly less than that of $\hat{\mbb\theta}$. 
	This gives us that $Q(\check{\mbb\theta}) < Q(\hat{\mbb\theta})$, contradicting the optimality of $\hat{\mbb\theta}$.
	
	Similarly, if $\hat{\theta}_{k_{j-1} + 1} < \bar{Y}_{k_{j-1}+1}$ then the corresponding statement that for any $k$ with $k_{j - 1} + 1 \leq k_j$, if $\hat{\theta}_k < \bar{Y}_{k_{j-1}+1}$ then $\hat{\theta}_k = \hat{\theta}_{k_{j-1}+1}$.
	
	We now establish a simple preliminary result.
	Suppose that for some $j$ in $\{ 1, \ldots , s \}$ there exists $k$ in $\{ k_{j - 1} + 1, \ldots , k_j \}$ with $\hat{\theta}_k \notin [ \bar{Y}_{k_{j - 1} + 1} , \bar{Y}_{k_j} ]$, such that $\sum_{\{ l \colon \hat{\theta}_l = \hat{\theta}_k \} }w_l \geq \eta / 2 s$.
	We claim that if $\hat{\theta}_k > \bar{Y}_{k_j}$ then $\hat{\theta}_k \leq \bar{Y}_{k_j} +  ( \sqrt{2 s / \eta} \sqrt{\gamma} \lambda \vee \gamma \lambda)$. 
	Similarly, if $\hat{\theta}_k < \bar{Y}_{k_{j-1} + 1 }$ then ${ \hat{\theta}_k \geq \bar{Y}_{k_{j-1}+1} - ( \sqrt{2 s / \eta} \sqrt{\gamma} \lambda \vee \gamma \lambda) }$.
	
	To prove the claim, we consider the case $\hat{\theta}_k > \bar{Y}_{k_j}$ (the other is identical). 
	By the first observation, if $\hat{\theta}_l > \bar{Y}_{k_j}$ for $l$ in $\{ k_{j - 1} + 1 , \ldots k_j \} $ then $\hat{\theta}_l = \hat{\theta}_k$.
	Now, for contradiction, suppose $\hat{\theta}_k > \bar{Y}_{k_j } + ( \sqrt{2 s / \eta} \sqrt{\gamma} \lambda \vee \gamma \lambda)$ and let this $k$ be minimal.
	Then we can construct $\check{\mbb\theta}$ by
	\[
	\check{\theta}_l = \begin{cases}
		\sum_{l = k }^{k_j} w_l \bar{Y}_l / \sum_{ l = k }^{k_j} w_l &\text{ for } l = k, \ldots, k_j \\
		\hat{\theta}_l &\text{ otherwise}.
	\end{cases}
	\]
		By appealing to the optimality of $\hat{\mbb\theta}$, we can easily observe that $\hat{\theta}_{k - 1} \leq \bar{Y}_{k - 1}$ and therefore that the ordering of the entries of $\check{\mbb\theta}$ matches that of $\hat{\mbb\theta}$. Here, we use that $( \sqrt{2 s / \eta} \sqrt{\gamma} \lambda \vee \gamma \lambda) \geq \gamma \lambda$.
		
		We can now see that the loss term in $Q(\check{\mbb\theta})$ is less than in $Q(\hat{\mbb\theta})$, with a difference of more than $(\eta / 4 s ) ( \sqrt{2 s / \eta} \sqrt{\gamma} \lambda)^2 = \gamma \lambda^2 / 2$, which outweighs the possible increase in the penalty contribution.
		This gives us that $Q(\check{\mbb\theta}) < Q(\hat{\mbb\theta})$, contradicting the optimality of $\hat{\mbb\theta}$.
		
		We now return to the proof of the main result. Suppose, for contradiction, that there exists some $j \in \{ 1, \ldots, s \}$ such that $\hat{\theta}_k \notin [ \bar{Y}_{k_{j-1} + 1}, \bar{Y}_{k_j} ]$ for all ${k = k_{j-1}+1 , \ldots, k_j}$ and let this $j$ be minimal.
		By the first observation, we know that entries of $\hat{\mbb\theta}$ corresponding to level $j$ can take one of at most two distinct values.
		That is, for $k \in \{ k_{j - 1} + 1, \ldots , k_j \}$, 
		if we have $\hat{\theta}_{k} < \bar{Y}_{k_{j -1} + 1}$, then it follows that $\hat{\theta}_k = \hat{\theta}_{k_{j-1} + 1}$.
		Similarly, if $\hat{\theta}_{k} > \bar{Y}_{k_j}$, then $\hat{\theta}_k = \hat{\theta}_{k_j}$.
		
		By the assumption $w^0_{\text{min}} \geq \eta / s$, we have that either
		\[
		\sum_{ k \colon \hat{\theta}_k = \hat{\theta}_{k_{j - 1} + 1} } w_k \geq \frac{\eta}{2 s} \quad \text{or} \quad \sum_{ k \colon \hat{\theta}_k = \hat{\theta}_{k_j} } w_k \geq \frac{\eta}{2 s}.
		\]
		We will without loss of generality take the second statement to be true (the proof for the first case follows identically).
		Let $k'$ denote the minimal element in $ \{ k_{j -1} + 1, \ldots ,k_j \} $ such that $\hat{\theta}_{k'} = \hat{\theta}_{k_j}$.
		From the preliminary result established earlier, $\hat{\theta}_{k_j} \leq \bar{Y}_{k_j} +( \sqrt{2 s / \eta} \sqrt{\gamma} \lambda \vee \gamma \lambda)$.
		By appealing to the optimality of $\hat{\mbb\theta}$, we see that $\hat{\theta}_{k_j + 1} < \hat{\theta}_{k_j} + \gamma \lambda$ (otherwise, we could take $\hat{\theta}_{k_j}$ to be $\bar{Y}_{k_j}$ and strictly reduce the value of the objective).
		
		Now, we will use that the separation is at least $ 2 ( \sqrt{2 s / \eta} \sqrt{\gamma} \lambda \vee \gamma \lambda) + \gamma \lambda$.
		By our earlier observation \eqref{eq:2pointsinmid}, it is clear that any $l \in \{ k_j + 1, \ldots , k_{j+1} \}$ with $\hat{\theta}_l < \bar{Y}_{k_j + 1}$ has $\hat{\theta}_l = \hat{\theta}_{k_j + 1}$.
		Note that since $\hat{\theta}_{k_j + 1} - \bar{Y}_{k_j} < ( \sqrt{2 s / \eta} \sqrt{\gamma} \lambda \vee \gamma \lambda) + \gamma \lambda$, it follows that $\bar{Y}_{k_j + 1} - \hat{\theta}_{k_j + 1} > ( \sqrt{2 s / \eta} \sqrt{\gamma} \lambda \vee \gamma \lambda) + \zeta$ and therefore that $\sum_{ \{ k \colon \hat{\theta}_k = \hat{\theta}_{k_j + 1} \} } w_k < \eta / 2 s$ by the preliminary result.
		Since $w^0_{\text{min}} \geq \eta / s$ and separation $\eqref{eq:oraclelselemsigsep} \geq 2 ( \sqrt{2 s / \eta} \sqrt{\gamma} \lambda \vee \gamma \lambda) + \gamma \lambda + \zeta$, we can define $l' \in \{ k_j + 1, \ldots , k_{j+1} \}$ minimal such that $\hat{\theta}_{l'} \geq \bar{Y}_{k_j + 1}$.

		Now, in order to contradict the optimality of $\hat{\mbb\theta}$ we construct a new feasible point $\tilde{\mbb\theta}$ by setting
		\[
		\tilde{\theta}_{l} = 
		\begin{cases}
			\bar{Y}_{k_j} &\text{ for } l = k', \ldots, k_j \\
			\hat{\theta}_{l'} &\text{ for } l = k_j + 1, \ldots, l' - 1 \\
			\hat{\theta}_{l} &\text{ otherwise}.
		\end{cases}
		\]
				It follows that for $l = k_j + 1, \ldots, l' - 1$ we have 
		 \begin{align*}
		 	| \hat{\theta}_{l} - \bar{Y}_{l} | &> ( \sqrt{2 s / \eta} \sqrt{\gamma} \lambda \vee \gamma \lambda) + \zeta \\
		 	| \tilde{\theta}_{l} - \bar{Y}_{l} | &\leq ( \sqrt{2 s / \eta} \sqrt{\gamma} \lambda \vee \gamma \lambda) + \zeta.
		 \end{align*}
		 It is also straightforward to see that $|\hat{\theta}_{k_j} - \bar{Y}_{l} | \geq |\bar{Y}_{k_j} - \bar{Y}_{l}|$ for $l = k', \ldots, k_j$.
		 If follows that the loss contribution in $Q( \tilde{\mbb\theta})$ is strictly less than that in $Q( \hat{\mbb\theta})$.
		Hence, using $\hat{\theta}_{l'} - \hat{\theta}_{k_j} > \gamma \lambda$, we obtain
		\begin{align*}
			Q( \hat{\mbb\theta}) - Q( \tilde{\mbb\theta}) > &\rho ( \hat{\theta}_{l'} - \hat{\theta}_{k_j + 1} ) + \rho ( \hat{\theta}_{k_j + 1} - \hat{\theta}_{k_j}) + \rho ( \hat{\theta}_{k_j } - \hat{\theta}_{k' - 1} ) \\
			&- \frac{1}{2} \gamma \lambda^2 - \rho ( \bar{Y}_{k_j} - \hat{\theta}_{k' - 1}) \\
			\geq& 0,
		\end{align*}
		contradicting the optimality of $\hat{\mbb\theta}$.
	We conclude that for $j = 1, \ldots, s$, there exists $k^*_j$ in $\{ k_{j - 1}+ 1, \ldots, k_j \}$ such that $\hat{\theta}_{k^*_j} \in [ \bar{Y}_{k_{j - 1} + 1}, \bar{Y}_{k_j} ]$.
		\end{proof}

\begin{lem} \label{thm:nullconsistency}
Consider the univariate objective \eqref{eq:univar}, relaxing the normalisation constraint to $\check{w} := \sum_k w_k \leq 1$.
Suppose that $w^T\bar{Y} = 0$, and that $	\| \bar{Y} \|_{\infty} < \left( 2 \wedge \sqrt{\gamma \check{w} } \right) \lambda / \check{w} $. Then $\hat{\mbb\theta} = 0 $.
\end{lem}
\begin{proof}
Let $P_w = I -  \mb 1  w^T / \check{w}$ and $D_w \in \R^{K \times K}$ be the diagonal matrix with entries $D_{kk} \sqrt{w_k}$. First note that
\begin{align*}
	Q(\mbb\theta) - Q(P_w\mbb\theta) &= \frac{1}{2}\sum_{k=1}^K w_k (\bar{Y}_k - \theta_k)^2 - \frac{1}{2}\sum_{k=1}^K w_k (\bar{Y}_k - \theta_k + w^T \mbb\theta)^2 \\ 
	&= -\frac{1}{2}\sum_{k=1}^K w_k (w^T \mbb\theta) ( 2 \bar{Y}_k - 2\theta_k + w^T \mbb\theta) \\
	&= \left( 1 - \frac{1}{2} \check{w} \right) (w^T \mbb\theta)^2 \; \geq 0.
\end{align*}
Thus for all $\mbb\theta\in \R^K$, we have
	\begin{align*}
		Q(\mbb\theta) - Q(0) &\geq \frac{1}{2} \| D_w P_w ( \bar{Y} - \mbb\theta) \|_2^2 - \frac{1}{2} \| D_w P_w  \bar{Y} \|_2^2 + \sum_{k =1}^{K-1} \rho( \theta_{(k+1)} - \theta_{(k)}) \\ 
		&\geq \frac{1}{2} \| D_w P_w ( \bar{Y} - \mbb\theta) \|_2^2 - \frac{1}{2} \| D_w P_w  \bar{Y} \|_2^2 +  \rho( \theta_{(K)} - \theta_{(1)}) \\
		&\geq \min_{\xi \in [-\tau, \tau]^K} F(\mbb\theta, \xi, w) 
	\end{align*}
where
\[
F(\mbb\theta, \xi, w) = \frac{1}{2} \| D_w P_w ( \xi - \mbb\theta) \|_2^2 - \frac{1}{2} \| D_w P_w  \xi \|_2^2 + \rho ( \theta_{(K)} - \theta_{(1)}).
\]
Consider minimising $F$ over $\R^K \times [-\tau, \tau]^K \times S$, where $S \subseteq \R^K$ is the unit simplex scaled by $\check{w}$. We aim to show this minimum is $0$. As with the first claim in the proof of Lemma~\ref{thm:oraclelselemma}, it is straightforward to see that for any feasible $(\mbb\theta, \xi, w)$, there exists $\mbb\theta'$ with $\|\mbb\theta'\|_\infty \leq \|\xi\|_\infty$ and $F(\mbb\theta', \xi, w) \leq F(\mbb\theta, \xi, w)$.
Hence,
\[
\inf_{(\mbb\theta, \xi, w) \in \R^K \times [-\tau, \tau]^K \times S} F(\mbb\theta, \xi, w) = \inf_{(\mbb\theta, \xi, w) \in [-\tau, \tau]^K \times [-\tau, \tau]^K \times S} F(\mbb\theta, \xi, w).
\]
As on the RHS we are minimising a continuous function over a compact set, we know a minimiser must exist. Let $(\tilde{\mbb\theta}, \tilde{\xi}, \tilde{w})$ be a minimiser (to be specified later). 
Observe that 
\[
\| D_{\tilde{w}}  P_{\tilde{w}} ( \xi -  \mbb\theta) \|_2^2 - \| D_{\tilde{w}}P_{\tilde{w}}  \xi \|_2^2 = - 2\xi^T P_{\tilde{w}}^T D_{\tilde{w}}^2 P_{\tilde{w}}  \mbb\theta+ \theta^T P_{\tilde{w}}^T D_{\tilde{w}}^2 P_{\tilde{w}} \mbb\theta
\]
is linear as a function of $\xi$. Hence it is minimised over the set
\[
\{\xi: \| \xi \|_{\infty} \leq \tau\} = \text{conv}(\{ -\tau, \tau \}^K)
\]
at some point in $\{ -\tau, \tau \}^K$. Here $\text{conv}(\cdot)$ denotes the convex hull operation. We thus have
\[
Q(\mbb\theta) - Q(0) \geq \min_{\xi \in  \{-\tau, \tau\}^K} \frac{1}{2} \| D_{\tilde{w}} P_{\tilde{w}}( \xi -  \mbb\theta) \|_2^2 - \frac{1}{2} \| D_{\tilde{w}} P_{\tilde{w}} \xi \|_2^2 + \rho ( \theta_{(K)} - \theta_{(1)}).
\]
Let us take $(\tilde{\mbb\theta}, \tilde{\xi}) \in \R^K \times \{-\tau, \tau\}^K$ to be a minimiser of the RHS.

	Note that if we have $\tilde{\xi}_{j} = \tilde{\xi}_{k}$ then we may take $\tilde{\theta}_{j} = \tilde{\theta}_{k}$. Indeed, we may construct $\check{\mbb\theta}\in \R^K$ by setting 
	\[
	\check{\theta}_l = \begin{cases}
		\argmin_{b \in \{ \tilde{\theta}_j, \tilde{\theta}_k \} } ( \tilde{\xi}_j - b)^2 &\text{ for } l = j, k \\
		\tilde{\theta}_l &\text{ otherwise}.
	\end{cases}
	\]
	Since the penalty contribution from $\check{\mbb\theta}$ is not greater than that of $\tilde{\mbb\theta}$, it follows that $Q(\check{\mbb\theta} ) \leq Q( \tilde{\mbb\theta})$.
	Thus, we can assume that entries of $\tilde{\mbb\theta}$ can take one of only two distinct values.

Next we write $\tilde{\alpha} = \sum_{k : \tilde{\xi}_k = - \tau} \tilde{w}_k$ and observe that $\tilde{w}^T\tilde{\xi} = (\check{w} - 2 \tilde{\alpha}) \tau$. Let us set $s = \min_k \tilde{\theta}_k$ and  $x = \max_k \tilde{\theta}_k - \min_k \tilde{\theta}_k$. Then we have
\begin{align}
F(\tilde{\mbb\theta}, \tilde{\xi}, \tilde{w}) =& \frac{1}{2} \tilde{\alpha} \{(2\tilde{\alpha}-1 - \check{w})\tau - s \}^2 + \frac{1}{2} ( \check{w} - \tilde{\alpha}) ( (2\tilde{\alpha} + 1 - \check{w}) \tau - s - x)^2 \notag  \\ 
	&+ \rho ( x ) - \frac{2}{\check{w}} \tilde{\alpha}(\check{w} - \tilde{\alpha})\tau^2 \notag \\
	=& \frac{1}{2\check{w}} \tilde{\alpha} ( \check{w} - \tilde{\alpha}) ( 2 \tau - x )^2 + \rho(x) - \frac{2}{\check{w}} \tilde{\alpha}(\check{w} - \tilde{\alpha})\tau^2 \notag \\
	=& \frac{\check{w}}{8}(2\tau - x)^2 + \rho(x) - \frac{1}{2}\tau^2 .
	\label{eq:ncfinalobjective}
\end{align}
In the second line above, we have solved for $s$ to find that
\begin{align*}
	s = \frac{1}{\check{w}} \{\tau ( 1 - \check{w} )( \check{w} - 2 \tilde{\alpha} ) + ( \tilde{\alpha} - \check{w}) x  \}.
\end{align*}
In the third line above, we have solved for $\tilde{\alpha}$ to obtain $\tilde{\alpha} = \check{w}/2$ and hence $\tilde \alpha  ( \check{w} - \tilde \alpha ) / \check{w} = \check{w} /4$. These follow from optimality of $\tilde{\mbb\theta}$ and $\tilde{w}$ respectively.
The result follows from applying Lemma~\ref{thm:quadoptlem}, setting $\kappa = \check{w}/4$.
\end{proof}

\subsection{Proof of Theorem~\ref{thm:multivaroracle}} \label{sec:multiproof}
\begin{proof}[\unskip\nopunct]
We begin by defining $P^0$ to be the orthogonal projection onto the linear space 
\begin{align*}
	V_0 = \left\{ \mu + \sum_{j=1}^{j} \sum_{k=1}^{K_j} \ind_{\{ X_{ij} = k \} } \theta_{jk} : (\mu, \theta) \in \R \times \Theta_0 \right\}.
\end{align*}
The residuals from the oracle least-squares fit are $(I - P^0)Y = (I - P^0)\varepsilon$.
The partial residuals $R^{(j)}$ as defined in \eqref{eq:firstpartialresiduals} for the $j$th variable are therefore
\begin{align}
R^{(j)}_i = \sum_{k=1}^{K_j} \ind_{\{ X_{ij} = k \}} \hat{\theta}^0_{jk} + \left[(I - P^0) \varepsilon\right]_i. \label{eq:thm6_univarmodel}
\end{align}

For $j = 1, \ldots, p$, we define $\bar{R}^{(j)}_k = \sum_{i=1}^n \ind_{\{X_{ij} = k \}}R^{(j)}_i / n_{jk}$ for $k = 1, \ldots, K_j$, reordering the labels such that $\bar{R}^{(j)}_1 \leq \cdots \leq \bar{R}^{(j)}_{K_j}$. We then aim to apply the arguments of Theorem~\ref{thm:univarglobal} to $\hat{\mbb\theta}_j$ defined by
\begin{align}
	\hat{\mbb\theta}_j \in \argmin_{\mbb\theta_j \in \Theta_j} \frac{1}{2} \sum_{k=1}^{K_j} w_{jk} \left( \bar{R}^{(j)}_k - \theta_{jk} \right)^2 + \sum_{k=1}^{K_j - 1} \rho( \theta_{j k+1} - \theta_{jk}). \label{eq:multivarproof_innerobjective}
\end{align}
In order to do this, we define the events (for some $\tau_j$ to be determined later):
\begin{align}
\Lambda^{(1)}_j &= \left\{ | \hat{\theta}^0_{jk_l} - \theta^0_{jk_l} | \leq \tau_j : l = 1, \ldots, s_j \right\} \label{eq:multivar_lambda1} \\
\Lambda^{(2)}_{jk} &=  \left\{   \left\vert \frac{1}{n_{jk}} \sum_{i=1}^n  \ind_{\lbrace X_{ij} = k \rbrace}  ((I - P^0) \varepsilon)_i \right\vert < \frac{1}{2} \sqrt{ \eta {\gamma_*}_j s_j} \lambda_j  \right\}  .\label{eq:multivar_lambda2}
\end{align}

On the intersection of events $\cap_{k=1}^{K_j} \Lambda^{(2)}_{jk}$, we have that $| \bar{R}^{(j)}_k - \hat{\theta}^0_{jk} | < \sqrt{ \eta {\gamma_*}_j s_j} \lambda_j  / 2$. By following an identical approach to that involved in computing \eqref{eq:univartailboundcompute}, we have that
\begin{align*}
	\mathbb{P} \left( \cap_{k=1}^{K_j} \Lambda^{(2)}_{jk} \right) &\geq 1 - 2 \exp\left( - \frac{ n w_{j,\text{min}}\eta {\gamma_*}_j s_j \lambda_j^2}{8 \sigma^2} + \log(K_j) \right),
\end{align*}
where we recall that $w_{jk} = n_{jk} / n$.

We now turn our attention to the event $\Lambda^{(1)}_j$. Note that if $s_j = 1$, then this is immediately satisfied since $\hat{\mbb\theta}^0_j = \mbb\theta^0_j = 0$. If $s_j > 1$, we use that the oracle least squares estimate  $ \hat{\mbb\theta}^0  = A Y$ is a linear transformation $A$ of the responses $(Y_i)_{i=1}^n$. For each $i = 1, \ldots, n$, $Y_i$ has an independent (non-central) sub-Gaussian distribution with parameter $\sigma$.
Therefore for each $k = 1, \ldots, K_j$, $\hat{\theta}^0_{jk} - \theta^0_{jk}$ also has a sub-Gaussian distribution, with parameter at most $\sigma  c_\text{min}^{-1/2}$ (recalling that $c_\text{min} = ( \max_l (AA^T)_{ll})^{-1}$). This enables us to show that
\begin{align*}
	\mathbb{P}\left( \Lambda^{(1)}_j \right) &\geq 1 - 2 \exp \left( - \frac{ c_\text{min} \tau_j^2}{2 \sigma^2 }  + \log( s_j)\right).
\end{align*}

We can now set $\tau_j = \sqrt{ \eta {\gamma_*}_j s_j} \lambda_j / 2$. From \eqref{eq:multivarsigsep} and the triangle inequality, on the event $\Lambda^{(1)}_j$ we have that
\begin{align*}
	\Delta( \hat{\mbb\theta}^0_j) &\geq \Delta( \mbb\theta^0_j ) - \sqrt{ \eta {\gamma_*}_j s_j} \lambda_j \\
	&\geq 3 \left(1 + \frac{\sqrt{2}}{\eta} \right) \sqrt{\gamma_j \gamma_j^*} \lambda_j .
\end{align*}
Thus, on the intersection of events $\Lambda^{(1)}_j \cap \left(  \cap_{k=1}^{K_j} \Lambda^{(2)}_{jk} \right)$, we can proceed as in the proof of Theorem~\ref{thm:univarglobal} from \eqref{eq:univarproof_innerobjective},
to conclude that $\hat{\mbb\theta}_j = \hat{\mbb\theta}^0_j$.

It immediately follows that on the intersection of events $\cap_{j=1}^p \left( \Lambda^{(1)}_j \cap \left(  \cap_{k=1}^{K_j} \Lambda^{(2)}_{jk}  \right) \right)$, we have $\hat{\mbb\theta} = \hat{\mbb\theta}^0$. By a union bound, this occurs with probability at least
\begin{align*}
	\mathbb{P} \left( \cap_{j=1}^p \left( \Lambda^{(1)}_j \cap \left(  \cap_{k=1}^{K_j} \Lambda^{(2)}_{jk} \right) \right) \right) &\geq 1 - 2\sum_{j=1}^p \Bigg[ \exp\left( - \frac{ n_{j,\text{min}}\eta {\gamma_*}_j s_j \lambda_j^2}{8 \sigma^2} + \log(K_j) \right)\\
	& \hspace{2cm} +  \exp \left( - \frac{c_\text{min} \eta {\gamma_*}_j s_j \lambda_j^2}{8 \sigma^2 }  + \log( s_j)\right) \Bigg] \\
	&\geq 1 - 4 \sum_{j=1}^p \exp \left( - \frac{(n_{j, \text{min}} \wedge c_{\text{min}}) \eta {\gamma_*}_j s_j \lambda_j^2}{8 \sigma^2} + \log(K_j) \right),
\end{align*}
where in the final line we use $s_j \leq K_j$. 
\end{proof}

\section{Additional experimental information} \label{sec:extraexp}
\subsection{Details of methods} \label{sec:methoddetails}
\subsubsection*{Tree-based methods}
We used the implementation of the random forest procedure \citep{breiman2001random} in the R package \texttt{randomForest} \citep{randomForest} with default settings. CART \citep{breiman1984classification} was implemented in the R package \texttt{rpart} \citep{rpart}, with pruning according to the 1-SE rule (as described in the package documentation).

\subsubsection*{CAS-ANOVA}

The CAS-ANOVA estimator $\hat{\mbb\theta}^{\text{cas}}$ optimises over $(\mu, \mbb\theta)$ a sum of a squared loss term \eqref{eq:least_squares} and an all-pairs penalty term \eqref{eq:all_pairs}. 
In particular, \cite{bondell2009simultaneous} consider two regimes of weight vectors $w$. 
The first is not data-dependent and sets $w_{j,k_1 k_2} = (K_j + 1)^{-1} \sqrt{n_{j k_1} + n_{j k_2} }$.
The second, `adaptive CAS-ANOVA', uses the ordinary least squares estimate for $\mbb\theta$ to scale the weights.
Here, $w_{j,k_1 k_2 } =  (K_j + 1)^{-1} \sqrt{n_{j k_1} + n_{j k_2} } | \hat{\theta}^{\text{OLS}}_{j k_1} - \hat{\theta}^{\text{OLS}}_{j k_2}|^{-1} $. 

Here we introduce a new variant of adaptive CAS-ANOVA, following ideas in \cite{buhlmann2011statistics} for a 2-stage adaptive Lasso procedure.
Instead of using the ordinary least squares estimate $\hat{\mbb\theta}^{\text{OLS}}$ in the above expression, an initial (standard) CAS-ANOVA estimate is used to scale the weights, with $\lambda$ selected for the initial estimate by 5-fold cross-validation.
In simulations, this outperformed the adaptive CAS-ANOVA estimate using ordinary least squares initial estimates so in the interests of time and computational resources this was omitted from the simulation study.
Henceforth adaptive CAS-ANOVA will refer to this 2-stage procedure.

The authors describe the optimisation of $\hat{\mbb\theta}^{\text{cas}}$ as a quadratic programming problem, which was solved using the R package \texttt{rosqp} \citep{rosqp}. 
Here we used our own implementation of the quadratic programming approach described by the authors. We found it considerably faster than the code available from the authors' website, and uses ADMM-based optimisation \citep{boyd2011distributed} tools not available at the time of its publication. We also found, as discussed in Section 5.1 of \cite{maj2015delete}, that we could not achieve the best results using the publicly available code.
Lastly, using our own implementation allowed us to explore a modification of CAS-ANOVA using the more modern approach of adaptive weights via a 2-stage procedure \citep{buhlmann2011statistics} to compare SCOPE to a wider class of all-pairs penalty procedures.

For large categorical variables, solutions are slow to compute and consume large amounts of memory.
In the case of binary response, CAS-ANOVA models were fitted iterating a locally quadratic approximation to the loss function.

\subsubsection*{DMR}
The DMR algorithm \citep{maj2015delete} is implemented in the R package \texttt{DMRnet} \citep{DMRnet}. The degrees of freedom in the model is decided by 5-fold cross-validation.
It is based on pruning variables using the Group Lasso \citep{yuan2006model} to obtain at a low-dimensional model, then performing backwards selection based on ranking $t$-statistics for hypotheses corresponding to each fusion between levels in categorical variables.

The cross-validation routine appeared to error when all levels of all categorical variables were not present in one of the folds. In Section~\ref{sec:adultds1}, cross-validation was therefore not possible so model selection was performed based on Generalized Information Criterion (GIC) \citep{zheng1995consistent}. In all other examples, models were selected via 5-fold cross-validation.

\subsubsection*{Bayesian effect fusion}
In Section~\ref{sec:lowdimsims} we include Bayesian effect fusion \citep{pauger2019bayesian}, implemented in the R package \texttt{effectFusion} \citep{effectFusion}. 
Coefficients within each categorical variable were modelled with a sparse Gaussian mixture model.
The posterior mean was estimated with 1000 samples.

\subsubsection*{Lasso}
In Section~\ref{sec:highdimsims} we also include Lasso \citep{tibshirani1996regression} fits, to serve as a reference point.
Of course, this is unsuitable for models where levels in categorical variables should be clustered together, but the advanced development of the well-known R package \texttt{glmnet} \citep{friedman2010regularization} nevertheless sees its use in practice.

In order to make the fit symmetric across the categories within each variable, models were fitted with an unpenalised intercept and featuring dummy variables for all of the categories within each variable. This is instead of the corner-point dummy variable encoding of factor variables that is commonly used when fitting linear models. Models are fitted and cross-validated with \texttt{cv.glmnet} using the default settings.

\subsubsection*{SCOPE}
For SCOPE, we have provided the R package \texttt{CatReg} \citep{CatReg}. The univariate update step (see Section~\ref{sec:univariate}) is implemented in C++ using Rcpp \citep{rcpp}, with models fitted using a wrapper in R. For the binary response case, the outer loop to iterate the local quadratic approximations in the proximal Newton algorithm are done within R.
In the future, performance could be improved by iterating the univariate update step (and the local quadratic approximations, as in Sections~\ref{sec:adultds1} and \ref{sec:adultds2}) within some lower-level language.
In higher-dimensional experiments, SCOPE was slowed by cycling through all the variables; an active-set approach to this could make it faster still.

\subsection{Further details of numerical experiments} \label{sec:furtherexpdet}
For the experiments in Section~\ref{sec:simulations}, we define the signal-to-noise ratio (SNR) as $\sigma_S / \sigma$, where $\sigma_S$ is the standard deviation of the signal $Y - \varepsilon$, and $\sigma$ is the standard deviation of the noise $\varepsilon$.

\subsubsection{Low-dimensional simulations} \label{sec:lowdimsims2}
In Table~\ref{fig:lowdimdimtime} we include details of computation time and dimension of the fitted models. Figure~\ref{fig:lowdimsimsboxplot} visualises the results also summarised in Table~\ref{tab:simmspetable} in the main paper.

\begin{figure}
\centering
\begin{minipage}[b]{0.95\textwidth} 
\includegraphics[width=\textwidth]{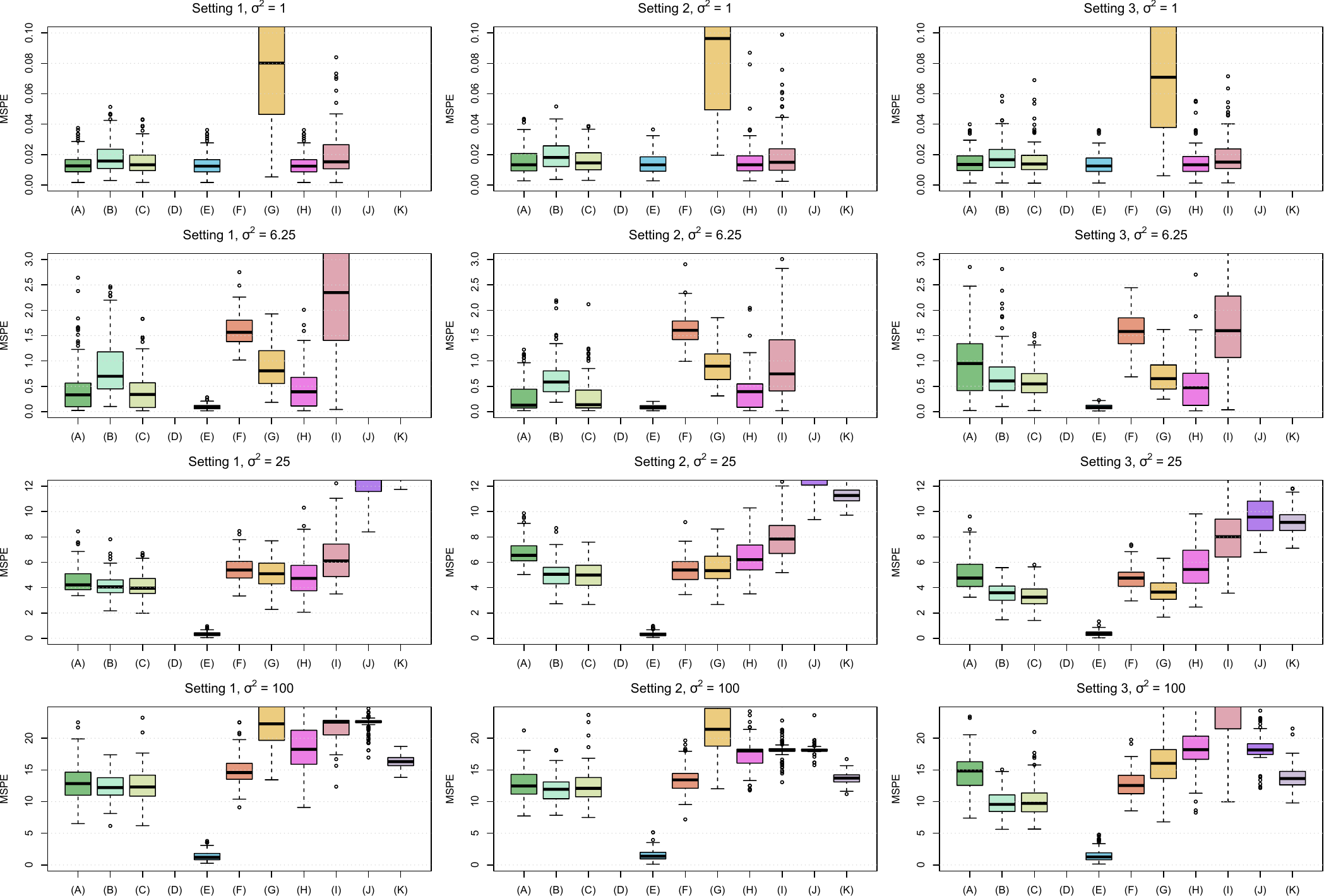}
\end{minipage}
\caption{Prediction performance of various methods: (A) SCOPE-8; (B) SCOPE-32; (C) SCOPE-CV; (D) Linear regression; (E) Oracle least squares; (F) CAS-ANOVA; (G) Adaptive CAS-ANOVA; (H) DMR; (I) BEF; (J) CART; (K) RF. Note that some `boxes' are not visible in some of the plots; this is due to the MSPE in the tests being beyond the range of the plot.}	
 \label{fig:lowdimsimsboxplot}
\end{figure}

\begin{table}[h!]
	\centering
	
	\footnotesize
	\begin{tabular}{r S[table-format=3.1] S[table-format=3.1] S[table-format=3.1] S[table-format=3.1] c S[table-format=4.2]}
		&  \multicolumn{4}{c}{Mean fitted model dimension} && \\
		$\sigma^2$: & \multicolumn{1}{c}{1} & \multicolumn{1}{c}{6.25} & \multicolumn{1}{c}{25} & \multicolumn{1}{c}{100} & & \multicolumn{1}{c}{Mean computation time (s)} \\ \hline \hline
		SCOPE-8  &7.2 & 8.5 & 4.7 & 4.3 & & 16 \\ \hline
		SCOPE-32  & 9.6 & 12.6 & 13.2 & 9.8 & & 48 \\ \hline
		SCOPE-CV & 7.9 & 10.3 & 16.8 & 10.9 & & 68  \\ \hline
		Oracle least squares  & 7.0 & 7.0 & 7.0 & 7.0 & & 0.00 \\ \hline 
		Linear regression  & 231.0 & 231.0 & 231.0 & 231.0 & &  0.01 \\ \hline
		CAS-ANOVA  & 35.2 & 70.0 & 74.3 & 52.4 & & 4679 \\ \hline
		Adaptive CAS-ANOVA  & 13.4 & 31.3 & 36.9 & 32.5 & & 9659 \\ \hline
		DMR & 7.0 & 7.2 & 5.3 & 2.7 & & 21 \\ \hline
		BEF  & 7.3 & 6.3 & 4.1 & 2.0 &  & 975 \\ \hline
		CART &  & & & && 0.01  \\ \hline
		RF  &  & & & &  & 0.66
	\end{tabular}
	\caption{Mean fitted model dimension and computation time for the various methods.}
	\label{fig:lowdimdimtime}
\end{table}

\FloatBarrier
\subsubsection{High-dimensional simulations} \label{sec:highdimsims2}
Here we include additional results relating to the high-dimensional experiments. Figure~\ref{fig:highdimsimsboxplot} visualises the results in Table~\ref{tab:hdsimmspetable} of the main paper.
\begin{figure}
\centering
\begin{minipage}[b]{0.95\textwidth} 
\includegraphics[width=\textwidth]{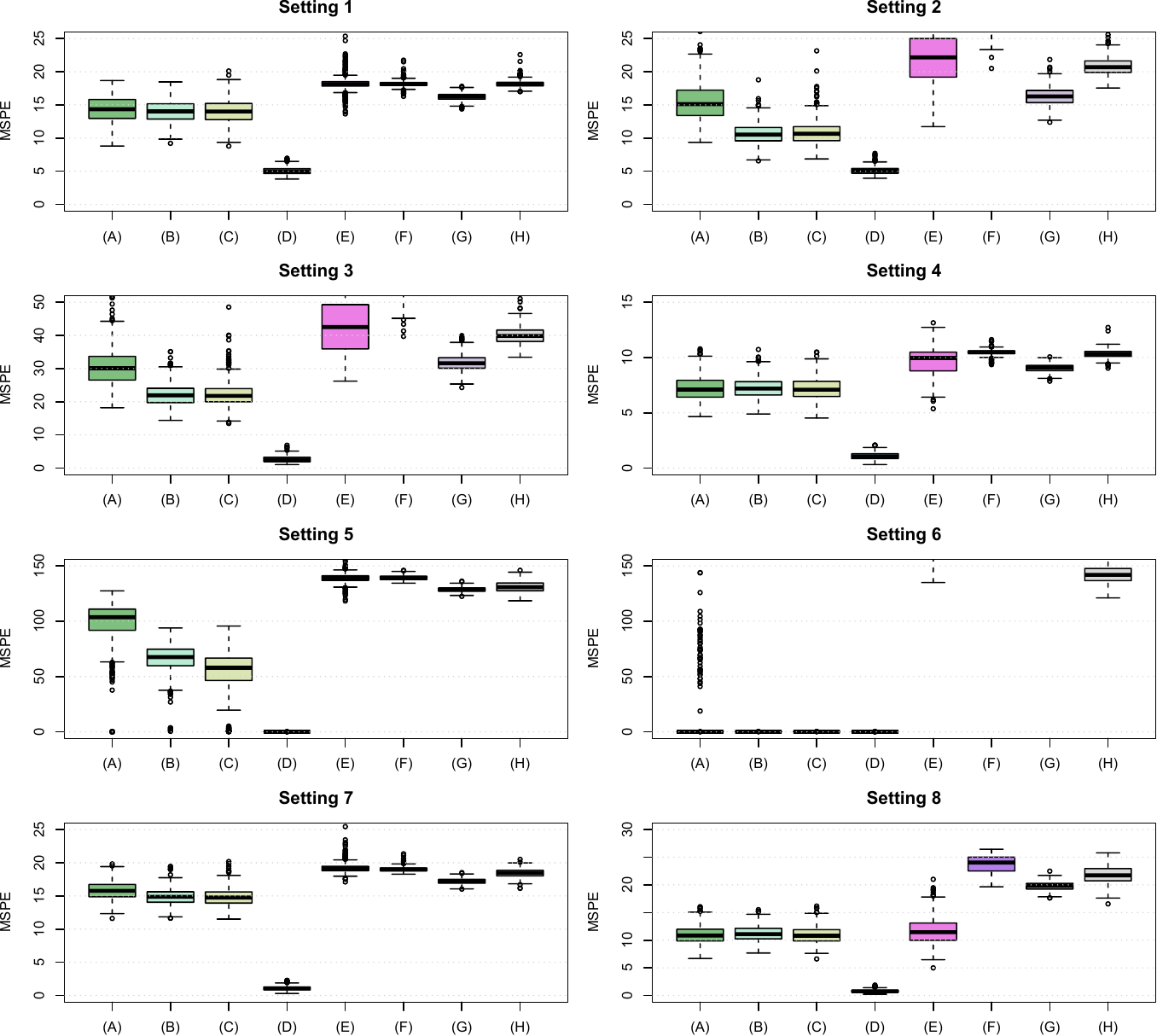}
\end{minipage}
\caption{Prediction performance of various methods: (A) SCOPE-8; (B) SCOPE-32; (C) SCOPE-CV; (D) Oracle least squares; (E) DMR; (F) CART; (G) RF; (H) Lasso. Note that some `boxes' are not visible in some of the plots; this is due to the MSPE in the tests being beyond the range of the plot.}	
 \label{fig:highdimsimsboxplot}
\end{figure}

\begin{table}[h!]
	\centering
	\footnotesize
	
	\begin{tabular}{ r S[table-format=4.1] S[table-format=4.1] S[table-format=4.1] S[table-format=4.1] S[table-format=4.1] S[table-format=4.1] S[table-format=4.1] S[table-format=4.1]}
		Setting: & \multicolumn{1}{c}{1} & \multicolumn{1}{c}{2} & \multicolumn{1}{c}{3} & \multicolumn{1}{c}{4} & \multicolumn{1}{c}{5} &  \multicolumn{1}{c}{6} & \multicolumn{1}{c}{7} & \multicolumn{1}{c}{8}   \\ \hline \hline
		SCOPE-8  & 224 & 322 & 348 & 76 & 234 & 518 & 209 & 175 \\ \hline
		SCOPE-32  &134 & 341 & 502 & 51 & 283 & 650 & 113 & 161   \\ \hline
		SCOPE-CV &  951 & 1739 & 2450 & 332 & 1516 & 2892 & 767 & 902\\ \hline
		DMR & 26 & 38 & 39 & 26 & 30 & 36 & 30 & 29 \\ \hline
		CART & 0.1 & 0.1 & 0.1 & 0.0 & 0.1 & 0.1 & 0.1 & 0.1  \\ \hline
		RF & 5.7 & 5.7 & 5.9 & 2.7 & 5.8 & 5.8 & 5.9 & 5.8 \\ \hline
		Lasso & 1.5 & 1.5 & 1.6 & 1.2 & 1.4 & 1.5 & 1.5 & 1.5 
	\end{tabular}
	\caption{Mean computation time (s)}
	\label{fig:highdimtime}
\end{table}

\begin{table}[h!]
	\centering
	\footnotesize
	
	\begin{tabular}{r S[table-format=3.1] S[table-format=3.1] S[table-format=3.1] S[table-format=3.1] S[table-format=3.1] S[table-format=3.1] S[table-format=3.1] S[table-format=3.1]}
		Setting: & \multicolumn{1}{c}{1} & \multicolumn{1}{c}{2} & \multicolumn{1}{c}{3} & \multicolumn{1}{c}{4} & \multicolumn{1}{c}{5} &  \multicolumn{1}{c}{6} & \multicolumn{1}{c}{7} & \multicolumn{1}{c}{8}   \\ \hline \hline
		SCOPE-8  & 6.9 & 9.4 & 9.8 & 6.9 & 21.3 & 27.1 & 9.3 & 7.2 \\ \hline
		SCOPE-32  &20.7 & 37.5 & 38.0 & 19.9 & 75.8 & 26.1 & 32.9 & 31.3   \\ \hline
		SCOPE-CV &  21.4 & 40.4 & 40.8 & 19.5 & 103.7 & 26.2 & 36.6 & 17.9\\ \hline
		DMR & 1.9 & 4.9 & 4.7 & 3.4 & 3.7 & 22.8 & 2.3 & 7.5 \\ \hline
		Lasso & 15.7 & 167.1 & 152.0 & 32.7 & 143.7 & 469.7 & 35.8 & 82.8
	\end{tabular}
	\caption{Mean fitted model dimension}
	\label{fig:highdimdim}
\end{table}

\begin{table}[h!]
	\centering
	\footnotesize
	\begin{tabular}{r r S[table-format=1.3] S[table-format=1.3] S[table-format=1.3] S[table-format=1.3] S[table-format=1.3]}
		Setting &  $\gamma$: & \multicolumn{1}{c}{4} & \multicolumn{1}{c}{8} & \multicolumn{1}{c}{16} & \multicolumn{1}{c}{32} & \multicolumn{1}{c}{64}    \\ \hline \hline
		1 && 0.028 & 0.290 & 0.196 & 0.138 & $\mathbf{0.348}$  \\ \hline
		2 && 0.002 & 0.016 & 0.234 & 0.298 & $\mathbf{0.450}$ \\ \hline
		3 && 0.006 & 0.012 & 0.286 & 0.248 & $\mathbf{0.448}$ \\ \hline
		4 && 0.030 & $\mathbf{0.356}$ & 0.244 & 0.100 & 0.270 \\ \hline
		5 && 0.000 & 0.000 & 0.026 & 0.070 & $\mathbf{0.904}$ \\ \hline
		6 && 0.000 & 0.000 & 0.464 & $\mathbf{0.534}$ & 0.002 \\ \hline
		7 && 0.006 & 0.092 & 0.234 & 0.144 & $\mathbf{0.524}$ \\ \hline
		8 && 0.264 & $\mathbf{0.446}$ & 0.102 & 0.018 & 0.170 
	\end{tabular}
	\caption{Proposition of times each $\gamma$ was selected by cross-validation.}
	\label{fig:highdimgamma}
\end{table}
\FloatBarrier

\bibliographystyle{abbrvnat}

\end{cbunit}

\end{document}